%% file: main.tex
\documentclass[prb,twocolumn,aps,groupedaddress]{revtex4-2}

\usepackage{amsmath,amssymb,amsthm,mathtools,tikz,braket,framed}
\usetikzlibrary{cd}

\usepackage[bookmarksnumbered=true]{hyperref}
\usepackage{cleveref}

\theoremstyle{definition}\newtheorem{defi}{Definition}[section]
\theoremstyle{plain}\newtheorem{prop}[defi]{Proposition}
\theoremstyle{plain}\newtheorem{thm}[defi]{Theorem}
\theoremstyle{plain}\newtheorem{cor}[defi]{Corollary}
\theoremstyle{plain}\newtheorem{lem}[defi]{Lemma}
\theoremstyle{definition}\newtheorem{eg}[defi]{Example}
\theoremstyle{definition}
\theoremstyle{remark}\newtheorem{rk}[defi]{Remark}

\crefname{eg}{Example}{Examples}
\crefname{defi}{Definition}{Definitions}
\crefname{prop}{Proposition}{Propositions}
\crefname{thm}{Theorem}{Theorems}
\crefname{cor}{Corollary}{Corollaries}
\crefname{lem}{Lemma}{Lemmas}
\crefname{rk}{Remark}{Remarks}
\crefname{table}{Table}{Tables}
\crefname{equation}{Eq.}{Eqs.}
\Crefname{equation}{Eq.}{Eqs.}
\crefname{section}{Sec.}{Secs.}
\Crefname{section}{Sec.}{Secs.}
\crefname{subsection}{Sec.}{Secs.}
\Crefname{subsection}{Sec.}{Secs.}
\crefname{subsubsection}{Sec.}{Secs.}
\Crefname{subsubsection}{Sec.}{Secs.}
\crefname{appendix}{Appendix}{Appendixes}
\Crefname{appendix}{Appendix}{Appendixes}

\newcommand{\clspace}{K}
\newcommand{\symalg}{A_\mathrm{sym}}
\newcommand{\ktheory}{\mathcal{K}}
\newcommand{\Gr}{\mathrm{Gr}}

\begin{document}

\title{Algebra of free fermions: Classifying spaces, Hamiltonians, and computation}
\author{Tian Yuan}
\email[]{tyuan19@fudan.edu.cn}
\affiliation{State Key Laboratory of Surface Physics and Department of Physics, Fudan University, Shanghai 200433, China}
\author{Yang Qi}
\email[]{qiyang@fudan.edu.cn}
\affiliation{State Key Laboratory of Surface Physics and Department of Physics, Fudan University, Shanghai 200433, China}

\date{\today}

\begin{abstract}
    Research on topological phases of matter is a core field in modern condensed matter physics. Free fermion systems, such as topological insulators and superconductors, have been studied using the ``tenfold way'' and K-theory. Building on Kitaev's idea of $\Omega$-spectrum and classifying space, as well as Freed--Moore's K-theory, this work demonstrates that free fermionic systems form a genuine $G$-$\Omega$-spectrum and clarifies its connection to several distinct classification schemes appearing in the physical literature. By introducing the $\mathbb{Z}_2$-graded algebra $\symalg^V$, the classification problem for systems with general symmetries, including antilinear symmetries, antisymmetries, projective representations, and point group symmetries, is turned into an extension problem in representation theory. To solve this, a computational method for the $\mathbb{Z}_2$-graded Wedderburn--Artin decomposition of $\symalg^V$ is developed. This decomposition not only yields a classification but also enables the explicit construction of the corresponding Dirac Hamiltonian. Furthermore, a GAP programming package has been developed to automate these calculations.
\end{abstract}

\maketitle

\input{1.introduction.tex}
\input{2.algebra_of_free_fermion.tex}
\input{3.graded_algebra.tex}
\input{4.decomp_of_AV.tex}
\input{5.conclusion.tex}

\begin{acknowledgments}
  The authors thank Chen Fang for insightful discussions.
  Y.Q. acknowledges support from National Key R\&D Program of China (Grant No. 2022YFA1403402), from the National Natural Science Foundation of China (Grant No. 12174068), and from the Science and Technology Commission of Shanghai Municipality (Grants No. 24LZ1400100 and No. 23JC1400600).
\end{acknowledgments}

\appendix

\input{a1.some_linear_algebra.tex}
\input{a2.technical_graded.tex}

\input{a3.group_algebra.tex}
\input{a4.graded_group_algebra.tex}
\input{a5.K-theory.tex}
\input{a6.gap_example.tex}

\bibliography{refs}

\end{document}

%% file: 1.introduction.tex
\section{Introduction}\label{section:introduction}

Topological phases of matter play a central role in modern condensed matter physics. Among these, gapped free fermionic systems encompassing topological insulators and superconductors
\cite{bernevigQuantumSpinHall2006,kaneQuantumSpinHall2005,hasanColloquiumTopologicalInsulators2010,qiTopologicalInsulatorsSuperconductors2011,kane2TopologicalOrder2005,fuTopologicalInsulatorsInversion2007,konigQuantumSpinHall2007,zhangTopologicalInsulatorsBi2Se32009,kitaevUnpairedMajoranaFermions2001,fuSuperconductingProximityEffect2008,ivanovNonAbelianStatisticsHalfQuantum2001,lutchynMajoranaFermionsTopological2010,oregHelicalLiquidsMajorana2010,readPairedStatesFermions2000}
constitute a fundamental and rigorously understood class. With the theoretical development of topological insulators and superconductors, a systematic topological classification of free fermionic systems has become an urgent necessity. Altland and Zirnbauer proposed the well-known tenfold symmetry classes \cite{altlandNonstandardSymmetryClasses1997,heinznerSymmetryClassesDisordered2005}. Subsequently, the ``tenfold way'' was connected to free fermionic systems via time-reversal ($T$), particle-hole ($C$), and chiral ($P$) symmetries \cite{schnyderClassificationTopologicalInsulators2008}, leading to the establishment of the well-known topological periodic table. Kitaev \cite{kitaevPeriodicTableTopological2009} synthesized these phenomena using Clifford algebras and K-theory, thereby introducing homotopy theory---a cornerstone of modern mathematics---into the study of topological phases of matter.

Despite the tremendous success of the existing periodic table, it becomes inadequate when more complex symmetries---such as crystalline symmetries \cite{fuTopologicalCrystallineInsulators2011} and magnetic space group symmetries \cite{shiozakiClassificationSurfaceStates2022}---must be taken into account. Building upon Kitaev's work, Freed and Moore generalized the K-theoretic framework to incorporate arbitrary symmetries \cite{freedTwistedEquivariantMatter2013}. Their approach can handle general antilinear symmetries and antisymmetries (i.e., symmetries that anticommute with the Hamiltonian). Furthermore, it allows for projective representations of symmetries (e.g., in spin-$\frac{1}{2}$ systems) and accommodates nontrivial symmetry actions on the spatial lattice.
For other related works, we refer the reader to Refs. \cite{gomiFreedMooreKtheory2021,thiangKTheoreticClassificationTopological2016,kubotaNotesTwistedEquivariant2016,scaglioneComparisonTwoApproaches2024}.

For systems with general symmetries, computing the classification is highly nontrivial. For strong topological insulators/superconductors, the classification can be derived from a generalized Clifford algebra extension problem \cite{cornfeldClassificationCrystallineTopological2019,shiozakiClassificationSurfaceStates2022}. Furthermore, strong topological phases serve as the ``root states'' for computations in more general settings; determining their classification paves the way for calculating the general cases. For a specific class of space groups, the complete classification problem can be reduced to the classification of strong phases \cite{cornfeldTenfoldTopologyCrystals2021}. For general crystalline symmetries, physicists typically employ symmetry-based indicator theory \cite{kruthoffTopologicalClassificationCrystalline2017,poSymmetrybasedIndicatorsBand2017,songQuantitativeMappingsSymmetry2018} (for fragile topology beyond symmetry indicators, see \cite{bouhonGeometricApproachFragile2020}). However, the classification obtained via this method is incomplete. A complete classification requires the application of the Atiyah--Hirzebruch spectral sequence (AHSS) in either real space or momentum space \cite{shiozakiAtiyahHirzebruchSpectralSequence2022,shiozakiAtiyahHirzebruchSpectralSequence2023,onoCompleteCharacterizationTopological2024}.

These previous works lack a method to explicitly construct a model Hamiltonian, and require case-by-case derivation to obtain the topological classification of Dirac Hamiltonians. To solve this problem, we introduce an algorithm to explicitly calculate the $\mathbb{Z}_2$-graded Wedderburn--Artin decomposition of the $\mathbb{Z}_2$-graded algebra $\symalg^V$, which serves as the algebra of Dirac Hamiltonians. As a byproduct of this decomposition, we can obtain both the model Hamiltonians and the topological classification of Dirac Hamiltonians for strong topological insulators/superconductors in the most general settings. This algorithm is implemented as a GAP package \cite{ttiiddeeTtiiddeeFreeFermionicSPT2026}, which is our main computational result and will be heavily used in our upcoming work \cite{boundaryanomaly2026} on the boundary states of topological insulators/superconductors.

Classification represents the lowest-order information of topological phases of matter; however, many physical problems require more sophisticated data. In his pioneering talks \cite{alexeiTopologicalClassificationPhases2011,alexeiClassificationShortRangeEntangled2013,alexeiHomotopytheoreticApproachSPT2015}, Kitaev proposed a more refined approach to characterize invertible topological orders, which encompasses free fermionic systems as a special case. Kitaev's ideas were developed in the subsequent works \cite{gaiottoSymmetryProtectedTopological2019,xiongMinimalistApproachClassification2018}. His core idea is to utilize an $\Omega$-spectrum to study the classifying space of invertible topological orders. An $\Omega$-spectrum is a sequence of spaces $X_n$ indexed by integers, satisfying the homotopy equivalence $\Omega X_{n+1} \simeq X_{n}$. From the perspective of modern homotopy theory \cite{mayGeometryIteratedLoop2006}, an $\Omega$-spectrum is synonymous with an ``Abelian group''\footnote{To be precise, a connective $\Omega$-spectrum.}; that is, it is a homotopy type endowed with a commutative multiplication structure. Physically, $X_n$ represents the space of $n$-dimensional Hamiltonians or quantum states, which we uniformly refer to as the classifying space in this paper. The Abelian group structure on this space physically corresponds to invertible stacking. This perspective makes the validity and generality of Kitaev's ideas self-evident. The difficulty in applying this approach lies in identifying a concrete $\Omega$-spectrum to characterize a specific type of invertible topological order. In homotopy theory, there is a one-to-one correspondence between $\Omega$-spectra and generalized cohomology theories. For many physical systems of interest, their classifying spaces are not directly specified by an $\Omega$-spectrum, but rather by a generalized cohomology theory. Examples include interacting bosonic systems \cite{chenSymmetryProtectedTopological2013,kapustinSymmetryProtectedTopological2014}, interacting fermionic systems \cite{kapustinFermionicSymmetryProtected2015}, invertible stabilizer codes \cite{geikoClassificationInvertibleStabilizer2025}, quantum cellular automata \cite{czajkaAnomaliesLatticeHomotopy2025}, and invertible topological quantum field theories \cite{freedShortrangeEntanglementInvertible2014,freedReflectionPositivityInvertible2021}.

For free fermionic systems, the classifying space corresponds to the space of gapped Hamiltonians. Mathematically, this is precisely the $\Omega$-spectrum associated with the K-theory mentioned earlier. In the presence of symmetries, we need to consider the genuine $G$-$\Omega$-spectrum \cite{cornfeldTenfoldTopologyCrystals2021,mayEquivariantHomotopyCohomology1996}. A genuine $G$-$\Omega$-spectrum is no longer a sequence of spaces indexed by integers; rather, it consists of $G$-spaces $X_V$ indexed by virtual real representations $V$ of the group $G$, satisfying the equivalence $\Omega^V X_{V\oplus W}\simeq X_W$. When $G$ is the trivial group, its real representations are determined solely by their dimensions, and we recover the sequence of spaces indexed by integers. Here, the real representation $V$ (or the virtual real representation $-V$) of $G$ can be interpreted as the real space (or the momentum space). An ordinary $\Omega$-spectrum can be viewed as a mathematical structure that packages systems of all dimensions together, whereas a genuine $G$-$\Omega$-spectrum further incorporates the different actions of symmetries on space.

This paper is organized into two main parts. In the first part, \cref{section:algebra of free fermionic systems}, we study the classifying spaces for gapped free fermionic systems. We first define the algebra $\symalg$ which describes the symmetries of free fermionic systems, as well as the algebra $\symalg^V$ consisting of Dirac Hamiltonians that possess such symmetries. Next, we construct several classifying spaces for free fermionic systems, elucidate the relationships between them, and demonstrate how these spaces constitute a genuine $G$-$\Omega$-spectrum. There are four different but equivalent classifications for free fermionic systems, which are summarized in \cref{eq:diagram of classification}. Here, we only consider the classification of strong phases. Several previously ambiguous concepts in the physical literature are now rigorously clarified by \cref{eq:diagram of classification}. At the end of the section, we review the tenfold way using the $\mathbb{Z}_2$-graded Wedderburn--Artin decomposition of $\symalg^V$, and the extension problems of $\symalg^V$ that lead to an algebraic classification of free fermionic systems.

The second part, comprising \cref{section:Graded Algebra,section:decomposition of AV}, establishes the theoretical foundation for computing the classifications, explicitly constructing Dirac Hamiltonians, and deriving the topological invariants of Dirac Hamiltonians. It also provides the underlying mathematical framework for the GAP software package. In \cref{section:Graded Algebra}, we review the theory of graded semisimple algebras and present the classification of $\mathbb{Z}_2^N$-graded division algebras. Section \ref{section:decomposition of AV} constructs the $\mathbb{Z}_2$-graded Wedderburn--Artin decomposition of $\symalg^V$. This decomposition is the main tool to compute the classifications, as well as to explicitly construct Dirac Hamiltonians and their corresponding topological invariants. At the end of this section, we introduce the GAP software package \cite{ttiiddeeTtiiddeeFreeFermionicSPT2026} and provide several computational examples demonstrating its application.

\subsection{Conventions and notations}\label{subsection:conventions and notations}

We use $\mathbb{K}$ to denote a field; in this work, $\mathbb{K}$ is always $\mathbb{R}$ or $\mathbb{C}$. $Cl^{p,q}$ is the Clifford algebra over $\mathbb{R}$, and $\mathbb{C}l^{p,q}$ or $\mathbb{C}l^{p}$ is the Clifford algebra over $\mathbb{C}$, as defined in \cref{subsection:Clifford Algebra}.

$\symalg=\mathbb{C}[G,\omega,s]$ is a $\mathbb{Z}_2$-graded algebra with grading $h$, as defined in \cref{subsection:0+1-dimensional gapped free fermionic systems}.

A vector space $V$ that has an action of an algebra $A$ is called an $A$ module. The corresponding representation of $A$ on $V$ is the algebra homomorphism $\rho_V: A \to \mathrm{End}(V)$. Module and representation are equivalent concepts in principle; however, we distinguish them by referring to the vector space $V$ as the $A$ module and to the algebra homomorphism $\rho_V$ as the representation of $A$.

When we use the word ``space'', we refer to a topological space or a homotopy type, and more generally a $G$ space or a $G$-homotopy type depending on the context. $\cong$ of spaces stands for (G-)homeomorphism, $\simeq$ of spaces stands for ($G$-)homotopy equivalence. We use $\Omega X$ to denote the loop space of a space $X$ with a basepoint. If $V$ is a $G$-representation over $\mathbb{R}$, $S^V$ is the one-point compactification of $V$ with the point at infinity as the basepoint. Similarly, we use $\Omega^V X$ to denote the mapping space from $S^V$ to $X$ preserving the basepoint.

%% file: 2.algebra_of_free_fermion.tex
\section{Algebras of Free Fermionic Systems}\label{section:algebra of free fermionic systems}

In \cref{subsection:symmetry of the free fermionic systems}, we introduce a $\mathbb{Z}_2$-graded algebra $\symalg=\mathbb{C}[G,\omega,s]$, which will be used to describe the symmetries of free fermionic systems. From the algebra $\symalg$, we can construct the classifying space $\clspace(\symalg)$ for $(0+1)$-d gapped free fermionic systems, as discussed in \cref{subsection:classifying space}, and the classifying space $\Omega^V(\clspace(\symalg))$ for $(q+1)$-d gapped free fermionic systems, as discussed in \cref{subsection:q+1-dimensional gapped free fermionic systems}. In \cref{subsection:dirac hamiltonians}, we construct the algebra $\symalg^V$ of Dirac Hamiltonians, and the classifying space $\clspace(\symalg^V)$ for Dirac Hamiltonians in terms of their mass terms. The classifying spaces $\clspace(\symalg^V)$ and $\Omega^V(\clspace(\symalg))$ have the same $G$-homotopy type; the $G$-homotopy equivalence is given by the construction of Dirac Hamiltonians from the mass terms. That is, we have an equivalence of classifying spaces
\begin{equation}
    \begin{tikzcd}
        \left\{\textrm{Dirac Hamiltonians}\right\}
            \arrow[r, "\simeq"{above} , hook]
        & \left\{\textrm{Hamiltonians}\right\}.
    \end{tikzcd}
\end{equation}
or in math symbols,
\begin{equation}
    \begin{tikzcd}
        \clspace(\symalg^V)
            \arrow[r, "\simeq"{above}, hook]
        & \Omega^V(\clspace(\symalg)).
    \end{tikzcd}
\end{equation}
In \cref{subsubsection:atomic insulator and Bott periodicity}, we show that $\clspace(\symalg^V)$ forms a genuine $G$-$\Omega$-spectrum, which follows from a well-known physical folk-theorem: atomic insulators are systems without nontrivial boundary states. Mathematically, the physical folk-theorem is equivalent to a generalization of $(1,1)$-Bott periodicity for $KR$-theory \cite{atiyahKtheoryReality1966}.

In \cref{subsection:The classification of gapped free fermionic hamiltonians}, we consider the problem of classification. The main result is the following diagram:
\begin{widetext}
    \begin{equation}\label{eq:diagram of classification}
    \begin{tikzcd}
        \left\{\textrm{Classification of Dirac Hamiltonians}\right\}
            \arrow[r, "\cong"{above}, hook]
        & \left\{\textrm{Classification of Hamiltonians}\right\}
            \arrow[d, "\cong"{left}, "\textrm{Solving the Hamiltonian}"{right}] \\
        \left\{\textrm{Algebraic Classification}\right\}
            \arrow[u, "\cong"{left}, "\textrm{Changing the sign of the mass term}"{right}]
            \arrow[r, "\textrm{ABS construction}"{above}, "\cong"{below}]
        & \left\{\textrm{bundle-theoretic classification}\right\},
    \end{tikzcd}
\end{equation}
\end{widetext}
or in math symbols,
\begin{equation}
    \begin{tikzcd}
        \pi_0\left(\clspace(\symalg^V)^G\right)
            \arrow[r, "\cong"{above}, hook]
        & \pi_0\left(\left(\Omega^V\clspace(\symalg)\right)^G\right)
            \arrow[d, no head, "\cong"{right}] \\
        \mathcal{A}(\symalg^V)
            \arrow[u, "\cong"{left}]
            \arrow[r, "\cong"{below}]
        & \tilde{\ktheory}_{\symalg}(S^V).
    \end{tikzcd}
\end{equation}
The two classification groups mentioned above are determined by the connected components of the $G$-fixed point of the classifying spaces, as discussed in \cref{subsubsection:the topological classification}. Since $\clspace(\symalg)$ also serves as the classifying space for vector bundles with symmetry $\symalg$, it yields a bundle-theoretic classification $\tilde{\ktheory}_{\symalg}(S^V)$, as detailed in \cref{subsubsection:bundle-theoretic classification}. The canonical map from the classifying space to the bundle-theoretic classification corresponds to the process of ``solving the Hamiltonian'', which is the right-hand map in the diagram. To determine these classifications, we only need to calculate the classification of the mass terms of the Dirac Hamiltonians. Section \ref{subsubsection:the algebraic classification} introduces an algebraic classification $\mathcal{A}(\symalg^V)$ of these mass terms, which corresponds to the left map, as established in \cref{subsubsection:comparison of two classifications}. The proof that this map constitutes an isomorphism can be found in \cref{subsection:Tenfold Way}. Finally, the lower map represents the classic Atiyah--Bott--Shapiro (ABS) construction, as discussed in \cref{subsubsection:comparison of algebraic classification and abs construction}.

\subsection{Symmetries of free fermionic systems}\label{subsection:symmetry of the free fermionic systems}

In quantum mechanics, the symmetry of a system is often described by group theory and group representation. In free fermionic systems, the symmetry behaves in a more complicated way: we need to consider antilinear symmetries, projective representations, and anticommutation with the Hamiltonian. Thus, it is inconvenient to use only the language of group theory; we need to generalize the group to an algebra.

An algebra is a vector space with an associative and unital multiplication. For each group $G$, there is an algebra $\mathbb{C}[G]$ called the group algebra. $\mathbb{C}[G]$ is spanned by the elements $g\in G$ as a vector space, and the multiplication is given by $g\cdot h=gh$. The representation theory of $G$ is equivalent to the representation theory of $\mathbb{C}[G]$; therefore, we can replace $G$ by $\mathbb{C}[G]$ in the study of symmetries in quantum mechanics.

As a symmetry, $\mathbb{C}[G]$ only describes the linear (or unitary) symmetry of the system. In general, the symmetry can be antilinear (or antiunitary). Furthermore, the multiplication rule $g\cdot h=gh$ is also violated in many cases; for example, the rotation symmetry of a spin-$\frac{1}{2}$ system only satisfies this multiplication rule up to a sign.

To describe these symmetries, we generalize $\mathbb{C}[G]$ to a new algebra $\mathbb{C}[G,\omega,s]$. The group homomorphism $s:G\to \mathbb{Z}_2$ is called a $\mathbb{Z}_2$-grading on $G$. $G_0=s^{-1}(0)$ is the subgroup of linear symmetries, and $G_1=s^{-1}(1)$ is the subset of antilinear symmetries. $\omega$ is a two-cocycle, i.e., $[\omega]\in H^2(G,U(1))$. The coefficient group $U(1)$ of the cohomology has a $G$-action that is given by ${}^ga=a$ for $g\in G_0$, and ${}^ga=\bar{a}$ for $g\in G_1$, where $a\in U(1)$, and $\bar{a}$ is the complex conjugate of $a$. The generators of $\mathbb{C}[G,\omega,s]$ are still $g\in G$, but we have $g\cdot h=\omega(g,h)gh$ for $g,h\in G$, and $gi=-ig$ for $g\in G_1$.

Notice that the algebra $\mathbb{C}[G,\omega,s]$ is not an algebra over $\mathbb{C}$. The multiplication of a $\mathbb{C}$ algebra should be $\mathbb{C}$-bilinear, but we have $gi=-ig$ for $g\in G_1$. Thus, $\mathbb{C}[G,\omega,s]$ is actually an algebra over $\mathbb{R}$: $i\in \mathbb{C}$ is considered as an independent generator of this $\mathbb{R}$ algebra, rather than just a scalar.

\subsection{\texorpdfstring{$(0+1)$-dimensional gapped free fermionic systems}{(0+1)-dimensional gapped free fermionic systems}}\label{subsection:0+1-dimensional gapped free fermionic systems}

A 0+1-d gapped free fermionic system is a vector space with a Hamiltonian and a symmetry acting on it. The requirement of a gap partitions the vector space into two subspaces, of the occupied states and of the unoccupied states. Since we only care about the topology, we should focus on the partition itself but not the specific value of energy. Thus, we assume the energy of all occupied states to be $-1$, and the energy of all unoccupied states to be 1.

Under the assumption above, the Hamiltonian of the gapped system is a linear map $H$ such that $H^2=1$. If the system has symmetry, the vector space should be a representation of the algebra $\symalg \coloneqq \mathbb{C}[G,\omega,s]$. Together with the Hamiltonian $H$, we form a new algebra $Cl^{0,1}\otimes\mathbb{C}[G,\omega,s]$, where $Cl^{0,1}$ is the Clifford algebra that has one generator $\tilde{e}_1$ and satisfies the relation $\tilde{e}_1^2 = 1$.

There are symmetries anticommuting with the Hamiltonian $H$, for example, the chiral symmetry. Therefore, we need to introduce a new grading on $G$ that is given by $h:G\to \mathbb{Z}_2$. With this grading, $\mathbb{C}[G,\omega,s]$ becomes a $\mathbb{Z}_2$-graded algebra. A $\mathbb{Z}_2$-graded algebra is a $\mathbb{Z}_2$-graded vector space equipped with a unital and associative multiplication which preserves the grading. There is a theory of $\mathbb{Z}_2$-graded algebras similar to that of ordinary algebras, as will be discussed in \cref{section:Graded Algebra}. \footnote{Notice that $G$ is actually a $\mathbb{Z}_2\times \mathbb{Z}_2$-graded group with grading $s\times h$, and $\mathbb{C}[G,\omega,s]$ is also a $\mathbb{Z}_2\times \mathbb{Z}_2$-graded algebra. In practical calculations, it is convenient to consider both gradings, but for current conceptual purposes, we only need the grading of $h$.}

Now the algebra of the symmetry and the Hamiltonian becomes $Cl^{0,1}\hat{\otimes}\mathbb{C}[G,\omega,s]$. $\hat{\otimes}$ is the super tensor product of $\mathbb{Z}_2$-graded algebras; the generator $\tilde{e}_1\in Cl^{0,1}$ has grading 1. The ordinary tensor product $A\otimes B$ has the property of $a\otimes 1$ commuting with $1\otimes b$ for $a\in A$ and $b\in B$. The super tensor product $\hat{\otimes}$ has the property of $a\hat{\otimes} 1$ anticommuting with $1\hat{\otimes} b$ for $a\in A_1$ and $b\in B_1$, where $A_1$ and $B_1$ are the grading 1 subspaces of $A$ and $B$, when either $a$ or $b$ has grading 0, $a\hat{\otimes} 1$ commutes with $1\hat{\otimes} b$.

A 0+1-d free fermionic system with the symmetry given by the $\mathbb{Z}_2$-graded algebra $\mathbb{C}[G,\omega,s]$ is a representation of $Cl^{0,1}\hat{\otimes}\mathbb{C}[G,\omega,s]$. $Cl^{0,1}\hat{\otimes}\mathbb{C}[G,\omega,s]$ is a semisimple algebra, which means that any representation can be decomposed into a direct sum of irreducible representations. Moreover, as a consequence of semisimplicity, $Cl^{0,1}\hat{\otimes}\mathbb{C}[G,\omega,s]$ can be decomposed into a product of matrix algebras over division algebras, which is called the Wedderburn--Artin decomposition. Then the representation theory of $Cl^{0,1}\hat{\otimes}\mathbb{C}[G,\omega,s]$ reduces to the representation theory of the matrix algebras over division algebras, which is easy to study.

\subsection{Classifying spaces}\label{subsection:classifying space}

To classify quantum systems, let us imagine a space $\clspace$ of all quantum states (with some properties, for example, 0+1-d gapped free fermionic systems), which is called the classifying space. A point in the classifying space $\clspace$ is a quantum state. Given two points in $\clspace$, a path between them should be considered as an equivalence between two quantum states. Given two paths in $\clspace$, the homotopy between them should be considered as an equivalence between two different equivalences between quantum states. This keeps going to higher homotopies; therefore, the classifying space is a homotopy type.

What is the equivalence of two quantum systems? For different types of quantum systems, there should be different kinds of equivalence. For example, for the gapped free fermionic systems, the classifying space is the space of parameters of the gapped Hamiltonians, which is a topological space, and the equivalences are the paths and higher homotopies in the topological space. The physical meaning of this equivalence is the adiabatic deformation of the Hamiltonian. Another example is gapped interacting bosonic systems, where the equivalences are given by invertible gapped domain walls and by domain walls between domain walls.

To get the classifying space for 0+1-d gapped fermionic systems, let us denote the $\mathbb{Z}_2$-graded algebra of symmetry $\mathbb{C}[G,\omega,s]$ by $\symalg$. $Cl^{0,1}\hat\otimes \symalg$ is the algebra of the symmetry together with the Hamiltonian. There is a subtlety that, to get a meaningful classification, we need to consider an infinite-dimensional quantum system for two reasons: (1) the quantum systems in the real world have infinite-dimensional Hilbert spaces, and (2) it is generally allowed to add new degrees of freedom to the system in the study of topological phases.

In the language of representation theory, the infinite-dimensional system is given by a representation $\rho$ of $Cl^{0,1}\hat{\otimes}\symalg$; the $\symalg$ module corresponding to $\rho$ is denoted by $\mathcal{U}(Cl^{0,1}\hat{\otimes}\symalg)$. $\rho$ is a (countably) infinite-dimensional representation of $Cl^{0,1}\hat{\otimes}\symalg$; if we decompose it into irreducible representations, we require that the direct sum components contain all species of irreducible representations, and that for each irreducible representation there are infinitely many copies. $\mathcal{U}(Cl^{0,1}\hat{\otimes}\symalg)$ is called the universe of $Cl^{0,1}\hat{\otimes}\symalg$.

The Hamiltonian of the system is $\rho(\tilde{e}_1\hat\otimes 1)=H_0$. $H_0$ is chosen to be the trivial system. Other symmetric Hamiltonians of the system also form a representation of $Cl^{0,1}\hat\otimes \symalg$. The action of $\symalg$ is the same as $\rho$, but the Hamiltonian $H$ is different from $H_0$. We require $H$ to be different from $H_0$ only on a finite-dimensional subspace of $\mathcal{U}(Cl^{0,1}\hat{\otimes}\symalg)$. The reason is that we do not allow infinitely many high-energy states to be occupied and infinitely many low-energy states to be unoccupied.

Now we define the classifying space. We first define the classifying space\footnote{
    All classifying spaces considered in this work are topological spaces. To define the topology of $\clspace(\symalg)$, we first take an infinite sequence of finite-dimensional submodules of $\mathcal{U}(Cl^{0,1}\hat{\otimes}\symalg)$. We then consider the topological spaces of linear maps $H$ with $H^2 = 1$ on a finite-dimensional submodule. There is a canonical inclusion map between the spaces of linear maps on two different submodules, and if $V\hookrightarrow W$, the inclusion map is given by direct sum with $H_0|_{V'}$, where $W\cong V\oplus V'$. Finally, we take the direct limit of these topological spaces, which automatically has a topological structure.
}
of (not necessarily symmetric) Hamiltonians.
\begin{framed}
    For $(0+1)$-d gapped free fermionic systems with symmetry $\symalg$, the classifying space for (not necessarily symmetric) Hamiltonians is $\clspace(\symalg)$, which is the space of linear maps on $\mathcal{U}(Cl^{0,1}\hat{\otimes}\symalg)$ such that for $H\in \clspace(\symalg)$, $H^2=1$ and $H$ differs from $H_0$ only on a finite-dimensional subspace.
\end{framed}
Since $\symalg=\mathbb{C}[G,\omega,s]$, $\clspace(\symalg)$ has a $G$-action given by $(-1)^{h(g)}\rho(g)H\rho(g)^{-1}$. Notice that $\rho(g)^{-1}=\omega(g^{-1},g)^{-1}\rho(g^{-1})$. Here, we take the convention that $g^{-1}$ is the inverse as a group element in $G$, and $(g)^{-1}$ is the inverse as an algebra element in $\mathbb{C}[G,\omega,s]$. Therefore, we have $(g)^{-1} = \omega(g^{-1},g)^{-1}g^{-1}$ and $\rho(g)^{-1}=\rho((g)^{-1})$. The $G$-action makes $\clspace(\symalg)$ a $G$-space.
\begin{framed}
    For $(0+1)$-d gapped free fermionic systems with symmetry $\symalg$, the classifying space for symmetric Hamiltonians is $\clspace(\symalg)^G$, which is the $G$-fixed point subspace of $\clspace(\symalg)$, i.e., the subspace satisfying $(-1)^{h(g)}\rho(g)H\rho(g)^{-1}=H$ for $H\in \clspace(\symalg)$.
\end{framed}

\subsection{\texorpdfstring{$(q+1)$-dimensional gapped free fermionic systems}{(q+1)-dimensional gapped free fermionic systems}}\label{subsection:q+1-dimensional gapped free fermionic systems}

In $q+1$ dimension, free fermionic systems can be studied through the Hamiltonians on the Brillouin zone. In general, the Brillouin zone is a $q$-dimensional torus $T^q$. Here, we assume the Brillouin zone to be a $q$-dimensional sphere $S^q$, for the reason that we only consider strong topological insulators or superconductors
\footnote{The Brillouin zone comes from the Fourier transformation of a translation-invariant Hamiltonian. It is not known whether a translation-noninvariant Hamiltonian can be studied by a Hamiltonian on the Brillouin zone (Kitaev recently proposed a rigorous method to classify free fermionic systems \cite{alexeikitaevLocalDefinitionsGapped}, which yields the same results as the tenfold way), but we can give some remarks on this point. The periodic boundary conditions in each direction on the torus come from translation invariance in that direction. If we consider a translation-noninvariant system, at least the periodic boundary conditions should be broken. The spherical Brillouin zone is a way to break the periodic boundary conditions, but it retains part of the ``translation invariance'' in the sense that it does not have $q$ translation-invariant directions, but is translation invariant in a ``uniform'' way. We hope that the Hamiltonians on the spherical Brillouin zone are a good approximation to translation-noninvariant systems whose Hamiltonians vary slowly in real space. Then we could study each small region of the real space by the Hamiltonian on the spherical Brillouin zone. Finally, for any Hamiltonian on the spherical Brillouin zone, we can construct a Hamiltonian on the torus via the map $T^q\to S^q$ that maps the periodic boundary of the torus to one point.}.

Some symmetries have an action on the Brillouin zone, for example, the point group symmetries. In general, we consider a $q$-dimensional orthogonal real representation $\rho_V$ of $G$ on $V$, and form the representation sphere $S^V$ by taking the one-point compactification of $V$. The $G$-action on $V$ makes $S^V$ a $G$-space. The representation sphere $S^V$ is a model of spherical Brillouin zone with $G$-action.

To construct a symmetric Hamiltonian on $S^V$, we only need to construct a $G$-equivariant map $H:S^V\to \clspace(\symalg)$ satisfying 
\begin{equation}\label{eq:symmetry condition}
    (-1)^{h(g)}\rho(g)H(k)\rho(g)^{-1}=H(g\cdot k),
\end{equation}
for $k\in S^V$\footnote{The $G$-action on $S^V$ is the $G$-action on the momentum space. Notice that the antilinear symmetry has an extra space inversion on the real space. Our convention differs from that of \cite{shiozakiClassificationSurfaceStates2022}: they put the extra space inversion on the momentum space.}.
\begin{framed}
    For $(q+1)$-d gapped free fermionic systems with symmetry $\symalg$, the classifying space for (not necessarily symmetric) Hamiltonians on the Brillouin zone $S^V$ is $\Omega^V (\clspace(\symalg))$. The classifying space for symmetric Hamiltonians on the Brillouin zone $S^V$ is the $G$-fixed point subspace $\left(\Omega^V \clspace(\symalg)\right)^G$.
\end{framed}
$\Omega^V (\clspace(\symalg))\coloneq \mathrm{Top}_*(S^V,\clspace(\symalg))$ is the mapping space from $S^V$ to $\clspace(\symalg)$ preserving the basepoint. The basepoint of $S^V$ is the point at infinity, and the basepoint of $\clspace(\symalg)$ is $H_0$. There is a $G$-action on $\Omega^V (\clspace(\symalg))$ that is given by $g\cdot f(g^{-1}\cdot-)$ for $f$ in the mapping space.

We give a remark on why the Hamiltonians should preserve the basepoint, i.e., $H(\infty)=H_0$. We have
\begin{equation}
    \mathrm{Top}(S^V,\clspace(\symalg))\simeq \mathrm{Top}_*(S^V,\clspace(\symalg))\times \clspace(\symalg).
\end{equation}
The homotopy equivalence is given by a calculation using stable homotopy theory,
\begin{equation}
    \begin{aligned}
        \mathrm{Top}(S^V,\clspace(\symalg)) & \simeq \mathrm{Top}_*(S^V_+,\clspace(\symalg)) \\
        & \simeq \mathrm{Sp}(\Sigma^{\infty}S^V_+,\clspace(\symalg)) \\
        & \simeq \mathrm{Sp}(\Sigma^{\infty}S^V\vee \Sigma^{\infty}S^0,\clspace(\symalg))\\
        & \simeq \mathrm{Sp}(\Sigma^{\infty}S^V,\clspace(\symalg)) \\
        & \phantom{{}\simeq{}} \quad \times \mathrm{Sp}(\Sigma^{\infty}S^0,\clspace(\symalg)) \\
        & \simeq \mathrm{Top}_*(S^V,\clspace(\symalg))\\
        & \phantom{{}\simeq{}} \quad \times \mathrm{Top}_*(S^0,\clspace(\symalg)) \\
        & \simeq \mathrm{Top}_*(S^V,\clspace(\symalg))\times \clspace(\symalg).
    \end{aligned}
\end{equation}
The homotopy equivalence shows that given a $(q+1)$-d Hamiltonian with $H(\infty)=H_0$ and a $(0+1)$-d Hamiltonian $H'$, we can construct a new $(q+1)$-d Hamiltonian with $H(\infty)=H'$. The choice of $H(\infty)$ corresponds to stacking the Hamiltonian with a trivial band, which is considered as an equivalence relation.

\subsection{Dirac Hamiltonians}\label{subsection:dirac hamiltonians}

\subsubsection{Definition of Dirac Hamiltonians}\label{subsubsection:definition of dirac hamiltonians}

There is a special class of $(q+1)$-d Hamiltonians that are called Dirac Hamiltonians. A Dirac Hamiltonian on $S^V$ is given by
\begin{equation}\label{eq:Dirac Hamiltonian}
    H(k)=
    \begin{cases}
        -k_0M+\sum_{i=1}^q k_i\tilde{\Gamma}_i, &k\in S_+^V\\
        k_0M_0+\sum_{i=1}^q k_i\tilde{\Gamma}_i, &k\in S_-^V,
    \end{cases}
\end{equation}
where we embed $S^V$ into a $(q+1)$-dimensional vector space $V\oplus \mathbb{R}$ as the unit sphere; hence, we have $\sum_{i=0}^{q}k^2_i=1$, and $S_+^V$ and $S_-^V$ are the hemispheres with $k_0\geq0$ and $k_0\leq0$, respectively. $M$, $M_0$, and $\tilde{\Gamma}_i$ are operators that square to 1. $M$ and $\tilde{\Gamma}_i$ anticommute with each other, and $M_0$ and $\tilde{\Gamma}_i$ anticommute with each other.

In order to satisfy the symmetry condition \cref{eq:symmetry condition}, we require
\begin{align}
        (-1)^{h(g)}\rho(g)M\rho(g)^{-1} & =M, \\
        (-1)^{h(g)}\rho(g)M_0\rho(g)^{-1} & =M_0, \\
        (-1)^{h(g)}\rho(g)\tilde\Gamma_i\rho(g)^{-1} & =\sum_j\rho_V(g)_{ji}\tilde\Gamma_j.\label{eq:G-action on Gamma}
\end{align}

\subsubsection{Classifying spaces for Dirac Hamiltonians}\label{subsubsection:classifying space of dirac hamiltonians}

Motivated by Dirac Hamiltonians, we construct a new $\mathbb{Z}_2$-graded algebra $\symalg^V$. $\symalg^V$ is generated by generators of $\symalg$ and $q$ additional generators $\tilde{e}_i$, where $\tilde{e}_i$ satisfy the same relations as $\tilde\Gamma_i$ in Dirac Hamiltonians, hence generating a Clifford algebra, and they also satisfy \cref{eq:G-action on Gamma}, i.e.,
\begin{equation}
    (-1)^{h(g)}g\cdot \tilde{e}_i \cdot (g)^{-1} = \sum_j\rho_V(g)_{ji}\tilde{e}_j = \rho_V(g)(\tilde{e}_i).
\end{equation}
Here, $\{\tilde{e}_i\}$ is considered as a basis of $V$, $\rho_V(g)$ is a linear map on $V$, and $\rho_V(g)_{ji}$ is the matrix under this basis. Recall that we distinguish $g^{-1}$ and $(g)^{-1}$. $g^{-1}$ is the inverse of $g$ as a group element, $(g)^{-1}$ is the inverse of $g$ as an algebra element, and we have $(g)^{-1}=g^{-1}\omega(g,g^{-1})^{-1}$. The generators $\tilde{e}_i$ have grading 1, and the group elements $g\in G$ with $h(g)=1$ also have grading 1; this defines a $\mathbb{Z}_2$-grading structure on $\symalg^V$.

In parallel with the construction in \cref{subsection:classifying space}, but using $\symalg^V$ instead of $\symalg$, we obtain the algebra of Dirac Hamiltonians $Cl^{0,1}\hat{\otimes}\symalg^V$, the infinite-dimensional module $\mathcal{U}(Cl^{0,1}\hat{\otimes}\symalg^V)$, and the classifying space for Dirac Hamiltonians $\clspace(\symalg^V)$. Let $\rho$ be the representation homomorphism of $\mathcal{U}(Cl^{0,1}\hat{\otimes}\symalg^V)$. Then $M_0=\rho(\tilde{e}_1\hat\otimes 1)$ is the basepoint of $\clspace(\symalg^V)$, which represents the trivial Dirac Hamiltonian.
\begin{framed}
    For $(q+1)$-dimensional Dirac Hamiltonians with symmetry $\symalg$, the classifying space for (not necessarily symmetric) Hamiltonians is $\clspace(\symalg^V)$, which is the space of linear maps on $\mathcal{U}(Cl^{0,1}\hat{\otimes}\symalg^V)$ such that for any $M\in\clspace(\symalg^V)$, $M$ differs from $M_0$ only on a finite-dimensional subspace, and upon replacing $M_0$ with $M$, the map $M$ yields a representation of $Cl^{0,1}\hat{\otimes}\symalg^V$.

    For $(q+1)$-d Dirac Hamiltonians with symmetry $\symalg$, the classifying space for symmetric Hamiltonians is $\clspace(\symalg^V)^G$, which is the $G$-fixed point subspace of $\clspace(\symalg^V)$, i.e., $(-1)^{h(g)}\rho(g)M\rho(g)^{-1}=M$ for $M\in \clspace(\symalg^V)$.
\end{framed}
The classifying space $\clspace(\symalg^V)$ is the space of mass terms $-M$ on the north point of the Brillouin zone $S^V$. Given $M\in \clspace(\symalg^V)$, we can construct a Dirac Hamiltonian [\cref{eq:Dirac Hamiltonian}] with $-M$ on the north point and $-M_0$ on the south point of $S^V$.

\subsubsection{Comparison of two classifying spaces: A conjecture}\label{subsubsection:comparison of two classifying spaces}

Notice that there is a canonical map from $\clspace(\symalg^V)$ to $\Omega^V (\clspace(\symalg))$ induced\footnote{
    The classifying space $\clspace(\symalg)$ is the space of $H$ differing from $H_0$ only on a finite-dimensional subspace, but the Dirac Hamiltonian [\cref{eq:Dirac Hamiltonian}] differs from $H_0$ on almost all points of the sphere $S^V$; therefore, we have to show that this map is well defined. The map is constructed from a sequence of Dirac Hamiltonians on finite-dimensional submodules of $\mathcal{U}(Cl^{0,1}\hat{\otimes}\symalg^V)$, followed by taking the direct limit. However, even this construction is not well defined because the construction of Dirac Hamiltonians doesn't commute with the inclusion of classifying spaces. Fortunately, it commutes up to homotopy, therefore, we can take the homotopy direct limit, and since the inclusion maps of finite-dimensional classifying spaces are cofibrations, the direct limit of the classifying spaces is homotopy equivalent to the homotopy direct limit. We therefore obtain a well-defined map up to homotopy.
}
by the construction of Dirac Hamiltonians from the mass term using \cref{eq:Dirac Hamiltonian}, and we set the basepoint $H_0$ of $\clspace(\symalg)$ and the basepoint $M_0$ of $\clspace(\symalg^V)$ to be equal [we assume the universe $\mathcal{U}(Cl^{0,1}\hat{\otimes}\symalg)$ is constructed from $\mathcal{U}(Cl^{0,1}\hat{\otimes}\symalg^V)$ by restriction of algebras $Cl^{0,1}\hat{\otimes}\symalg\hookrightarrow Cl^{0,1}\hat{\otimes}\symalg^V$]. There is a conjecture:
\begin{framed}
    $\clspace(\symalg^V)\simeq \Omega^V (\clspace(\symalg))$ is a $G$-equivariant homotopy equivalence.
\end{framed}
This demonstrates that studying the classifying space for Dirac Hamiltonians is sufficient. It should be noted that this result has been used implicitly in physics-related works, for example, Refs. \cite{cornfeldTenfoldTopologyCrystals2021,shiozakiClassificationSurfaceStates2022}.

The $G$-homotopy equivalence relies on a generalization of the fundamental theorem of equivariant $K$-theory established by Karoubi \cite{karoubiKtheorieEquivariante1970}. Assuming this generalization holds, the proof of the $G$-homotopy equivalence is presented in Appendix \ref{section:K-theory}.

\subsubsection{Atomic insulators and Bott periodicity}\label{subsubsection:atomic insulator and Bott periodicity}

The aforementioned $G$-equivariant homotopy equivalence alone is insufficient to package $\clspace(\symalg^V)$ into a genuine $G$-$\Omega$-spectrum. If we denote a genuine $G$-$\Omega$-spectrum as $X_V$, it should satisfy $\Omega^V X_{V\oplus W} \simeq X_W$. Denoting $\clspace(\symalg)$ as $X_0$, currently we only have $\Omega^V X_0 \simeq \clspace(\symalg^V)$; therefore, $\clspace(\symalg^V)$ is effectively $X_{-V}$, and we still need to find the definition of $X_{V}$. Here, $-V$ employs the concept of a virtual representation, which is the formal difference $V-W$ of two representations, making the set of all representations an abelian group.

The task below is to find a real representation $V'$ of $G$ such that $\clspace(\symalg^{V\oplus V'}) \simeq \clspace(\symalg)$. Consequently, we have $X_0 \simeq X_{-V-V'}$, which further leads to $X_{V} \simeq X_{V-V-V'} \simeq X_{-V'}$. Thus, we can define $X_{V}$ as $X_{-V'}$, namely $\clspace(\symalg^{V'})$.

We will prove that $V'$ is the representation satisfying $\rho_{V'}(g) = (-1)^{s(g)}\rho_V(g)$. There exists the following $\mathbb{Z}_2$-graded algebra isomorphism:
\begin{equation}
    \symalg^{V\oplus V'} \cong \symalg^{Q_V\oplus Q_{V''}} \cong Cl^{d, 0} \hat{\otimes} Cl^{0, d} \hat{\otimes} \symalg \textrm{,}
\end{equation}
where the first isomorphism is induced by $\tilde{e}_i\mapsto i e_i$ with $\tilde{e}_i\in V'$. The definition of $\symalg^{Q}$ can be found in \cref{subsection:Z2-Graded Algebra AV}. $(V, Q_V)$ is a space equipped with a quadratic form $Q_V$ whose diagonal elements are all 1. $(V'', Q_{V''})$ is a space equipped with a quadratic form $Q_{V''}$ whose diagonal elements are all $-1$, satisfying $\rho_{V''}(g) = \rho_V(g)$. The second isomorphism is given by the ``separation of variables''; see \cref{thm:separation of variable}. We also need to check whether the two-cocycle and the $\mathbb{Z}_2$-grading after the separation of variables are consistent with $\symalg$. Note that $o_V = o_{V''}$ and $o_{V\oplus V''} = o_{V} + o_{V''} = 0$. Therefore, the $\mathbb{Z}_2$-grading after separation of variables remains $h$. The two-cocycle after separation of variables is
\begin{align}
    & \phantom{{}={}} \omega + \frac{1}{2}\omega_{Q_V\oplus Q_{V''}} + \frac{1}{2}(h + o_{V\oplus V''})\cup o_{V\oplus V''}  \\
    & = \omega + \frac{1}{2}( w_2^+ + w_2^- + w_1^+\cup w_1^- + w_1^-\cup w_1^- ) \\
    & = \omega
\end{align}
The third term in the first line is 0. The second term follows from \cref{cor:2 of sw class}, noting that $w_1^+ = o_V$, $w_1^- = o_{V''}$, and $w_2^+ = w_2^-$. This completes the proof of the algebraic isomorphism. Note that $Cl^{d,d}$ is a $\mathbb{Z}_2$-graded matrix algebra; therefore, the representation theory of $\symalg^{V\oplus V'}$ is identical to that of $\symalg$, which implies that their classifying spaces are $G$-homotopy equivalent. Thus, we obtain $\clspace(\symalg^{V\oplus V'}) \simeq \clspace(\symalg)$.

The $G$-homotopy equivalence $\clspace(\symalg^{V\oplus V'}) \simeq \clspace(\symalg)$ is a generalization of the $(1,1)$-Bott periodicity in KR-theory~\cite{atiyahKtheoryReality1966}. Using this $G$-homotopy equivalence, we can define the classifying space $X_{V-U} \simeq X_{-V'-U} \simeq \clspace(\symalg^{V'\oplus U})$ for any virtual representation $V - U$, where $V'$ is as defined above. Furthermore, we have 
\begin{equation}
    \begin{aligned}
        \Omega^V X_{V + W} & \simeq \Omega^V X_{- V' - W'} \simeq \Omega^{V\oplus V' \oplus W'} X_{0} \\
        & \simeq \Omega^{W'} X_{- V - V'} \simeq \Omega^{W'} X_{0} \simeq X_{-W'} \\
        & \simeq X_{W}.
    \end{aligned}
\end{equation}
Therefore, $\clspace(\symalg^V)$ is a genuine $G$-$\Omega$-spectrum.

Physically, the classification of atomic insulators is identical to that of zero-dimensional systems, and atomic insulators lack nontrivial boundary states. Note that the physical significance of the equivalence $\clspace(\symalg^{V\oplus V'}) \simeq \clspace(\symalg)$ is that the classifying space of zero-dimensional systems is equivalent to the classifying space of systems on $V\oplus V'$, where the latter constitutes the classifying space of systems with no nontrivial boundary states. This $G$-homotopy equivalence provides a rigorous description of the boundary-state problem of atomic insulators.

\subsection{The classification of gapped free fermionic Hamiltonians}\label{subsection:The classification of gapped free fermionic hamiltonians}

\subsubsection{The topological classification}\label{subsubsection:the topological classification}

The classification of free fermionic systems is given by the connected components of the classifying spaces. We use 
\begin{equation}
    \pi_0\left( \left(\Omega^V \clspace(\symalg)\right)^G\right) \cong \pi_0\left(\clspace(\symalg^V)^G\right)
\end{equation}
as the classification, where $\pi_0$ is the zeroth homotopy group, which corresponds to the set of path-connected components.

\subsubsection{The algebraic classification}\label{subsubsection:the algebraic classification}

The Grothendieck group $\mathcal{M}(A)$ of a semisimple algebra $A$ is the free abelian group generated by the isomorphism classes of irreducible representations of $A$. Given an algebra homomorphism $f:A\to B$, it induces a group homomorphism $f^*:\mathcal{M}(B)\to \mathcal{M}(A)$.

We define the algebraic classification $\mathcal{A}(\symalg^V)$ of free fermionic systems with symmetry $\symalg$ as the cokernel of
\begin{equation}
    i^*:\mathcal{M}(Cl^{0,2}\hat\otimes \symalg^V)\to \mathcal{M}(Cl^{0,1}\hat\otimes \symalg^V)
\end{equation}
induced by the inclusion 
\begin{equation}
    i:Cl^{0,1}\hat\otimes \symalg^V\to Cl^{0,2}\hat\otimes \symalg^V,
\end{equation}
i.e., $\mathcal{A}(\symalg^V) \coloneq \mathcal{M}(Cl^{0,1}\hat\otimes \symalg^V)/\mathrm{Im}(i^*)$.

There is an equivalent way to construct the algebraic classification using the $\mathbb{Z}_2$-graded Grothendieck group. For a $\mathbb{Z}_2$-graded semisimple algebra $A$, the Grothendieck group $\mathcal{N}(A)$ is a free abelian group that is generated by the isomorphism classes of irreducible $\mathbb{Z}_2$-graded representations of $A$. The algebraic classification is given by the cokernel of $i^*:\mathcal{N}(Cl^{0,1}\hat\otimes \symalg^V)\to \mathcal{N}(\symalg^V)$; we still denote it by $\mathcal{A}(\symalg^V)$.

To show the equivalence of the two Grothendieck groups above, notice that a $\mathbb{Z}_2$-grading structure is equivalent to an extra linear map $f$ satisfying $f^2=\mathrm{id}$. The grading 0 subspace corresponds to $f=1$, and the grading 1 subspace corresponds to $f=-1$.  In our case, $f$ is precisely the element $\tilde{e}_1 \hat{\otimes} 1$ in $Cl^{0,1} \hat{\otimes} \symalg^V$.

\subsubsection{Comparison of two classifications}\label{subsubsection:comparison of two classifications}

There is a canonical map from the algebraic classification to the topological classification, i.e., $\mathcal{A}(\symalg^V)\to \pi_0\left(\clspace(\symalg^V)^G\right)$. The map is obtained from ``changing the sign of the mass term'', which is a very common technique in the physics study of boundary states. Here, we explain the precise mathematical meaning of this technique.

Given $[W]\in \mathcal{A}(\symalg^V)$, where $W$ is a module over $Cl^{0,1}\hat{\otimes}\symalg^V$, and an arbitrary injective map $W\hookrightarrow \mathcal{U}(Cl^{0,1}\hat{\otimes}\symalg^V)$, the image of $[W]$ is given by $-M_0|_W$, where $M_0$ is the basepoint of $\clspace(\symalg^V)^G$.
\begin{framed}
    The map ``changing the sign of the mass term'' $\mathcal{A}(\symalg^V)\cong \pi_0\left(\clspace(\symalg^V)^G\right)$ is a group isomorphism.
\end{framed}
The proof can be found in \cref{subsubsection:decomposition of the classifying spaces}.

\subsubsection{\texorpdfstring{Bundle-theoretic classification}{Bundle-theoretic classification}}\label{subsubsection:bundle-theoretic classification}

Let $H(k)$ be a Hamiltonian on the $k$-space $S^V$ with symmetry $\symalg$. There is a vector bundle $E$ over $S^V$ spanned by wave functions $\{ \psi_i(k) \}_{i=1}^N$. $\{ \psi_i(k) \}_{i=1}^N$ can be chosen such that it forms a global frame of $E$; therefore, the vector bundle $E$ is trivial. The Hamiltonian $H(k)$ is a vector bundle map from $E$ to itself. If we solve the Hamiltonian $H(k)$, we will get a vector bundle\footnote{It is proved in \cite{karoubiKtheoryIntroduction2009} that for a vector bundle map $f$ satisfying $f^2=\mathrm{id}$, $\mathrm{Ker}(f\pm 1)$ always exists. Thus, $\mathrm{Ker}(H-1)$ is $E_0$, and $\mathrm{Ker}(H+1)$ is $E_1$. This property is called Karoubi completeness.} $E_0$ on $S^V$ spanned by states with $H(k)=1$, and a vector bundle $E_1$ on $S^V$ spanned by states with $H(k)=-1$. We have $E\cong E_0\oplus E_1$.

There is an $\symalg$ action on the trivial bundle $E$. The action of $g\in \symalg$ maps a vector in the fiber over $k\in S^V$ to a vector in the fiber over $g\cdot k$. If $H(k)$ is $G$ symmetric, for $h(g)=0$, the action of $g$ preserves $E_0$ and $E_1$; for $h(g)=1$, the action of $g$ exchanges $E_0$ and $E_1$.

Motivated by the discussion above, we define a class of vector bundles equipped with these structures; we denote the category of such bundles on a $G$-space $X$ by $\mathrm{Vect}_{\symalg}(X)$. An object of $\mathrm{Vect}_{\symalg}(X)$ is a trivial vector bundle $E$ that is the product of $X$ and an ungraded $\symalg$ module. The bundle $E$ carries a $\mathbb{Z}_2$-grading compatible with the action of $G$, i.e., $E \cong E_0 \oplus E_1$, satisfying the condition that, for $g \in G$, if $h(g) = 0$ then $g$ acts as a linear or antilinear bundle map preserving $E_0$ and $E_1$, while if $h(g) = 1$ then $g$ acts as a linear or antilinear bundle map that exchanges $E_0$ and $E_1$.

A morphism of $\mathrm{Vect}_{\symalg}(X)$ from $E$ to $E'$ is a bundle map that is a $\symalg$-module map on each fiber and preserves the $\mathbb{Z}_2$-grading.

We define a bundle-theoretic Abelian group $\tilde{\ktheory}_{\symalg}(X)$
\footnote{The definition of $\tilde{\ktheory}_{\symalg}(X)$ may look unusual at first glance; however, it is actually equivalent to Karoubi's definition of $K$-theory (see \cref{subsection:reduced K-theory} for details). It also reduces to several familiar $K$-theories for specific choices of $\symalg$.}.
An element of $\tilde{\ktheory}_{\symalg}(X)$ is an equivalence class of isomorphism classes of $\mathrm{Vect}_{\symalg}(X)$. The equivalence is given by $(E_0,E_1)\sim (E_0\oplus F_0,E_1\oplus F_1)$, where $F$ is an element of $\mathrm{Vect}_{\symalg}(X)$ such that $F_0$ and $F_1$ are trivial vector bundles; note that $F_0$ and $F_1$ may carry nontrivial actions of $\symalg$. The abelian group structure is defined by $(E_0,E_1)\oplus(E'_0,E'_1):=(E_0 \oplus E'_0 , E_1 \oplus E'_1)$.

The physical interpretation of $F=(F_0, F_1)$ is that of a trivial state: a system consisting of trivial bands that carry some representation of $\symalg$. $F_0$ corresponds to a trivial band of unoccupied states, while $F_1$ corresponds to a trivial band of occupied states. The equivalence relation then reflects the idea that adding trivial bands to the system does not change its classification.

\begin{eg}
    If $\symalg$ has only linear symmetries, a trivial two-cocycle, and a trivial $\mathbb{Z}_2$-grading $h$, then $\tilde{\ktheory}_{\symalg}(X)$ reduces to the reduced $G$-equivariant complex $K$-theory. In this case, $E_0$ and $E_1$ are complex $G$ bundles, and $E_0$ and $E_1$ are inverses of each other in the $K$ group. Hence, $E_0$ and $E_1$ determine one another, so only one of them is needed to represent the $K$ group.
\end{eg}

\subsubsection{Comparison to ABS construction}\label{subsubsection:comparison of algebraic classification and abs construction}

The ABS construction \cite{atiyahCliffordModules1964} is an isomorphism from the algebraic classification $\mathcal{A}(\symalg^V)$ to the bundle-theoretic classification $\tilde{\ktheory}_{\symalg}(S^V)$.

Let $W$ be a $\mathbb{Z}_2$-graded $\symalg^V$ module. The ABS construction is given by the direct sum $E_0 \oplus E_1$. The space $E_0$ is the clutching of $S^V_+ \times W_0$ and $S^V_- \times W_1$ along the boundaries of $S^V_+$ and $S^V_-$, where the fibers are identified via the transition function $v \mapsto \rho_W(v)$. Here, $v$ lies in the boundary of $S^V_+$ and may also be regarded as an element of $\symalg^V$ by viewing $\tilde{e}_i$ as a basis of $V$. The space $E_1$ is the clutching of $S^V_+ \times W_1$ and $S^V_- \times W_0$ with the transition function $v \mapsto -\rho_W(v)$.

Next, we show that the ABS construction $E$ of $W$ is the vector bundle $F$ obtained by solving the Dirac Hamiltonian [\cref{eq:Dirac Hamiltonian}] with $M=-M_0$. Recall that $W$ can also be considered as an ungraded module over $Cl^{0,1}\hat{\otimes}\symalg^V$. As a direct consequence, the diagram \cref{eq:diagram of classification} commutes.

We only need to show $E$ and $F$ are isomorphic on each hemisphere, and that the transition function of $E$ acts trivially on each fiber of $F$ over the boundaries of the hemispheres (using the local isomorphism on each hemisphere to transfer the transition function of $E$ to $F$). Note that $F$ is a subbundle of a trivial bundle; we use the elements of the trivial bundle to represent the elements of $F$; therefore, the transferred transition function must act trivially on $F$.

We first consider $F_0$ with $H(k)=1$. The vector bundle on $S_+^V$ is trivial; there is an isomorphism $f: F_0|_{S^V_+} \to S^V_+ \times W_0\cong E_0|_{S^V_+}$. $f$ is given by $(k,v)\mapsto (k,v+M_0v)$ for $v\in W$ with $H(k)v=v$, and $f^{-1}$ is given by $(k,w)\mapsto (k,\frac{1}{2(1+k_0)}(w+H(k)w))$ for $w\in W_0$. The isomorphism $g: F_0|_{S^V_-} \to S^V_- \times W_1\cong E_0|_{S^V_-}$ is given by $(k,v)\mapsto (k,v-M_0v)$ for $v\in W$ with $H(k)v=v$, and $g^{-1}$ is given by $(k,w)\mapsto (k,\frac{1}{2(1-k_0)}(w+H(k)w))$ for $w\in W_1$. The transition function on $E$ from $S_+^V$ to $S_-^V$ coincides with $H(k)$ on $F$; therefore, it acts trivially.

For $F_1$ with $H(k)=-1$, we have the isomorphism $h: F_1|_{S^V_+} \to S^V_+ \times W_1\cong E_1|_{S^V_+}$. $h$ is given by $(k,v)\mapsto (k,v-M_0v)$ for $v\in W$ with $H(k)v=-v$, and $h^{-1}$ is given by $(k,w)\mapsto (k,\frac{1}{2(1+k_0)}(w-H(k)w))$ for $w\in W_1$. We also have the isomorphism $l: F_1|_{S^V_-} \to S^V_- \times W_0\cong E_1|_{S^V_-}$. $l$ is given by $(k,v)\mapsto (k,v+M_0v)$ for $v\in W$ with $H(k)v=-v$, and $l^{-1}$ is given by $(k,w)\mapsto (k,\frac{1}{2(1-k_0)}(w-H(k)w))$ for $w\in W_0$. The transition function on $E$ from $S_+^V$ to $S_-^V$ coincides with $-H(k)$ on $F$; therefore, it acts trivially.

\subsection{Tenfold way}\label{subsection:Tenfold Way}

\subsubsection{Calculation of the Algebraic Classification}\label{subsubsection:calculation of the algebraic classification}

$\symalg^V$ is a $\mathbb{Z}_2$-graded semisimple $\mathbb{R}$ algebra; it can be decomposed into a product of $\mathbb{Z}_2$-graded matrix algebras over $\mathbb{Z}_2$-graded division algebras, i.e., the $\mathbb{Z}_2$-graded Wedderburn--Artin decomposition. We devote the remaining sections to calculating the decomposition. Thus, throughout the rest of this section, we assume that the decomposition has already been obtained, i.e., we have \begin{equation}
    \symalg^V\cong \prod_i M_{d_i}(D_i)(\bar{c}_i),
\end{equation}
where $D_i$ is a $\mathbb{Z}_2$-graded division algebra. The definition of $\mathbb{Z}_2$-graded matrix algebra can be found in \cref{subsubsection:semisimple graded algebra}.

There are ten different isomorphism classes of $\mathbb{Z}_2$-graded division $\mathbb{R}$ algebras. Three of them are ordinary division algebras $\mathbb{R}$, $\mathbb{C}$, and $\mathbb{H}$, and the other seven are $Cl^{0,1}$, $Cl^{1,0}$, $Cl^{2,0}$, $Cl^{0,2}$, $\mathbb{C}l^{1}$, $Cl^{0,3}$ and $Cl^{3,0}$. This is where the ``tenfold way'' comes from.

To calculate $\mathcal{A}(\symalg^V)$, it is easy to show that the $\mathbb{Z}_2$-graded Grothendieck group has the property $\mathcal{N}(\prod_i A_i)\cong \bigoplus_i \mathcal{N}(A_i)$, thus giving the decomposition
\begin{equation}
    \mathcal{N}(\symalg^V)\cong \bigoplus_i \mathcal{N}(D_i),
\end{equation}
here we used the fact that tensoring with a $\mathbb{Z}_2$-graded matrix algebra doesn't change the equivalence class of the $\mathbb{Z}_2$-graded module category of a $\mathbb{Z}_2$-graded algebra, equivalently, they are $\mathbb{Z}_2$-graded Morita equivalent. Similarly, we have
\begin{equation}
    \mathcal{N}(Cl^{0,1}\hat\otimes \symalg^V)\cong \bigoplus_i \mathcal{N}(Cl^{0,1}\hat{\otimes}D_i).
\end{equation}
Thus, we obtain the decomposition,
\begin{equation}
    \mathcal{A}(\symalg^V)\cong \bigoplus_i \mathcal{A}(D_i).
\end{equation}
Now the calculation of $\mathcal{A}(\symalg^V)$ reduces to the calculation of $\mathcal{A}(D_i)$, which is a well-known result; see, for example, Ref. \cite{atiyahCliffordModules1964}. We record the value of $\mathcal{A}(D_i)$ in \cref{table:A(Clpq)}.

\subsubsection{Decomposition of the classifying spaces}\label{subsubsection:decomposition of the classifying spaces}

We calculate the structure of $\clspace(\symalg^V)^G$. Notice that we have
\begin{equation}
    \mathcal{U}(Cl^{0,1}\hat\otimes \symalg^V)\cong \bigoplus_i \mathcal{U}\left(Cl^{0,1}\hat\otimes M_{d_i}(D_i)(\bar{c}_i)\right).
\end{equation}
$Cl^{0,1}\hat\otimes \symalg^V$ acts on the right-hand side through the projection to each product factor of the Wedderburn--Artin decomposition. $\clspace(\symalg^V)^G$ is the space of the linear map $M$ representing $\tilde{e}_1\hat{\otimes} 1$, or equivalently, the space of gradings. Notice that $M$ is an $(\symalg^V)_0$-module map, and each direct sum factor of $\mathcal{U}(Cl^{0,1}\hat\otimes \symalg^V)$ only contains nonisomorphic $(\symalg^V)_0$ modules; therefore, $M$ will not mix different direct sum factors. Thus, we get the decomposition
\begin{equation}
    \clspace(\symalg^V)^G\cong \prod_i \clspace\left(M_{d_i}(D_i)(\bar{c}_i)\right)\simeq \prod_i \clspace(D_i).
\end{equation}

Finally, notice that the map ``changing the sign of the mass term'' is also compatible with the decomposition. Thus, to prove that the map is an isomorphism, it suffices to prove that $\mathcal{A}(D_i)\to \pi_0\left(\clspace(D_i)\right)$ is an isomorphism, which can be checked case by case for all ten basic cases.

%% file: 3.graded_algebra.tex
\section{Graded Algebras}\label{section:Graded Algebra}

Let $\Gr$ be a group. A $\Gr$-graded algebra is an algebra equipped with a $\Gr$-grading structure. Many algebras naturally have a grading structure. For example, the group algebra containing antilinear elements has a natural $\mathbb{Z}_2$-grading, and $\symalg$ has a natural $\mathbb{Z}_2\times \mathbb{Z}_2$-grading. In this section, we first review the basic theory of graded algebras, referring to \cite{nastasescuMethodsGradedRings2004} as a monograph on graded algebras. We then focus on $\Gr = \mathbb{Z}_2^N$. Readers familiar with graded algebra fundamentals may proceed by reading only \cref{thm:structure theorem of Gr-crossed product algebra} and the classification of $\mathbb{Z}_2^N$-graded division algebras in \cref{subsection:classification of Z2N-graded division algebra}; they may skip this chapter and return when needed.

In what follows, we assume that $\Gr$ is an abelian group, and we use $+$ for the multiplication.

\subsection{Definitions}

\subsubsection{Graded vector spaces}\label{subsubsection:graded vector space}

Let $\mathrm{Vec}_{\mathbb{K}}^{\Gr}$ be the category of $\Gr$-graded vector spaces over $\mathbb{K}$. The objects are $\Gr$-graded vector spaces $V=\bigoplus_{c\in \Gr}V_c$, and the morphisms are linear maps $f:V\to W$ such that $f(V_c)\subset W_c$. The tensor product of $V$ and $W$ is defined by $(V\otimes W)_c=\bigoplus_{c_1+c_2=c}V_{c_1}\otimes W_{c_2}$, which makes $\mathrm{Vec}_{\mathbb{K}}^{\Gr}$ a monoidal category.

An algebra object internal to a monoidal category is an object $A$ equipped with a multiplication $A\otimes A\to A$ and a unit $1\to A$ satisfying the associativity and unit axioms. Similarly, an $A$-module object internal to a monoidal category is an object $V$ equipped with an $A$ action $A\otimes V\to V$ satisfying the associativity and unit axioms. By taking the monoidal category to be $\mathrm{Vec}_{\mathbb{K}}^{\Gr}$, the algebra object becomes a $\Gr$-graded algebra, and the $A$-module object becomes a $\Gr$-graded $A$ module.

Let $\mathrm{Hom}^\Gr(V,W)$ be the vector space of $\Gr$-graded maps, and $\underline{\mathrm{Hom}}^\Gr(V,W)$ be the $\Gr$-graded vector space of all linear maps whose $c$ component is the vector space of $\Gr$-graded maps $V\to W(c)$, where $W(c)$ is the same vector space as $W$ with grading shifted by $c$, i.e., $W(c)_{c'}=W_{c'+c}$. Similarly, $\mathrm{Mod}_A^\Gr(V,W)$ is the vector space of $A$-module maps that preserve the grading, and $\underline{\mathrm{Mod}}_A^\Gr(V,W)$ is the $\Gr$-graded vector space of all $A$-module maps.

\subsubsection{Semisimple graded algebra}\label{subsubsection:semisimple graded algebra}

We first recall the definition of semisimple algebras. An algebra $A$ is semisimple if it is a direct sum of simple $A$ modules. The major property of semisimple algebras is the Wedderburn--Artin theorem, which states that any semisimple algebra is isomorphic to a finite product of matrix algebras over division algebras, i.e., $\Phi=(\rho_i):A \cong \prod_i M_{d_i}(D_i)$, where $i$ takes values in the isomorphism classes of simple $A$-modules, and $D_i$ are division algebras. For $\mathbb{K}=\mathbb{C}$, $D_i$ must be $\mathbb{C}$, and for $\mathbb{K}=\mathbb{R}$, $D_i$ is either $\mathbb{R}$, $\mathbb{C}$, or $\mathbb{H}$.

A simple module is of $D$ type if the endomorphism algebra of the module is $D$. The irreducible representations can be extracted from the components $\rho_i$ of $\Phi$. Taking the canonical module $D_i^d$ for $M_d(D_i)$ and forgetting the $D_i$ structure, $\rho_{i}$ becomes an irreducible representation of $A$ with a canonical $D_i^{\mathrm{op}}$ structure acting from the right-hand side (see Appendix \ref{subsection:Convention for Matrix Algebra} for details), i.e., the module is of $D_i^{\mathrm{op}}$ type.

The theory of $\Gr$-graded semisimple algebras is similar. A $\Gr$-graded algebra $A$ is semisimple if it is a direct sum of simple $\Gr$-graded $A$ modules. A $\Gr$-graded semisimple algebra can also be decomposed into a finite product of $\Gr$-graded matrix algebras over $\Gr$-graded division algebras, which is called the $\Gr$-graded Wedderburn--Artin decomposition. Thus, we first introduce the definition of $\Gr$-graded matrix algebras. Let $A$ be a $\Gr$-graded algebra. Then, $M_d(A)(\bar{c})$ is the $\Gr$-graded algebra of $d\times d$ matrices over $A$ with shifted grading $\bar{c}=(c_1,\dots,c_d)\in \Gr^d$. The grading $\lambda$ component is given by
\begin{equation}\label{eq:graded matrix algebra}
    M_d(A)(\bar{c})_{\lambda}=
    \begin{bmatrix}
        A_{c_1+\lambda-c_1} & A_{c_1+\lambda-c_2} & \cdots & A_{c_1+\lambda-c_d} \\
        A_{c_2+\lambda-c_1} & A_{c_2+\lambda-c_2} & \cdots & A_{c_2+\lambda-c_d} \\
        \vdots & \vdots & \ddots & \vdots \\
        A_{c_d+\lambda-c_1} & A_{c_d+\lambda-c_2} & \cdots & A_{c_d+\lambda-c_d}
    \end{bmatrix}.
\end{equation}
Then, the $\Gr$-graded Wedderburn--Artin theorem is
\begin{thm}\label{thm:structure theorem of semisimple graded algebra}
    Let $A$ be a semisimple $\Gr$-graded algebra. Then, $A$ is isomorphic to $\prod_{i}M_{d_i}(D_i)(\bar{c}_i)$, where $i$ takes values in the isomorphism classes of simple $A$ modules up to grading shifts, $\bar{c}_i=(c_{i1},\dots,c_{id_i})\in \Gr^{d_i}$, and $D_i$ is a graded division algebra, i.e., all nonzero homogeneous elements $a\in D_c$ are invertible.
\end{thm}
\begin{proof}
    $A$ can be decomposed into a direct sum of simple right $A$ modules. Let $V$ and $W$ be two simple modules, which are either isomorphic or nonisomorphic. If they are nonisomorphic, then $\mathrm{Mod}_A^\Gr(V,W)=0$. However, $\underline{\mathrm{Mod}}_A^\Gr(V,W)$ may not be 0, and $V$ could be isomorphic to $W(c)$ for some $c\in \Gr$. Thus, we can decompose $A$ into $\bigoplus_i\left( \bigoplus_{j=1}^{d_i} V_i(c_{ij})\right)$, where each $V_i$ is simple and $\underline{\mathrm{Mod}}_A^\Gr(V_i,V_j)=0$ for $i\not =j$.

    We have
    \begin{widetext}
    \begin{equation}
        A \cong \underline{\mathrm{Mod}}_A^\Gr(A,A) \cong \bigoplus_i \left(\bigoplus_{j=1}^{d_i} \bigoplus_{k=1}^{d_i} \underline{\mathrm{Mod}}_A^\Gr\left(V_i(c_{ij}),V_i(c_{ik})\right)\right),
    \end{equation}
    \end{widetext}
    and define $\underline{\mathrm{Mod}}_A^\Gr\left(V_i(c_{ij}),V_i(c_{ik})\right)$ to be the $(k,j)$ entry of the graded matrix algebra $M_{d_i}(D_i)(\bar{c}_i)$, where $D_i$ is $\underline{\mathrm{Mod}}_A^\Gr(V_i,V_i)$.
\end{proof}
All $\Gr$-graded irreducible representations can be extracted from the decomposition, in a way similar to the ungraded case.

\subsubsection{Braiding structure}\label{subsubsection:Braiding in Graded Vector Space}

Many concepts in representation theory actually involve braiding structure. For example, in order to define the multiplication of the tensor product algebra $A\otimes B$, we first need to use an exchange map from $A\otimes B\otimes A\otimes B$ to $A\otimes A\otimes B\otimes B$ acting on the middle two terms, and then multiply two pairs of $A$ and $B$. The exchange of $A$ and $B$ is exactly the braiding structure. Another example is the multiplication of the opposite algebra $A^{\mathrm{op}}$. We first need to use the exchange map from $A\otimes A$ to itself, and then multiply them. Other examples include the graded vector space structure on the hom space of $A$ modules, and the tensor product of $A$ modules.

The braiding structure on $\mathrm{Vec}_{\mathbb{K}}^{\Gr}$ is not unique. The classification of braiding structures on $\mathrm{Vec}_{\mathbb{K}}^{\Gr}$ is given in \cite{etingofTensorCategories2015}. In this work, we only consider the braiding structure twisted by a symmetric bilinear form $b:\Gr\times \Gr\to \mathbb{Z}_2$, where $\Gr=\mathbb{Z}_2^N$.

For example, if we take $\Gr=\mathbb{Z}_2$, we have two braiding structures. The first one is the trivial braiding structure, and the second one is the braiding structure of the super vector space, which is given by $V_c\hat{\otimes} W_{c'}\xrightarrow{b(c,c')} W_{c'}\hat{\otimes} V_c$, where $b(c,c')=-1$ if $c=c'=1$, and $b(c,c')=1$ otherwise.

We adopt the following convention: the braiding structure on $V\otimes W$ is trivial, while the braiding structure on $V\hat{\otimes} W$ is that of super vector spaces.

\subsection{\texorpdfstring{$\Gr$-crossed product algebra}{Gr-crossed product algebra}}\label{subsection:Gr-Crossed Product Algebra}

A $\Gr$-crossed product algebra is a $\Gr$-graded algebra $A$ such that there exists an invertible element $a_c\in A_c$ for each $c\in \Gr$.

\begin{prop}\label{prop:equivalence of categories between right A-modules and A_0-modules}
    Let $A$ be a $\Gr$-crossed product algebra. There is an equivalence between the category of $A_0$ modules and the category of graded $A$ modules,
    \begin{equation}
        A\otimes_{A_0}-:
        \begin{tikzcd}
            \mathrm{Mod}_{A_0}
              \arrow[r, "\simeq"{above}, shift left=1.2ex]
            & \mathrm{Mod}_A^{\Gr}
              \arrow[l, shift left=1.2ex]
            \arrow[phantom, from=1-1, to=1-2, "\perp" sloped, midway]
        \end{tikzcd}
        : (-)_0,
    \end{equation}
    where $(-)_0$ is the restriction of the graded $A$ modules to the grading 0 $A_0$ submodule.
\end{prop}
\begin{proof}
    We prove $V_0\otimes_{A_0} A\cong V$ for any $A$ module $V$. From the left to the right, we have $v\otimes_{A_0}a \mapsto v\cdot a$ for $v\in V_0$ and $a\in A$. From the right to the left, we have $v_c\mapsto v_0\otimes_{A_0}a_c$ for $v_c\in V_c$ and $v_0=v_c \cdot a_c^{-1}$, where $a_c$ is the invertible element.
\end{proof}

Let $V$ be a right $A_0$ module. We define a new $A_0$ module $V^{a_c}$ by $v\cdot_{V^c}b=v\cdot (a_cba_c^{-1})$, where $b\in A$, $v\in V$ and $a_c\in A_c$ is invertible. We also denote $V^{a_c}$ by $V^c$ if a set of invertible elements $\{a_c\}_{c\in \Gr}$ is given. Notice that the isomorphism class of $V^c$ does not depend on the choice of $a_c$. It is easy to show that there is a commutative diagram up to natural isomorphisms,
\begin{equation}\label{eq:commutative diagram of N^c}
    \begin{tikzcd}
        \mathrm{Mod}_{A_0}
            \arrow[r, "\simeq"{above}, shift left=1.2ex]
            \arrow[d,"(-)^c"{left},"\simeq"{right}]
        & \mathrm{Mod}_A^{\Gr}
            \arrow[l, shift left=1.2ex]
            \arrow[d,"(-)(c)","\simeq"{left}]
            \arrow[phantom, from=1-1, to=1-2, "\perp" sloped, midway]\\
        \mathrm{Mod}_{A_0}
            \arrow[r, "\simeq"{above}, shift left=1.2ex]
        & \mathrm{Mod}_A^{\Gr}
            \arrow[l, shift left=1.2ex]
            \arrow[phantom, from=2-1, to=2-2, "\perp" sloped, midway]
    \end{tikzcd}
\end{equation}
The functor $(-)(c)$ is the grading shifting functor.

Now we introduce the Wedderburn--Artin theorem for $\Gr$-crossed product algebras. According to \cref{thm:structure theorem of semisimple graded algebra}, we already know that the semisimple graded algebras can be decomposed into graded matrix algebras. However, \cref{eq:commutative diagram of N^c} gives more restrictions on each graded matrix component.
\begin{thm}\label{thm:structure theorem of Gr-crossed product algebra}

    Let $A$ be a semisimple $\Gr$-crossed product algebra, and $V$ be a simple right $A_0$-module that corresponds to a matrix algebra component $M_{d}(D_0)$ of the Wedderburn--Artin decomposition of $A_0$. Then, the graded matrix algebra component that corresponds to $V\otimes_{A_0}A$ has the form $M_{\left|\Gr/E\right|d}(D)(\bar{c})$. Here, $E=\left\{c\in \Gr|V^c\cong V\right\}$ is a subgroup of $\Gr$, $\{c_i\}$ is a set of representatives of $\Gr/E$, and we have
    \begin{equation}
        \bar{c}=(c_1,\dots,c_1,c_2,\dots,c_2,\cdots,c_{\left|\Gr/E\right|},\dots,c_{\left|\Gr/E\right|}),
    \end{equation}
    where each $c_i$ appears $d$ times. $D$ is an $E$-graded division algebra with $D_c\not = 0$ for $c\in E$, and $D_0$ happens to be the grading 0 subalgebra of $D$.

\end{thm}

\begin{proof}
    In the proof of \cref{prop:equivalence of categories between right A-modules and A_0-modules}, we showed that each graded matrix algebra component in the decomposition is the endomorphism of the direct sum of all simple $A$ submodules of $A$ that are isomorphic to $W$ up to a grading shift, where $W$ is a simple $A$ module. Using \cref{eq:commutative diagram of N^c}, we can find all such simple $A$ submodules at the level of $A_0$ modules. Let $V = W_0$, which corresponds to the matrix component $M_{d}(D_0)$ of $A_0$. We get a set of nonisomorphic $A_0$ modules $\{V^{c_i}\}$, and the induced $A$ modules $V^{c_i}\otimes_{A_0}A$ are isomorphic up to a grading shift. Notice that $A_0$ has direct summands $(V^{c_i})^{\oplus d}$ as a right $A_0$ module, and we have
    \begin{equation}
        \prod_{[c_i]\in \Gr/E} \mathrm{Mod}_{A_0}\left(\left(V^{c_i}\right)^{\oplus d},\left(V^{c_i}\right)^{\oplus d}\right) \cong \prod_{[c_i]\in \Gr/E} M_{d}(D_0),
    \end{equation}
    where $D_0\cong\mathrm{Mod}_{A_0}(V^c,V^c)\cong \mathrm{Mod}_{A_0}(V,V)$.

    The $A$ modules $V^{c_i}\otimes_{A_0}A\cong V\otimes_{A_0}A(c_i)$ are exactly what we are finding. Thus, the graded matrix algebra component is
    \begin{widetext}
    \begin{equation}
        \begin{aligned}
             & \underline{\mathrm{Mod}}_A^\Gr\left(\bigoplus_{[c_j]\in \Gr/E}\left(V\otimes_{A_0}A(c_j)\right)^{\oplus d},\bigoplus_{[c_k]\in \Gr/E}\left(V\otimes_{A_0}A(c_k)\right)^{\oplus d}\right)  \\
            \cong & M_{\left|\Gr/E\right|d}(D)(\bar{c}),
        \end{aligned}
    \end{equation}
    \end{widetext}
    where $D\cong \underline{\mathrm{Mod}}_A^\Gr(V\otimes_{A_0}A,V\otimes_{A_0}A)$, and the grading $c$ component of $D$ is $\mathrm{Mod}_A^\Gr(V\otimes_{A_0}A,V\otimes_{A_0}A(c))\cong \mathrm{Mod}_{A_0}(V,V^c)$. Thus, $D$ is $E$-graded, since $D_c=0$ for $c\not \in E$.
\end{proof}

\subsection{\texorpdfstring{The classification of $\mathbb{Z}_2^N$-graded division algebra}{The classification of Z2N-graded division algebra}}\label{subsection:classification of Z2N-graded division algebra}

In this section, we classify $\mathbb{Z}_2^N$-graded division algebras over $\mathbb{R}$ and $\mathbb{C}$.

\subsubsection{\texorpdfstring{$\mathbb{Z}_2$-graded division algebra}{Z2-graded division algebra}}\label{subsubsection:classification of Z2-graded division algebra}

We review the classification of $\mathbb{Z}_2$-graded division algebras \cite{brunshidleFundamentalTheoremFinite2012}. If $D$ is a $\mathbb{Z}_2$-graded division algebra, then $D_0$ is a division algebra. When $D_1=0$, the classification reduces to the classification of ordinary division algebras. Thus, we assume $D_1\not = 0$ in what follows.

Let $\theta\in D_1$ be an invertible element. We denote $\theta^2\in D_0$ by $a$. $\theta$ induces an algebra automorphism of $D_0$, which is denoted by $f$ with $f(d)=\theta d \theta^{-1}$. It is straightforward to show that $f^2(d)=ada^{-1}$ and $f(a)=a$.

Conversely, given $(D_0,f,a)$ such that $f^2(d)=ada^{-1}$ and $f(a)=a$, we can construct a $\mathbb{Z}_2$-graded algebra $D\cong D_0\oplus D_0\theta$. An element of $D$ has the form $d$ or $d\theta$ for $d\in D_0$. The multiplication is induced by the rules $\theta^2=a$ and $\theta d \theta^{-1}=f(d)$.

We have shown that the classification of $\mathbb{Z}_2$-graded division algebras is equivalent to the classification of the data $(D_0,f,a)$ satisfying $f^2(d)=ada^{-1}$ and $f(a)=a$. When $D_0$ is central simple, we have the following lemma:
\begin{lem}
   If $D_0$ is a central simple algebra over $\mathbb{K}$, the $\mathbb{Z}_2$-graded division algebra constructed from $(D_0,f,a)$ is isomorphic to the $\mathbb{Z}_2$-graded division algebra constructed from $(D_0,f',a')$ such that $a'\in \mathbb{K}$ and $f'=\mathrm{id}$.
\end{lem}
\begin{proof}
    Since $D_0$ is a central simple algebra, according to the Skolem--Noether theorem, any automorphism of $D_0$ is an inner automorphism. Thus, we have $f(d)=bdb^{-1}$ for some $b\in D_0$. If we set $\theta'=b^{-1}\theta$, then we have $f'(d)=\theta'd\theta'^{-1}=d$ for $d\in D_0$. Moreover, $a'=\theta'^2\in \mathbb{K}$, since $d=(f')^2(d)=a'd(a')^{-1}$ for $d\in D_0$.
\end{proof}

For $\mathbb{Z}_2$-graded division algebras over $\mathbb{C}$, $D_0=\mathbb{C}$ is a central simple algebra. Therefore, the only choice of $D$ is $(\mathbb{C},\mathrm{id},1)$ up to isomorphism. Notice that if $a\not = 1$, we can take $\theta'$ to be $a^{-\frac{1}{2}}\theta$.

For $\mathbb{Z}_2$-graded division algebras over $\mathbb{R}$, when $D_0=\mathbb{R}$ or $\mathbb{H}$, the only two choices are $(D_0,\mathrm{id},\pm 1)$ up to isomorphism. When $D_0=\mathbb{C}$, it is well known that the Galois group of $\mathbb{C}$ over $\mathbb{R}$ is $\mathbb{Z}_2$; hence, the only two choices of $f$ are $\mathrm{id}$ and complex conjugation. When $f=\mathrm{id}$, $\theta$ commutes with $i$; therefore, $D$ is an algebra over $\mathbb{C}$. Thus, the only choice is $(\mathbb{C},\mathrm{id},1)$ up to isomorphism. When $f$ is the complex conjugation, $f(a)=a$ shows $a\in\mathbb{R}$; hence, the only two choices are $(\mathbb{C},\overline{(-)},\pm 1)$ up to isomorphism.

In summary, there are two isomorphism classes of $\mathbb{Z}_2$-graded division algebras over the field $\mathbb{C}$, which are $\mathbb{C}$ and $(\mathbb{C},\mathrm{id},1)$. 
There are ten isomorphism classes of $\mathbb{Z}_2$-graded division algebras over $\mathbb{R}$, 
which are $\mathbb{R}$, $\mathbb{C}$, $\mathbb{H}$, $(\mathbb{R},\mathrm{id},1)$, $(\mathbb{R},\mathrm{id},-1)$, $(\mathbb{C},\mathrm{id},1)$, $(\mathbb{C},\overline{(-)},1)$, $(\mathbb{C},\overline{(-)},-1)$, $(\mathbb{H},\mathrm{id},1)$, and $(\mathbb{H},\mathrm{id},-1)$.

\subsubsection{\texorpdfstring{$\mathbb{Z}_2^N$-graded division algebra}{Z2N-graded division algebra}}\label{subsubsection:classification of Z2N-graded division algebra}

We classify $\mathbb{Z}_2^N$-graded division algebras over $\mathbb{R}$ and $\mathbb{C}$. Let $D$ be a $\mathbb{Z}_2^N$-graded division algebra such that $D_c\not = 0$ for each $c\in \mathbb{Z}_2^N$. Notice that $D_i\cong D_0\oplus D_{z_i}$ is a $\mathbb{Z}_2$-graded division algebra,  where $z_i\in \mathbb{Z}_2^N$ is the $i$th generator of $\mathbb{Z}_2^N$. We denote the canonical grading 1 element in $D_{z_i}$ by $\theta_i$; therefore, we have $D_i\cong (D_0,f_i,a_i)$.

It remains to determine the commutation relation between $\theta_i$ and $\theta_j$ for $i\not = j$. Here, we take $a_{ij} = (\theta_i\theta_j)^2$, which encodes the commutation relation $\theta_i\theta_j=b_{ij}\theta_j\theta_i$, where $b_{ij}=a_{ij}a_j^{-1}f_j(a_i^{-1})$. The $\mathbb{Z}_2$-graded division algebra $D_0\oplus D_{z_i+z_j}$ is isomorphic to $(D_0,f_i\circ f_j,a_{ij})$ if we choose $\theta_i\theta_j$ to be the canonical grading 1 element in $D_{z_i+z_j}$.
\begin{lem}
    If $D_0$ is a central simple algebra over $\mathbb{K}$, according to \cref{subsubsection:classification of Z2-graded division algebra}, $\theta_i$ can be chosen such that $f_i=\mathrm{id}$ and $a_i=\pm 1$ for any $i$. Then, we have $a_{ij}\in \mathbb{K}$ and $(a_{ij})^2=1$ for any $i$ and $j$.
\end{lem}
\begin{proof}
    We have
    \begin{equation}
        \theta_i\theta_j\theta_i=\theta_i b_{ij}^{-1}\theta_i\theta_j=f_i(b_{ij}^{-1})a_i\theta_j.
    \end{equation}
    On the other hand, we have
    \begin{equation}
        \theta_i\theta_j\theta_i=b_{ij}\theta_j a_i=b_{ij}f_j(a_i)\theta_j.
    \end{equation}
    Therefore, we have
    \begin{equation}\label{eq:the equation satisfied by bij}
        f_i(b_{ij}^{-1})a_i=b_{ij}f_j(a_i).
    \end{equation}
    Thus, we obtain $b_{ij}^2=1$, since $f_i=\mathrm{id}_{D_0}$. Using $b_{ij}=a_{ij}a_j^{-1}f_j(a_i^{-1})$ and $a_i=\pm 1$, we get $(a_{ij})^2 = 1$. Notice that $(f_i\circ f_j)^2(d)=a_{ij}da_{ij}^{-1}=d$ for any $d\in D_0$; hence, $a_{ij}\in \mathbb{K}$.
\end{proof}

For $\mathbb{Z}_2^N$-graded division algebras over $\mathbb{C}$, $D_0=\mathbb{C}$ is a central simple algebra. Therefore, the only choices of $a_{ij}$ are $\pm 1$. For different choices of $a_{ij}$, $D_0\oplus D_{z_i+z_j}$ are nonisomorphic. Hence, $D$ with different choices of $a_{ij}$ are nonisomorphic.

For $\mathbb{Z}_2^N$-graded division algebras over $\mathbb{R}$, $D_0=\mathbb{R}$ or $\mathbb{H}$ is a central simple algebra. Therefore, the only choices of $a_{ij}$ are $\pm 1$.

For $D_0=\mathbb{C}$, there are two choices of $f_i$, which are $\mathrm{id}$ and the complex conjugation. It reduces to the case of $\mathbb{Z}_2^N$-graded division algebras over $\mathbb{C}$ when all $f_i=\mathrm{id}$. For the cases in which only one $f_i$ equals $\overline{(-)}$, notice that \cref{eq:the equation satisfied by bij} is correct in general; therefore, we have $b_{ij}^2=a_{ij}^2=1$ if $f_i=\mathrm{id}$. Similar to \cref{eq:the equation satisfied by bij}, we have $f_j(b_{ij})a_j=b_{ij}^{-1}f_i(a_j)$, which is obtained by multiplying $\theta_j\theta_i\theta_i$ in two different ways. Thus, we obtain $b_{ij}^2=a_{ij}^2=1$ if $f_j=\mathrm{id}$. It remains to consider the case $f_i=\overline{(-)}$ and $f_j=\overline{(-)}$. In this case, we only have $|f_{ij}|^2=1$; therefore, $f_{ij}$ is a phase.

To further determine $f_{ij}$, we need to change the choice of basis $\{z_i\}$ of $\mathbb{Z}_2^N$. Without loss of generality, we start with $f_i=\overline{(-)}$ for the first $n$ generators $z_i$ of $\mathbb{Z}_2^N$, and $f_j=\mathrm{id}$ for the last $N-n$ generators $z_j$ of $\mathbb{Z}_2^N$. We take a new basis $z'_n=z_{n-1}+z_n$, and $z'_i=z_i$ for $i\not = n$. At the same time, we use $\theta'_n = \theta_{n-1}\theta_n$ to replace $\theta_n$, while the other $\theta_i$ remain unchanged. Now $f_i$ is $\overline{(-)}$ for the first $n-1$ generators of $\mathbb{Z}_2^N$, and $f_i$ is $\mathrm{id}$ for the last $N-n+1$ generators of $\mathbb{Z}_2^N$. If we repeat this several times, we find that only $f_1$ can be $\overline{(-)}$ and $f_i$ must be $\mathrm{id}$ for $i > 1$. Therefore, we never encounter the case where both $f_i$ and $f_j$ are $\overline{(-)}$ for $i\not = j$.

In summary, for a $\mathbb{Z}_2^N$-graded division algebra $D$ over $\mathbb{C}$, the classification is given by $(\{D_i\},\{a_{ij}\})$, where each $D_i\cong D_0\oplus D_{z_i}$ is a $\mathbb{Z}_2$-graded division algebra, and $a_{ij}=\pm 1$.

For a $\mathbb{Z}_2^N$-graded division algebra $D$ over $\mathbb{R}$ with $D_0=\mathbb{R}$ or $\mathbb{H}$, the classification is given by $(\{D_i\},\{a_{ij}\})$, where each $D_i\cong D_0\oplus D_{z_i}$ is a $\mathbb{Z}_2$-graded division $\mathbb{R}$ algebra, and $a_{ij}=\pm 1$. For $D_0=\mathbb{C}$ and all $f_i=\mathrm{id}$, the classification is given by $(\{D_i\},\{a_{ij}\})$, where $D_i\cong D_0\oplus D_{z_i}$ is a $\mathbb{Z}_2$-graded division $\mathbb{C}$-algebra, and $a_{ij}=\pm 1$. For the cases in which at least one $f_i$ equals $\overline{(-)}$, we need to choose a basis as explained above such that $f_1=\overline{(-)}$ and $f_i=\mathrm{id}$ for $i\not = 1$. The classification is given by $(\{D_i\},\{a_{ij}\})$, where $D_i\cong D_0\oplus D_{z_i}$ is a $\mathbb{Z}_2$-graded division algebra, and $a_{ij}=\pm 1$.

\begin{eg}
    We enumerate all isomorphism classes of $\mathbb{Z}_2\times \mathbb{Z}_2$-graded division algebras over $\mathbb{R}$. If two $\mathbb{Z}_2\times \mathbb{Z}_2$-graded division algebras are isomorphic after a change of grading, we count them as the same one. For $D_0=\mathbb{R}$ or $\mathbb{H}$, $D$ is determined by $a_1$, $a_2$ and $a_{12}$ equal to $\pm 1$. It gives 4 for $\mathbb{R}$ and 4 for $\mathbb{H}$. For $D_0=\mathbb{C}$ and $f_1=f_2=\mathrm{id}$, it gives 2. For $f_1=\overline{(-)}$ and $f_2=\mathrm{id}$, $a_2$ must be $1$, $a_1=\pm 1$ and $a_{12}=\pm 1$. Therefore, it gives 3. In total, the number is 13, as shown in \cite{kuznetsova10foldWay132023}.
\end{eg}

\subsection{Clifford algebra}\label{subsection:Clifford Algebra}

In this section, we briefly review the basic concepts of Clifford algebras; for more details, see Appendix \ref{subsection: Clifford Algebra in Appendix}.

For each vector space $V$ over $\mathbb{R}$ (resp. $\mathbb{C}$) with a quadratic form $Q:V\to \mathbb{R}$ (resp. $\mathbb{C}$), there is a Clifford algebra $Cl_Q$ (resp. $\mathbb{C}l_Q$). We mainly focus on nondegenerate quadratic forms. In this case, we can choose a basis of $V$ such that $Q$ is diagonal with diagonal entries $\pm 1$.

The algebra $Cl^{p,q}$ is the Clifford algebra corresponding to $\mathbb{R}^{p,q}$, where $\mathbb{R}^{p,q}$ is the vector space $\mathbb{R}^{p+q}$ with quadratic form having $p$ positive and $q$ negative diagonal entries. To be explicit, $Cl^{p,q}$ is an $\mathbb{R}$ algebra generated by $\{e_{1} ,\ldots,e_{p}, \tilde{e}_{1} ,\ldots,\tilde{e}_{q}\}$ that satisfy the relations $e_{i}^{2} =-1$, $\tilde{e}_{i}^{2} =1$, $e_{i} e_{j} +e_{j} e_{i} =\tilde{e}_{i}\tilde{e}_{j} +\tilde{e}_{j}\tilde{e}_{i} = e_{i}\tilde{e}_{j} +\tilde{e}_{j}e_{i} =0$.

Similarly, $\mathbb{C}l^{p,q}$ is the Clifford algebra corresponding to $\mathbb{C}^{p,q}$. $\mathbb{C}l^{p,q}$ is actually isomorphic to $\mathbb{C}l^{0,p+q}$ by taking another set of generators $\{ie_i, \tilde{e}_i\}$. We also denote $\mathbb{C}l^{0,p+q}$ by $\mathbb{C}l^{p+q}$. $\mathbb{C}l^{p,q}$ is a $\mathbb{C}$ algebra generated by the same sets of generators and relations as $Cl^{p,q}$.

$Cl^{p,q}$ and $\mathbb{C}l^{p,q}$ are semisimple algebras and $\mathbb{Z}_2$-graded semisimple algebras, as proved in \cref{thm:structure theorem of Clifford algebra} and \cref{thm:structure theorem of Clifford algebra as Z2-graded algebra}. The construction of simple modules and simple graded modules can be found in the proof where we construct the Wedderburn--Artin decomposition explicitly. The results of the Wedderburn--Artin decomposition are given in \cref{table:Clifford algebras}. The relation between Clifford algebras and graded division algebras is summarized in \cref{table:The isomorphisms between Z2-graded division algebras and Clifford algebras}. Using \cref{thm:structure theorem of Clifford algebra} and \cref{table:Clifford algebras}, it is straightforward to calculate $\mathcal{A}(Cl^{p,q})$, which is shown in \cref{table:A(Clpq)}.

\begin{table*}
    \centering
    \begin{tabular}{c | c c | c c}
        \hline
        $p$ & \multicolumn{2}{c}{$\mathbb{C}l^{p}$} & \multicolumn{2}{c}{} \\
        \hline
        0 & $\mathbb{C}$ & $\mathbb{C}$ &  &  \\
        1 & $\mathbb{C}\times\mathbb{C}$ & $\mathbb{C}l^1$ &  &  \\
        2 & $M_2(\mathbb{C})$ & $M_2(\mathbb{C})(\bar{c})$ &  &  \\
        3 & $M_2(\mathbb{C})\times M_2(\mathbb{C})$ & $M_2(\mathbb{C}l^1)$ &  &  \\
        \vdots & \vdots & \vdots &  &  \\
        \hline\hline
        $p$ or $q$ & \multicolumn{2}{c}{$Cl^{p,0}$} & \multicolumn{2}{c}{$Cl^{0,q}$} \\
        \hline
        0 & $\mathbb{R}$ & $\mathbb{R}$ & $\mathbb{R}$ & $\mathbb{R}$ \\
        1 & $\mathbb{C}$ & $Cl^{1,0}$ & $\mathbb{R}\times\mathbb{R}$ & $Cl^{0,1}$ \\
        2 & $\mathbb{H}$ & $Cl^{2,0}$ & $M_2(\mathbb{R})$ & $Cl^{0,2}$ \\
        3 & $\mathbb{H}\times\mathbb{H}$ & $Cl^{3,0}$ & $M_2(\mathbb{C})$ & $Cl^{0,3}$ \\
        4 & $M_2(\mathbb{H})$ & $M_2(\mathbb{H})(\bar{c})$ & $M_2(\mathbb{H})$ & $M_2(\mathbb{H})(\bar{c})$ \\
        5 & $M_4(\mathbb{C})$ & $M_2(Cl^{0,3})$ & $M_2(\mathbb{H})\times M_2(\mathbb{H})$ & $M_2(Cl^{3,0})$ \\
        6 & $M_8(\mathbb{R})$ & $M_4(Cl^{0,2})$ & $M_4(\mathbb{H})$ & $M_4(Cl^{2,0})$ \\
        7 & $M_8(\mathbb{R})\times M_8(\mathbb{R})$ & $M_8(Cl^{0,1})$ & $M_8(\mathbb{C})$ & $M_8(Cl^{1,0})$ \\
        8 & $M_{16}(\mathbb{R})$ & $M_{16}(\mathbb{R})(\bar{c})$ & $M_{16}(\mathbb{R})$ & $M_{16}(\mathbb{R})(\bar{c})$ \\
        9 & $M_{16}(\mathbb{C})$ & $M_{16}(Cl^{1,0})$ & $M_{16}(\mathbb{R}) \times M_{16}(\mathbb{R})$ & $M_{16}(Cl^{0,1})$ \\
        \vdots & \vdots & \vdots & \vdots & \vdots \\
        \hline
    \end{tabular}
    \caption{The Wedderburn--Artin decomposition of Clifford algebras. The left side shows the decomposition as nongraded algebras, while the right side shows the decomposition as $\mathbb{Z}_2$-graded algebras. $Cl^{p,0}$ and $Cl^{0,q}$ exhibit eight-periodicity, and $\mathbb{C}l^p$ exhibits two-periodicity. $\bar{c}$ indicates that half of the dimensions have a grading of 0, and the other half have a grading of 1. For example, for $M_2(\mathbb{C})(\bar{c})$, we have $\bar{c} = (0, 1)$.}\label{table:Clifford algebras}
\end{table*}

\begin{table}
    \centering
    \begin{tabular}{c | c }
            \hline
            $p$ & $\mathcal{A}(\mathbb{C}l^{p})$ \\
            \hline
            0 & $\mathbb{Z}$ \\
            1 & $0$ \\
            2 & $\mathbb{Z}$ \\
            \vdots & \vdots \\
            \hline\hline
            $q$ & $\mathcal{A}(Cl^{0,q})$ \\
            \hline
            0 & $\mathbb{Z}$ \\
            1 & $\mathbb{Z}_2$ \\
            2 & $\mathbb{Z}_2$ \\
            3 & $0$ \\
            4 & $\mathbb{Z}$ \\
            5 & $0$ \\
            6 & $0$ \\
            7 & $0$ \\
            8 & $\mathbb{Z}$\\
            \vdots & \vdots \\
            \hline
    \end{tabular}
    \caption{The value of $\mathcal{A}(Cl^{p,q})$.}\label{table:A(Clpq)}
\end{table}

\begin{table}
    \centering
    \begin{tabular}{|c|c|c|}
        \hline
        \multicolumn{2}{|c|}{$D$ over $\mathbb{C}$} & $M_{\left|\Gr/E\right|d}(D)(\bar{c})$ \\
        \hline
        $\mathbb{C}$ & $\mathbb{C}$ & $M_{d}(\mathbb{C}l^2)$ \\
        $(\mathbb{C},\mathrm{id},1)$ & $\mathbb{C}l^{1}$ & $M_d(\mathbb{C}l^1)$\\
        \hline
        \hline
        \multicolumn{2}{|c|}{$D$ over $\mathbb{R}$} & $M_{\left|\Gr/E\right|d}(D)(\bar{c})$ \\
        \hline
        $\mathbb{R}$ & $\mathbb{R}$ & $M_d(Cl^{1,1})$ \\
        $\mathbb{C}$ & $\mathbb{C}$ & $M_{d}(\mathbb{C}l^2)$ \\
        $\mathbb{H}$ & $\mathbb{H}$ & $M_d(Cl^{4,0})$ \\
        $(\mathbb{R},\mathrm{id},1)$ & $Cl^{0,1}$ & $M_d(Cl^{0,1})$ \\
        $(\mathbb{R},\mathrm{id},-1)$ & $Cl^{1,0}$ & $M_d(Cl^{1,0})$ \\
        $(\mathbb{C},\mathrm{id},1)$ & $\mathbb{C}l^{1}$ & $M_d(\mathbb{C}l^{1})$ \\
        $(\mathbb{C},\overline{(-)},1)$ & $Cl^{0,2}$ & $M_d(Cl^{0,2})$ \\
        $(\mathbb{C},\overline{(-)},-1)$ & $Cl^{2,0}$ & $M_d(Cl^{2,0})$ \\
        $(\mathbb{H},\mathrm{id},1)$ &  $Cl^{3,0}$ & $M_d(Cl^{3,0})$ \\
        $(\mathbb{H},\mathrm{id},-1)$ & $Cl^{0,3}$ & $M_d(Cl^{0,3})$ \\
        \hline
    \end{tabular}
    \caption{The first two columns are the isomorphisms between $\mathbb{Z}_2$-graded division algebras and Clifford algebras. The third column is the Wedderburn--Artin components of $\mathbb{Z}_2$-crossed product algebras as nongraded matrix algebras over Clifford algebras.}\label{table:The isomorphisms between Z2-graded division algebras and Clifford algebras}
\end{table}

%% file: 4.decomp_of_AV.tex
\section{\texorpdfstring{Dirac Hamiltonians and Decomposition of $\symalg^V$}{Dirac Hamiltonian and Decomposition of AV}}\label{section:decomposition of AV}

Section \ref{subsection:twisted group algebra with more gradings} introduces the $\mathbb{Z}_2^N$-graded algebra $\mathbb{C}[G,\omega,s]$ and its Wedderburn--Artin decomposition. Section \ref{subsection:Z2-Graded Algebra AV} explains how to perform the $\mathbb{Z}_2$-graded Wedderburn--Artin decomposition of $\symalg^V$. Sections \ref{subsection:constructing Dirac Hamiltonians} and \ref{subsection:topological invariant of Dirac Hamiltonians} respectively describe how to extract Dirac Hamiltonians and their topological invariants from the decomposition. Finally, \cref{subsection:use GAP to compute more complicated examples} presents a GAP package developed by the authors to compute the aforementioned decompositions.

\subsection{\texorpdfstring{$\mathbb{C}[G,\omega,s]$ with $\mathbb{Z}_2^N$-grading}{C[G,omega,s] with Z2N-grading}}\label{subsection:twisted group algebra with more gradings}

\subsubsection{The Wedderburn--Artin decomposition}\label{subsubsection:decomposition into matrix algebras C[G,omega,s] with more gradings}

Let $G$ be a $\mathbb{Z}_2^N$-graded group with grading map $G\to \mathbb{Z}_2^N$. Since $\mathbb{Z}_2$ is a field, $\mathbb{Z}_2^N$ is a vector space over $\mathbb{Z}_2$. We assume the grading map to be surjective; if it is not, we can always restrict $\mathbb{Z}_2^N$ to the image of the grading map. We also assume that $s:G\to \mathbb{Z}_2^N\to \mathbb{Z}_2$ is the grading of antilinear elements, i.e., the elements in $G_{c}$ are antilinear when $s(c)=1$. Since we assumed the grading map to be surjective, we can fix a set of group elements $g_{c}\in G_{c}$, for each $c\in\mathbb{Z}_2^N$.

We are going to decompose $\mathbb{C}[G,\omega,s]$ into $\mathbb{Z}_2^N$-graded matrix algebras. We assume that the Wedderburn--Artin decomposition of the subalgebra $\mathbb{C}[G_0,\omega]$ is already known, as constructed in \cref{subsection:Twisted Group Algebra}. That is, we have the isomorphism $\Phi_0=(\rho_{i_0})$,
\begin{equation}\label{eq:decomposition of C[G_0]}
    \mathbb{C}[G_0,\omega]\cong \prod_{i_0} M_{d_{i_0}}(\mathbb{C}),
\end{equation}
where $i_0$ takes values in the isomorphism classes of simple $\mathbb{C}[G_0,\omega]$-modules.

Let $V$ be a simple $\mathbb{C}[G_0,\omega]$-module and $\rho_0$ be its representation. We can construct a $\mathbb{Z}_2^N$-graded simple module $\mathbb{C}[G,\omega,s]\otimes_{\mathbb{C}[G_0,\omega]}V$. As a $\mathbb{C}[G_0,\omega]$-module, $\mathbb{C}[G,\omega,s]\otimes_{\mathbb{C}[G_0,\omega]} V$ is isomorphic to $\bigoplus_{c\in \mathbb{Z}_2^N} g_c \otimes_{\mathbb{C}[G_0,\omega]} V$, and thus to $\bigoplus_{c\in \mathbb{Z}_2^N} V^{g_c}$, where $V^{g_c}$ is the $\mathbb{C}[G_0,\omega]$-module given by
\begin{equation}
    g\cdot v = \nu(g_c, g, g_c) \rho_0(g_c^{-1}gg_c)(v),
\end{equation}
and $z\cdot v=\rho_0({}^{g_c}z)(v)$ for $z\in \mathbb{C}$ and $v\in V$. The phase factor $\nu(g_1, g_2, g_3)$ appears in the formula $(g_1)^{-1}\cdot g_2\cdot g_3 = \nu(g_1, g_2, g_3)g_1^{-1}g_2g_3$, to be explicit, we have
\begin{equation}
    \nu(g_1, g_2, g_3) := {}^{g_1}\omega(g_1,g_1^{-1})^{-1} \, {}^{g_1}\omega(g_2,g_{3})\omega(g_1^{-1},g_2g_{3}).
\end{equation}
A direct calculation shows that the graded representation $\rho$ is given by the block entries $\rho_{cc'}$ for $c,c'\in \mathbb{Z}_2^N$. For $z\in \mathbb{C}$, we have
\begin{equation}\label{eq:representation of C[G,omega,s] with more gradings complex number}
    \rho(z)_{cc'} = \delta_{c,c'}\rho_0({}^{g_c}z).
\end{equation}
For $g\in G_{d}$ with $d\in \mathbb{Z}_2^N$, we have
\begin{equation}\label{eq:representation of C[G,omega,s] with more gradings g linear}
    \rho(g)_{cc'}=
            \delta_{c,c'+d} \nu(g_c, g, g_{c'})\rho_0(g_c^{-1}gg_{c'}).
\end{equation}

Next, we determine the graded division algebra $D$ corresponding to the graded simple module. $D_0$ is always $\mathbb{C}$, and it acts on $\bigoplus_{c\in \mathbb{Z}_2^N} V$ diagonally by $\rho_0(z)$. As discussed in \cref{thm:structure theorem of Gr-crossed product algebra}, we first need to determine the subspace $E=\{c\in\mathbb{Z}_2^N | V^{g_c}\cong V\}$. Then, we choose a basis $\{z_i\}$ of $E$ such that the basis satisfies $s(z_1)=1$ and $s(z_i)=0$ for $i\neq 1$ when $s$ is nontrivial. Note that $D$ is determined by $\{D_i=D_0\oplus D_{z_i}\}$ and $\{a_{ij}\}$, as in \cref{subsubsection:classification of Z2N-graded division algebra}.

To determine $D_i$, we only need to construct an invertible element $\theta_i$ in $D_{z_i}$. Let $f_i$ be the isomorphism $V\cong V^{g_{z_i}}$. Notice that the map $\theta_i:V^{g_{c'}}\to V^{g_{c'+z_i}}$ is a $\mathbb{C}[G_0,\omega]$-map, and it must be induced by $(f_i)^{g_{c'}}$, i.e., the result of the action of the functor $(-)^{g_{c'}}$ on the morphism $f_i$. Therefore, the components of $\theta_i$ are given by
\begin{equation}\label{eq:action of theta_i}
    (\theta_i)_{cc'}=
            \delta_{c,c'+z_i} \nu(g_c, g_{c'}, g_{z_i})\rho_0(g_c^{-1}g_{c'}g_{z_i})f_i.
\end{equation}
Together with $D_0$, we can generate the whole $D$. $\theta_i^2$ is the identity up to a scalar, by Schur's lemma. For $i$ with $s(z_i) = 0$, we can multiply $f_i$ by a scalar such that $\theta_i^2 = 1$. For $i = 1$ with $s(z_1) = 1$, $\theta_1$ is antilinear; therefore, we can only multiply $f_1$ by a scalar such that $\theta_1^2 = \pm 1$. Finally, we have $a_{ij} = (\theta_i\theta_j)^2 \in \{\pm 1\}$ for $i\not =j$ according to \cref{subsection:classification of Z2N-graded division algebra}.

In practice, we can use the canonical antilinear map $K$ on $V$ if $V$ has a basis. For $s(c) = 1$, we take a new $\mathbb{C}[G_0,\omega]$-module $\bar{V}^{g_c}$ defined by
\begin{equation}
    g\cdot v = \bar{\nu}(g_c, g, g_c) \bar{\rho_0}(g_c^{-1}gg_c)(v),
\end{equation}
and $z\cdot v=\rho_0(z)(v)$ for $z\in \mathbb{C}$ and $v\in V$. It has the advantage that the complex structure becomes the original one on $V$, and $K:V^{g_c} \to \bar{V}^{g_c}$ is an isomorphism. Replacing $V^{g_c}$ with $\bar{V}^{g_c}$ when $s(c) = 1$, we get, for $g\in G_{d}$ with $d\in \mathbb{Z}_2^N$,
\begin{equation}
    \rho(g)_{cc'}=
            \delta_{c,c'+d} K^{s(c)}\nu(g_c, g, g_{c'})\rho_0(g_c^{-1}gg_{c'})K^{s(c')}.
\end{equation}
And the components of $\theta_i$ becomes
\begin{equation}
  \begin{aligned}
    (\theta_i)_{cc'} &=
        \delta_{c,c'+z_i} K^{s(c)} \nu(g_c, g_{c'}, g_{z_i}) \\
        & \phantom{{}={}} \quad \cdot \rho_0(g_c^{-1}g_{c'}g_{z_i}) K^{s(z_i)} f_i K^{s(c')},
  \end{aligned}
\end{equation}
where $f_i: V\to \bar{V}^{g_{z_i}}$ is an isomorphism.

The $D$-structure on the $\mathbb{Z}_2^N$-graded module can be used to extract a $\mathbb{Z}_2^N$-graded Wedderburn--Artin component with coefficient in $D^{\mathrm{op}}$, as introduced in \cref{thm:structure theorem of Gr-crossed product algebra}. We can construct one such component for each equivalence class of simple $\mathbb{C}[G_0,\omega]$-modules under the equivalence relation given by $V\cong W^{g_c}$ for some $c\in \mathbb{Z}_2^N$.

Therefore, we get the Wedderburn--Artin decomposition of $\mathbb{C}[G,\omega,s]$ as a $\mathbb{Z}_2^N$-graded algebra, which is given by the isomorphism $\Phi=(\rho_{i})$,
\begin{equation}\label{eq:decomposition Phi of graded group algebra with antilinear element with more grading}
    \mathbb{C}[G,\omega,s]\cong \prod_{i} M_{d_{i}'}(D_i)(\bar{c}_{i}),
\end{equation}
where $i$ takes values in the equivalence classes of simple $\mathbb{Z}_2^N$-graded $\mathbb{C}[G,\omega,s]$-modules; the equivalence relation is given by isomorphism up to a grading shift. We have $d_{i}' = |\Gr/E_{i}|d_{i_{0}} = d_{i}/\mathrm{dim}(D_i)$, where $d_{i}$ is the dimension of the simple $\mathbb{Z}_2^N$-graded $\mathbb{C}[G,\omega,s]$-module.

\subsubsection{\texorpdfstring{The inverse of $\Phi$}{The inverse of Phi}}\label{subsubsection:inverse of Phi with antilinear element with more grading}

For $f\in M_{d_{i}'}(D_i)(\bar{c}_{i})$, we have
\begin{widetext}
\begin{equation}\label{eq:inverse of Phi Z2N-graded twisted group algebra with antilinear elements}
    \Phi^{-1}(f) = \frac{d_i}{2|G|}\sum_{g\in G}\omega(g^{-1},g)^{-1}\left( \mathrm{Tr}\left(\mathrm{Re}\left(\rho_i(g^{-1})f\right)\right) - \mathrm{Tr}\left(\mathrm{Re}\left(\rho_i(ig^{-1})f\right) \right) i \right)\, g ,
\end{equation}
\end{widetext}
where $\mathrm{Re}$ is the projection from $D$ to the subspace $\mathbb{R}$ spanned by the unit. This formula is proved in \cref{subsubsection:inverse of Phi twisted group algebra with antilinear elements}.

\subsubsection{Character theory}\label{subsubsection:character theory of graded group algebra with antilinear element with more grading}

We first determine $E$ by characters. We need to calculate
\begin{equation}
    \frac{1}{|G_0|}\sum_{g\in G_0}\omega(g,g^{-1})^{-1}\chi^{g_c}(g^{-1})\chi(g)
\end{equation}
to determine whether $V\cong V^{g_c}$ for all $c\in \mathbb{Z}_2^N$, where $\chi^{g_c}$ is the character of $V^{g_c}$.

To determine $D$ by characters, we only need to determine $D_i$ and $a_{ij}$. For $i$ with $s(z_i) = 0$, $D_i\cong \mathbb{C}l^1$. When $s$ is nontrivial, $D_1$ is determined by $w_1$ as in \cref{subsubsection:character theory of twisted group algebra with antilinear elements},
\begin{equation}
    w_1=\frac{1}{|G_0|}\sum_{g\in G_{z_1}} \omega(g,g)\chi(g^2),
\end{equation}
where $\chi$ is the character of $V$. If $w_1=1$, then $D_1 = Cl^{0,2}$. If $w_1=-1$, then $D_1 = Cl^{2,0}$. Notice that $a_{ij}$ is determined by the $\mathbb{Z}_2$-graded division algebra $D_0\oplus D_{z_i+z_j}$; similarly, $a_{ij}$ is determined by
\begin{equation}
    w_{ij}=\frac{1}{|G_0|}\sum_{g\in G_{z_i+z_j}} \omega(g,g)\chi(g^2).
\end{equation}
If $w_{ij}=1$, then $a_{ij}=1$. If $w_{ij}=-1$, then $a_{ij}=-1$.

\subsection{\texorpdfstring{$\mathbb{Z}_2$-graded algebra $\symalg^V$}{Z2-graded algebra AV}}\label{subsection:Z2-Graded Algebra AV}

\subsubsection{The Wedderburn--Artin decomposition}\label{subsubsection:WA-decomp of AV}

Let us recall the definitions of the $\mathrm{Pin}$ and $\mathrm{Spin}$ groups \cite{karoubiKtheoryIntroduction2009}.
\begin{prop}\label{prop of Pin group}
    Let $(V,Q)$ be a finite-dimensional vector space over $\mathbb{R}$ equipped with a nondegenerate quadratic form. There are short exact sequences of groups,
    \begin{equation*}
        1\to \mathbb{Z}_2\to \mathrm{Pin}_Q\to O_Q\to 1
    \end{equation*}
    and
    \begin{equation*}
        1\to \mathbb{Z}_2\to \mathrm{Spin}_Q\to SO_Q\to 1.
    \end{equation*}
    where $\mathrm{Pin}_Q$ and $\mathrm{Spin}_Q$ are subgroups of the unit group of $Cl_Q$. $\mathrm{Pin}_Q$ is generated by elements $x=v_1 v_2 \dots v_k$ with $v_i \in V$ and $Q(v_i) \neq 0$, satisfying the normalization $|N(x)|=1$, where $N$ is the spinorial norm. $\mathrm{Spin}_Q$ consists of those elements in $\mathrm{Pin}_Q$ generated by an even number of vectors.

    The homomorphism from $\mathrm{Pin}_Q$ (resp. $\mathrm{Spin}_Q$) to $O_Q$ (resp. $SO_Q$) is defined by $x\mapsto (-1)^{deg(x)}x(-)x^{-1}$.
\end{prop}

\begin{proof}
    The action of a vector $v\in V \subset \mathrm{Pin}_Q$ on $V$ is given by $v \mapsto -uvu^{-1}$, which is a reflection across the hyperplane orthogonal to $u$. Since $O_Q$ is generated by reflections (and $SO_Q$ by an even number of reflections), the maps $\mathrm{Pin}_Q \to O_Q$ and $\mathrm{Spin}_Q \to SO_Q$ are surjective. The kernel $\mathbb{Z}_2$ of these maps is $\{\pm 1\}$.
\end{proof}

Let $\rho:G\to O_Q$ be a group homomorphism; we can construct a diagram of exact sequences,
\begin{equation}\label{equation of extension}
    \begin{tikzcd}
        1 \ar[r] & \mathbb{Z}_2 \ar[r]\ar[d] & G_{\omega_Q} \ar[r,"\pi"]\ar[d,"\tilde{\rho}"] & G \ar[l,dotted,bend left,"f"]\ar[r]\ar[d,"\rho"] & 1\\
        1 \ar[r] & \mathbb{Z}_2 \ar[r] & \mathrm{Pin}_Q \ar[r] & O_Q \ar[r] & 1
    \end{tikzcd}
\end{equation}
The right square is a pullback; $G_{\omega_Q}$ is the group extension of $G$ and $\mathbb{Z}_2$ by the two-cocycle $\omega_Q$. The following theorem on ``separation of variables'' plays a central role in the construction of the Wedderburn--Artin decomposition of $\symalg^V$. Here, we slightly generalize the definition of $\symalg^V$ by allowing $V$ to be equipped with a general quadratic form $Q$. The orthogonal real representation $V$ is generalized to $G \to O_Q$, yielding the $\mathbb{Z}_2$-graded algebra $\symalg^{(V,Q)}$. When no confusion arises, this is also denoted as $\symalg^{Q}$ or $\symalg^{V}$.

\begin{thm}[\cite{cornfeldTenfoldTopologyCrystals2021,shiozakiClassificationSurfaceStates2022}]\label{thm:separation of variable}
    Let $(V,Q)$ be a finite-dimensional vector space over $\mathbb{R}$ equipped with a quadratic form and a $G$ action given by $\rho :G\to O_Q$, $\tilde{\rho}:G_{\omega_Q}\to \mathrm{Pin}_Q$ be a lifting of $\rho$, and $f$ be the section of $\pi$ as in \cref{equation of extension}. Then, we have a $\mathbb{Z}_2$-graded isomorphism 
    \begin{equation}
        \symalg^{(V,Q)}\cong Cl_Q\hat{\otimes}_{\mathbb{R}}\mathbb{C}[G,\omega',s],
    \end{equation} 
    where $\omega'$ is a two-cocycle over the coefficient $\mathbb{R}/\mathbb{Z}$ with
    \begin{equation*}
        \omega'(g_1,g_2) = \omega(g_1,g_2) + \frac{1}{2}\omega_Q(g_1,g_2) + \frac{1}{2}(h+o)\cup o(g_1, g_2),
    \end{equation*}
    The $\mathbb{Z}_2$ grading of $\mathbb{C}[G,\omega',s]$ is given by $h+o$, where $o$ is the one-cocycle given by $\mathrm{det}(g) = (-1)^{o(g)}$.
\end{thm}

\begin{proof}

    We construct the isomorphism explicitly by $g\mapsto \tilde{\rho}\left(f\left(g\right)\right)\hat{\otimes}g$ and $e_i\mapsto e_i\hat{\otimes}1$. The map is obviously a bijection, and it is an algebra homomorphism according to the following calculation:
\begin{align*}
        & \phantom{{}={}} \left(\tilde{\rho}\left(f\left(g_1\right)\right)\hat{\otimes}g_1\right)\cdot \left(\tilde{\rho}\left(f\left(g_2\right)\right)\hat{\otimes}g_2\right) \\ 
        & = \left(-1\right)^{ \left(h\left(g_1\right)+o\left(g_1\right)\right)o\left(g_2\right) }
        \tilde{\rho}\left(f\left(g_1\right)f\left(g_2\right)\right)\hat{\otimes} e^{ i2\pi\omega'\left(g_1,g_2\right) } g_1g_2 \\
        & = \left(-1\right)^{ \left( h\left(g_1\right) + o\left(g_1\right) \right) o\left(g_2\right) + \omega_Q\left(g_1,g_2\right) } e^{ i2\pi\omega'\left(g_1,g_2\right) } \\
        & \phantom{{}={}} \quad \cdot \tilde{\rho}\left(f\left(g_1g_2\right)\right) \hat{\otimes} g_1g_2 \\
        & = e^{ i2\pi\omega\left(g_1,g_2\right) } \tilde{\rho}\left(f\left(g_1g_2\right)\right)\hat{\otimes}g_1g_2
\end{align*}
\begin{align*}
        & \phantom{{}={}} \left(\tilde{\rho}\left(f\left(g\right)\right)\hat{\otimes}g\right)\cdot \left(v\hat{\otimes} 1\right) \cdot \left(\tilde{\rho}\left(f\left(g\right)\right)\hat{\otimes}g\right)^{-1} \\
        & = \left(\tilde{\rho}\left(f\left(g\right)\right)\hat{\otimes}1\right) \cdot \left(1\hat{\otimes}g\right) \cdot \left(v\hat{\otimes}1\right) \cdot \left(1\hat{\otimes}g\right)^{-1} \\
        & \phantom{{}={}} \quad \cdot \left(\tilde{\rho}\left(f\left(g\right)\right)^{-1}\hat{\otimes}1\right) \\
        & = (-1)^{(h+o)(g)} \tilde{\rho}\left(f\left(g\right)\right)v\tilde{\rho}\left(f\left(g\right)\right)^{-1}\hat{\otimes}1 \\
        & = (-1)^{h(g)}\rho\left(g\right)\left(v\right)\hat{\otimes}1 .
\end{align*}
We used the following equations in the above calculation:
\begin{align}
    f(g_1)f(g_2) & = (-1)^{\omega_Q(g_1,g_2)} \cdot f(g_1g_2) ,
\end{align}
where $-1$ is considered as an element in $\mathbb{Z}_2\hookrightarrow G_{\omega_Q}$.
\end{proof}

$\mathbb{C}[G,\omega',s]$ is a $\mathbb{Z}_2\times \mathbb{Z}_2$-graded algebra, where the grading is given by $(s,h+o)$. We first obtain the Wedderburn--Artin decomposition of $\mathbb{C}[G,\omega',s]$ as a $\mathbb{Z}_2\times \mathbb{Z}_2$-graded algebra using the method described in \cref{subsection:twisted group algebra with more gradings}. We then forget the grading corresponding to $s$ to obtain the Wedderburn--Artin decomposition as a $\mathbb{Z}_2$-graded algebra, as treated in Appendix \ref{subsubsection:grading forgetting}. Finally, taking the graded tensor product with $Cl_Q$ yields the Wedderburn--Artin decomposition of $\symalg^V$ as a $\mathbb{Z}_2$-graded algebra (see Appendix \ref{subsubsection:Graded Tensor Product of Z2N-Graded Algebras}).

\subsubsection{\texorpdfstring{$H^2(BO_Q,\mathbb{Z}_2)$}{H2(BOQ,Z2)}}\label{subsubsection:H^2(BO_Q,Z_2)}

In this section, we provide supplementary details regarding $\omega_Q$. Similar to ordinary group extensions, the group extension of $O_Q$ by $\mathbb{Z}_2$ is classified by $\omega_Q\in H^2(BO_Q,\mathbb{Z}_2)$. The cocycle used in the extension of $G$ is the image of $\omega_Q$ under the action of $\rho_V^{*}$, and this image is likewise denoted as $\omega_Q$.

It is well known \cite{brownCohomology$BSO_n$$BO_n$1982} that 
\begin{equation}
    H^*(BO(n),\mathbb{Z}_2)\cong \mathbb{Z}_2[w_1,\dots,w_n],
\end{equation}
where $w_i\in H^i(BO(n),\mathbb{Z}_2)$ are called Stiefel--Whitney classes. $\mathrm{Pin}_+(n)$ is the extension by $w_2$. This result can be generalized to $O(p,q)$ \cite{trautmanDoubleCoversPseudoorthogonal2001}.

\begin{thm}
    \begin{equation}
        H^*(BO(p,q),\mathbb{Z}_2)\cong \mathbb{Z}_2[w_1^-,\dots,w_p^-] \otimes_{\mathbb{Z}_2} \mathbb{Z}_2[w_1^+,\dots,w_q^+].
    \end{equation}
\end{thm}

\begin{proof}
    $O(p,q)$ is a subgroup of $GL(p+q)$. $GL(p+q)$ has a homotopy retraction to $O(p+q)$. Restricting this homotopy retraction to $O(p,q)$, we obtain a homotopy retraction from $O(p,q)$ to $O(p)\times O(q)$. The result then follows from the Künneth formula.
\end{proof}

\begin{cor}\label{cor:1 of sw class}
    $H^2(BO(p,q),\mathbb{Z}_2)$ is generated by $w_2^+$, $w_2^-$, $w_1^+w_1^+$, $w_1^-w_1^-$, and $w_1^+w_1^-$, where the symbols for the cup product and tensor product are omitted. These cohomology classes characterize the algebraic relations satisfied by the generators of $O(p,q)$, i.e., the algebraic relations of reflections along a vector. Let $v,v'\in \mathbb{R}^q\subset \mathbb{R}^{p,q}$ and $u,u'\in \mathbb{R}^p\subset \mathbb{R}^{p,q}$, where $v$ is orthogonal to $v'$ and $u$ is orthogonal to $u'$.
    \begin{enumerate}
        \item $w_2^+$ (or $w_2^-$) characterizes the (anti)commutativity of reflections along $v$ and $v'$ (or $u$ and $u'$).
        \item $w_1^+w_1^+$ (or $w_1^-w_1^-$) characterizes the fact that the square of the reflection along $v$ (or $u$) equals $\pm 1$.
        \item $w_1^+w_1^-$ characterizes the (anti)commutativity of reflections along $v$ and $u$.
        \item $w_1^+$ (or $w_1^-$) is the one-cocycle determined by $\mathrm{det}:O(q)\to O(p,q)\to \mathbb{Z}_2$ (or $\mathrm{det}:O(p)\to O(p,q)\to \mathbb{Z}_2$).
    \end{enumerate}

    In particular, the group extension of $\mathrm{Pin}(p,q)$ is given by $w_2^+ + w_2^- + w_1^+w_1^- + w_1^-w_1^-$.
\end{cor}

\begin{cor}\label{cor:2 of sw class}
    There is an obvious isomorphism $O(p,q)\cong O(q,p)$; hence, there is an isomorphism 
    \begin{equation}
        H^*(BO(p,q),\mathbb{Z}_2)\cong H^*(BO(q,p),\mathbb{Z}_2).
    \end{equation}
    $\mathrm{Pin}(q,p)$ is also a group extension of $O(p,q)$, and the two-cocycle corresponding to this group extension is the element $w_2^++w_2^-+w_1^+w_1^-+w_1^+w_1^+$ in $H^*(BO(p,q),\mathbb{Z}_2)$. In particular, the difference between the two cocycles used to extend $\mathrm{Pin}(q,p)$ and $\mathrm{Pin}(p,q)$ is $w_1^+w_1^++w_1^-w_1^-$.
\end{cor}

\subsection{Constructing Dirac Hamiltonians}\label{subsection:constructing Dirac Hamiltonians}

In this section, we explain how to construct Dirac Hamiltonians from the elements in the algebraic classification $\mathcal{A}(\symalg^V)$. In fact, we only need to write down the Dirac Hamiltonians corresponding to the generators of $\mathcal{A}(\symalg^V)$, since the Hamiltonians corresponding to other elements can be obtained via direct sums. The generators of $\mathcal{A}(\symalg^V)$ correspond to the $\mathbb{Z}_2$-graded simple modules of $\symalg^V$, or the ungraded simple modules of $Cl^{0,1}\hat{\otimes}\symalg^V$. Therefore, it suffices to construct the simple modules from each component of the Wedderburn--Artin decomposition. Finally, we use \cref{eq:Dirac Hamiltonian} to write down the Dirac Hamiltonians from these simple modules. Examples of Dirac Hamiltonians are given in \cref{subsection:use GAP to compute more complicated examples}.

\subsection{Topological invariants of Dirac Hamiltonians}\label{subsection:topological invariant of Dirac Hamiltonians}

Consider a given Dirac Hamiltonian [see \cref{eq:Dirac Hamiltonian}] whose mass term satisfies $M = - M_0$. From this, we can extract a $\mathbb{Z}_2$-graded module $N$ of $\symalg^V$. Let $\rho_N$ denote its representation, where $M$ is regarded as the $\mathbb{Z}_2$ grading. To determine which element in $\mathcal{A}(\symalg^V)$ it belongs to, we consider two cases.

The first case is the $\mathbb{Z}_2$ topological invariant, which corresponds to a $\mathbb{Z}_2$ component in $\mathcal{A}(\symalg^V)$. Given the decomposition
\begin{equation}
    \Phi: \mathbb{C}[G,\omega,s] \cong \prod_{i} M_{d_{i}'}(D_i)(\bar{c}_{i}),
\end{equation}
suppose the $\mathbb{Z}_2$ component corresponds to the index $i$. We map the identity element of the matrix algebra of component $i$ back to $\mathbb{C}[G,\omega,s]$ via $\Phi^{-1}$ given in \cref{eq:inverse of Phi Z2N-graded twisted group algebra with antilinear elements}, and then map it to $\symalg^V$ via \cref{thm:separation of variable}, denoting the resulting element as $a_i$. The components of $\mathbb{C}[G,\omega,s]$ are in one-to-one correspondence with the components of $\symalg^V$. Note that $\rho_N(a_i)$ acts as the identity only on the submodule of $V$ corresponding to component $i$, and acts as zero on all other submodules. Assuming the dimension of the graded simple module of $\symalg^V$ corresponding to component $i$ is $n_i$, then $\frac{1}{n_i}\mathrm{Tr}_{\mathbb{R}}(\rho_N(a_i)) \pmod{2}$ is exactly the $\mathbb{Z}_2$ topological invariant. Note that the trace here requires treating the representation as a real representation.

The other case is the $\mathbb{Z}$ topological invariant, which corresponds to a $\mathbb{Z}$ component in $\mathcal{A}(\symalg^V)$ and to the component $i$ in the decomposition of $\mathbb{C}[G,\omega,s]$. If the dimension of the system is $q$, then the Clifford algebra part of $\symalg^V$ is $Cl^{0,q}$, and we take an element $\varepsilon_1 = \tilde{e}_1\tilde{e}_2\dots \tilde{e}_q$ in it. For $M_{d_{i}'}(D_i)(\bar{c}_{i})$, when $D_i$ is $\mathbb{R}$, $\mathbb{C}$, or $\mathbb{H}$, we take a diagonal matrix, map it back to $\symalg^V$, and denote it as $\varepsilon_2$. This diagonal matrix is chosen such that for $\bar{c}_{i} = (c_{i1}, c_{i2}, \dots)$, if $c_{ij} = 0$, the $j$th diagonal element is $1$; if $c_{ij} = 1$, the $j$th diagonal element is $-1$. For $D_i = \mathbb{C}$, if $\varepsilon_1^2 = -1$, the diagonal elements are changed to $\pm i$. When $D_i$ is one of the remaining seven $\mathbb{Z}_2$-graded division algebras, since they are all Clifford algebras, we take a diagonal matrix whose diagonal elements are $e_1e_2\dots \tilde{e}_1\tilde{e}_2\dots$, map it back to $\symalg^V$, and denote it as $\varepsilon_2$. Assuming the dimension of the graded simple module of $\symalg^V$ corresponding to component $i$ is $n_i$, then $\frac{1}{n_i}\mathrm{Tr}_{\mathbb{R}}\left(M\rho_N(\varepsilon_1\varepsilon_2)\right)$ is exactly the $\mathbb{Z}$ topological invariant.

\subsection{Computing more complicated examples using GAP}\label{subsection:use GAP to compute more complicated examples}

When computing more complicated examples, manual calculation is extremely cumbersome; therefore, a GAP package was developed by the authors \cite{ttiiddeeTtiiddeeFreeFermionicSPT2026}. GAP \cite{GAP4} is a computer algebra system specialized in group theory, capable of computing irreducible characters of finite groups and the irreducible complex representations corresponding to these characters. Consequently, the computational method introduced in this work assumes that all irreducible complex representations of the group are known, and subsequent calculations are performed on this basis.

The computational method introduced herein requires taking square roots and solving norm equations. Since GAP primarily represents numbers via cyclotomic fields and lacks real or complex number types, and since cyclotomic fields are not closed under taking square roots or solving norm equations, GAP cannot perform these operations on cyclotomic fields. Therefore, PARI/GP \cite{PARI2} was also utilized. PARI/GP is a computer algebra system specialized in number theory calculations. GAP provides a prototype for its interface, allowing its functionality to be integrated into GAP.

In the package developed by the authors, the defining information of $\symalg^V$, namely $G$, $\omega$, $s$, $h$, and $V$, is required as input to obtain the Wedderburn--Artin decomposition of $\mathbb{C}[G,\omega',s]$. The results of some examples are presented below.

\subsubsection{\texorpdfstring{Example: Spin-$\frac{1}{2}$ system with $S_2\times \mathbb{Z}_2^{\hat{C}}$ symmetry, where $\hat{C}^2 = 1$}{Example: Spin-1/2 system with S2 x Z2C symmetry, where C2 = 1}}\label{subsubsection:example Spin-1/2 S2 x Z2C C2 = 1}

When the symmetry is the spin-$\frac{1}{2}$ $G = S_2\times \mathbb{Z}_2^{\hat{C}}$, we have $s(\hat{C})=h(\hat{C})=1$, and their values are $0$ on the remaining group elements. $\hat{s}_2$ has a natural action on the three-dimensional space $V$, and $\hat{C}$ acts as a reflection on $V$. The two-cochain $\omega$ on $S_2$ is determined by the spin-$\frac{1}{2}$ nature, which is given by the group extension $Pin_{-}(3)$ of $O(3)$. The $\omega$ on $\mathbb{Z}_2^{\hat{C}}$ is determined by $\hat{C}^2=1$. Furthermore, since the elements in $S_2$ commute with $\hat{C}$, we have $\omega(\hat{s}_2, \hat{C})=\omega(\hat{C}, \hat{s}_2)=0$. The decomposition of $\mathbb{C}[G,\omega',s]$ is given by
\begin{equation}
    \mathbb{C}[G,\omega',s] \cong Cl^{0,3} \textrm{,}
\end{equation}
with the mappings $s_2\mapsto \tilde{e}_1$, $C\mapsto \tilde{e}_3\tilde{e}_1$, and $i\mapsto \tilde{e}_2\tilde{e}_3$. Furthermore, we can obtain the decomposition of $\symalg^V$ as
\begin{equation}
    \symalg^V \cong Cl^{0,6} \cong M_4(Cl^{2,0}) \textrm{,}
\end{equation}
Therefore, according to \cref{table:A(Clpq)}, its classification is $0$. Although its classification is trivial, we can still construct a trivial Dirac Hamiltonian from this decomposition, which is given by \cref{eq:Dirac Hamiltonian} with
\begin{equation}
    \begin{aligned}
        \tilde{\Gamma}_1 & = \sigma_z \otimes \sigma_y \otimes 1 \\
        \tilde{\Gamma}_2 & = \sigma_x \otimes \sigma_y \otimes 1 \\
        \tilde{\Gamma}_3 & = \sigma_y \otimes \sigma_y \otimes 1 \\
        \hat{s}_2 & = 1 \otimes \sigma_z \otimes \sigma_y \\
        \hat{C} & = (\sigma_y \otimes \sigma_y \otimes 1)K \\
        M = - M_0 & = 1 \otimes \sigma_z \otimes 1 \\
    \end{aligned}
\end{equation}
where $K$ is the complex conjugation operator. Substituting these into \cref{eq:Dirac Hamiltonian} yields the Dirac Hamiltonian.

\subsubsection{\texorpdfstring{Example: Spin-$\frac{1}{2}$ system with $S_2\times \mathbb{Z}_2^{\hat{C}}$ symmetry and nontrivial superconducting pairing, where $\hat{C}^2 = 1$}{Example: Spin-1/2 system with S2 x Z2C symmetry with nontrivial superconducting pairing, where C2 = 1}}\label{subsubsection:example Spin-1/2 S2 x Z2C with nontrivial pairing C2 = 1}

When the gap of a topological superconductor features a nontrivial superconducting pairing, meaning the gap function carries a one-dimensional representation, the two-cocycle $\omega$ of $G$ no longer satisfies $\omega(\hat{g}, \hat{C}) = \omega(\hat{C}, \hat{g}) = 0$ for $\hat{g}\in H$, where $G$ is $H\times \mathbb{Z}_2^{\hat{C}}$, and $H$ contains no antilinear symmetries. The rule for the modification of $\omega$ is given in Ref.~\cite{shiozakiClassificationSurfaceStates2022}. If the one-cochain corresponding to the one-dimensional representation on the gap function is denoted by $\theta$, we have
\begin{equation}
    (1, \hat{C}) \cdot (\hat{g}, 1) = e^{i \pi \theta(\hat{g})} (\hat{g}, \hat{C})\textrm{.}
\end{equation}
When the symmetry is the spin-$\frac{1}{2}$ $G = S_2\times \mathbb{Z}_2^{\hat{C}}$ with nontrivial superconducting pairing, the decomposition of $\mathbb{C}[G,\omega',s]$ is
\begin{equation}
    \mathbb{C}[G,\omega',s] \cong Cl^{3,0} \textrm{,}
\end{equation}
with $s_2\mapsto e_1e_2e_3$, $C\mapsto e_3e_1$, and $i\mapsto e_2e_3$. Furthermore, the decomposition of $\symalg^V$ is given by
\begin{equation}
    \symalg^V \cong Cl^{3,3} \cong M_8(\mathbb{R})(\bar{c}) \textrm{,}
\end{equation}
where the first four components of $\bar{c}$ are $0$, and the last four components are $1$. Therefore, according to \cref{table:A(Clpq)}, its classification is $\mathbb{Z}$. The Dirac Hamiltonian corresponding to the generator of $\mathbb{Z}$ is given by \cref{eq:Dirac Hamiltonian}, with
\begin{equation}
    \begin{aligned}
        \tilde{\Gamma}_1 & = \sigma_x \otimes 1 \\
        \tilde{\Gamma}_2 & = \sigma_y \otimes \sigma_z \\
        \tilde{\Gamma}_3 & = \sigma_y \otimes \sigma_x \\
        \hat{s}_2 & = \sigma_z \otimes 1 \\
        \hat{C} & = (\sigma_x \otimes 1) K \\
        M = - M_0 & = \sigma_z \otimes 1 \\
    \end{aligned}
\end{equation}
Substituting these into \cref{eq:Dirac Hamiltonian} yields the Dirac Hamiltonian.

\subsubsection{\texorpdfstring{Example: Spin-$\frac{1}{2}$ system with $D_{3d}\times \mathbb{Z}_2^{\hat{C}} \times \mathbb{Z}_2^{\hat{T}}$ symmetry, where $\hat{C}^2 = 1$ and $\hat{T}^2 = -1$}{Example: Spin-1/2 system with D3d x Z2C x Z2T symmetry, where C2 = 1 and T2 = -1}}\label{subsubsection:example Spin-1/2 D3d x Z2C x Z2T C2 = 1 T2 = -1}

Consider the spin-$\frac{1}{2}$ symmetry $G = D_{3d}\times \mathbb{Z}_2^{\hat{C}} \times \mathbb{Z}_2^{\hat{T}}$. We have $s(\hat{T}) = s(\hat{C}) = 1$, and the values are $0$ on the remaining group elements; we have $h(\hat{C}) = 1$, and the values are $0$ on the remaining group elements. $D_{3d}$ has a natural action on the three-dimensional space $V$, while $\hat{T}$ and $\hat{C}$ act as reflections on $V$. The two-cochain $\omega$ is determined by three parts: $\omega$ on $D_{3d}$ is determined by the spin-$\frac{1}{2}$ nature, given by the group extension $Pin_{-}(3)$ of $O(3)$; $\omega$ on $\mathbb{Z}_2^{\hat{C}}$ is determined by $\hat{C}^2 = 1$; and $\omega$ on $\mathbb{Z}_2^{\hat{T}}$ is determined by $\hat{T}^2 = -1$. In addition, since these three parts commute with each other, $\omega(\hat{g}_1,\hat{g}_2) = 0$ if $\hat{g}_1$ and $\hat{g}_2$ belong to different components. $D_{3d}$ is generated by two generators: the rotoreflection symmetry $s_6$ along the $z$ axis and the reflection symmetry $m_y$ across the $y$ axis. Therefore, each decomposition component $\rho_i$ can be given by the images of the four generators $s_6$, $m_y$, $C$, and $T$, as well as the imaginary unit $i$.

Using our GAP package, we obtain
\begin{equation}
    \begin{aligned}
            \symalg^{V} & \cong M_2(Cl^{2,3}) \times M_2(Cl^{2,3}) \times M_4(Cl^{2,3}) \\
            & \cong M_8(Cl^{0,1}) \times M_8(Cl^{0,1}) \times M_{16}(Cl^{0,1}) \\
    \end{aligned}
\end{equation}
Its classification can be read from \cref{table:A(Clpq)}, which is $\mathbb{Z}_2\times\mathbb{Z}_2\times\mathbb{Z}_2$. The Dirac Hamiltonian corresponding to the generator of the first $\mathbb{Z}_2$ is given by \cref{eq:Dirac Hamiltonian}, where the values of each term can be found in Appendix \ref{section:gap example}.

%% file: 5.conclusion.tex
\section{Conclusions}\label{section:conclusion}

In this work, we introduce the classifying spaces of free fermionic systems. Starting from the symmetries of free fermionic systems, we generalize these symmetries from groups to algebras. Subsequently, we introduce the classifying spaces for free fermionic systems in arbitrary dimensions endowed with such algebraic symmetries, and demonstrate that they are $G$-homotopy equivalent to the classifying spaces of Dirac Hamiltonians. Based on these classifying spaces, we then address the classification problem. We consider four distinct classification schemes prevalent in the physics literature, along with the canonical isomorphisms among them, noting that each scheme possesses its own specific applications and physical significance. Among these four schemes, the algebraic classification is the most amenable to computation. Consequently, we investigate the $\mathbb{Z}_2$-graded Wedderburn--Artin decomposition of the $\mathbb{Z}_2$-graded algebra $\symalg^V$ obtained from Dirac Hamiltonians, and we have developed a GAP package to facilitate these computations. The GAP package can be used to obtain the classification, the graded Wedderburn--Artin decomposition of $\symalg^V$, the Dirac Hamiltonians, and so on.

%% file: a1.some_linear_algebra.tex
\section{Some Linear Algebras}

\subsection{Conventions for matrix algebras}\label{subsection:Convention for Matrix Algebra}

Let $D$ be an algebra, and let $V$ be a free $D$ module. Notice that if $D$ is a division algebra, then a $D$ module is always free. $V$ can be equipped with a basis induced by $V\cong D^d$. Let $f:V\to W$ be a $D$-module map. $f$ can be expressed as a matrix over $D$ if both $V$ and $W$ have bases. If we insist on using column vectors to represent elements in $V$ and $W$, $f$ can be written as a matrix
\begin{equation}
    \begin{bmatrix}
        f_{11} & f_{12} & \cdots & f_{1d} \\
        f_{21} & f_{22} & \cdots & f_{2d} \\
        \vdots & \vdots & \ddots & \vdots \\
        f_{d'1} & f_{d'2} & \cdots & f_{d'd}
    \end{bmatrix},
\end{equation}
where $d'$ is the dimension of $W$, and $f_{ji}$ is the coefficient of the $j$th basis vector in the image of the $i$th basis vector under the action of $f$. Unfortunately, the matrix does not act from the left, since $f$ should be $D$-linear. If $v\in V$, the action of $f$ on $v$ is given by
\begin{equation}
    \left(\mathbf{v}^{\top} \mathbf{f}^{\top}\right)^{\top} = \left(
    \begin{bmatrix}
        v_{1} & v_{2} & \cdots & v_{d}
    \end{bmatrix}
    \begin{bmatrix}
        f_{11} & f_{21} & \cdots & f_{d'1} \\
        f_{12} & f_{22} & \cdots & f_{d'2} \\
        \vdots & \vdots & \ddots & \vdots \\
        f_{1d} & f_{2d} & \cdots & f_{d'd}
    \end{bmatrix}\right)^{\top}.
\end{equation}
Similarly, the composition $f\circ g$ is given by $\left(\mathbf{g}^{\top}\mathbf{f}^{\top}\right)^{\top}$. At first glance, you may think $\left(\mathbf{v}^{\top} \mathbf{f}^{\top}\right)^{\top} = \mathbf{f}\,\mathbf{v}$ and $\left(\mathbf{g}^{\top}\mathbf{f}^{\top}\right)^{\top}=\mathbf{f}\, \mathbf{g}$. However, $D$ is not commutative. The correct formula is $\left(\mathbf{v}^{\top} \mathbf{f}^{\top}\right)^{\top} = \mathbf{f}\cdot_{D^{\mathrm{op}}}\mathbf{v}$ and $\left(\mathbf{g}^{\top}\mathbf{f}^{\top}\right)^{\top}=\mathbf{f}\cdot_{D^\mathrm{op}}\mathbf{g}$, where $\cdot_{D^\mathrm{op}}$ is the multiplication of matrices over $D^{\mathrm{op}}$.

In the main text and in Appendixes \ref{section:Group Algebra over C and R} and \ref{section:graded group algebra}, we constructed $D$ structures on modules over various group algebras. Each representation map $\rho(g)$ commutes with the $D$ structure; therefore, we can represent $\rho(g)$ by a matrix in $M_d(D^{\mathrm{op}})$.

Conversely, given an algebra homomorphism $A\to M_d(D)$, the canonical $M_d(D)$ module $D^d$ induces a left $A$ module. However, the canonical left $D$ action on $D^d$ does not commute with the $A$ action; the canonical right $D$ action does commute with the $A$ action. Thus, $D^d$ is a $D^{\mathrm{op}}$-type module.

\subsection{Module over a division algebra}\label{subsection:Modules over Division Algebras}

For a $\Gr$-graded algebra, its $\Gr$-graded modules are also free. Notice that for a simple module $M$, we have $Dm_0=M$ for any nonzero $m_0\in M_0$. Then, we have a map $D\to Dm_0$ given by $d\mapsto dm_0$, which is an isomorphism since $D$ has no nontrivial left ideals.

%% file: a2.technical_graded.tex
\section{\texorpdfstring{Some Technical Results on Graded Algebras}{Some Technical Results on Graded Algebras}}\label{section:Some Calculation of Z2N-graded algebras}

\subsection{Clifford algebra}\label{subsection: Clifford Algebra in Appendix}

The basic concepts of Clifford algebras can be found in \cite{atiyahCliffordModules1964,karoubiKtheoryIntroduction2009}. One of the main purposes of this section is to derive the entries in \cref{table:Clifford algebras,table:The isomorphisms between Z2-graded division algebras and Clifford algebras}. Using \cref{table:Clifford algebras} together with \cref{eq:two isomorphism of Cl^{p+1 q+1}}, we will obtain the Wedderburn--Artin decompositions of all Clifford algebras both as algebras and as $\mathbb{Z}_2$-graded algebras.

\subsubsection{Definition}

\begin{defi}
    Let $V$ be a $\mathbb{K}$-vector space and let $Q$ be a quadratic form, i.e., a map $V\to \mathbb{K}$ satisfying $Q(kv)=k^2Q(v)$ for all $k\in \mathbb{K}$ and $v\in V$.

    There is a category whose objects are pairs $(A,j)$, where $A$ is an algebra and $j:V\to A$ is a linear map satisfying $j(v)^2 - Q(v)\cdot 1=0$. The morphism $(A,j)\to(B,k)$ is an algebra homomorphism $\psi: A \to B$ such that $\psi \circ j=k$.

    The initial object of this category is $(Cl_Q,j)$, where $Cl_Q$ is the Clifford algebra, and $j:V\to Cl_Q$ is the canonical inclusion.
\end{defi}
The definition seems to be very abstract, but the examples that we are interested in are very concrete.

\begin{eg}\label{eg:definition of Clifford algebras}
    Let $\mathbb{K}^{p,q}$ be a vector space with a bilinear form $(-,-)$. Let the basis of $\mathbb{K}^{p,q}$ be $\{e_{1} ,...,e_{p}, \tilde{e}_{1} ,...,\tilde{e}_{q}\}$ such that $(e_i,e_i)=-1$, $(\tilde{e}_i,\tilde{e}_i)=1$, and all other pairings are $0$. The bilinear form determines a quadratic form given by $Q(v)=(v,v)$.

    The algebra $Cl^{p,q}$ is the Clifford algebra corresponding to $\mathbb{R}^{p,q}$. Similarly, $\mathbb{C}l^{p,q}$ is the Clifford algebra corresponding to $\mathbb{C}^{p,q}$. $\mathbb{C}l^{p,q}$ is actually isomorphic to $\mathbb{C}l^{0,p+q}$ by taking a new basis $\{ie_i, \tilde{e}_i\}$. We also denote $\mathbb{C}l^{0,p+q}$ by $\mathbb{C}l^{p+q}$.

    To be explicit, $Cl^{p,q}$ is an $\mathbb{R}$ algebra generated by $\{e_{1} ,\ldots,e_{p}, \tilde{e}_{1} ,\ldots,\tilde{e}_{q}\}$ that satisfy the relations $e_{i}^{2} =-1$, $\tilde{e}_{i}^{2} =1$, $e_{i} e_{j} +e_{j} e_{i} =\tilde{e}_{i}\tilde{e}_{j} +\tilde{e}_{j}\tilde{e}_{i} = e_{i}\tilde{e}_{j} +\tilde{e}_{j}e_{i} =0$.
    Similarly, $\mathbb{C}l^{p,q}$ is a $\mathbb{C}$ algebra generated by the same sets of generators and relations.

    $Cl^{p,q}$ and $\mathbb{C}l^{p,q}$ are naturally $\mathbb{Z}_2$-crossed product algebras, i.e., $e_i$ and $\tilde{e}_i$ have grading 1.
\end{eg}

\subsubsection{The Wedderburn--Artin decomposition}

We have the following theorem for Clifford algebras.
\begin{thm}\label{thm:structure theorem of Clifford algebra}
    $Cl^{p,q}$ and $\mathbb{C}l^{p,q}$ are semisimple algebras.
\end{thm}
\begin{proof}
    We prove the theorem by constructing algebra isomorphisms from Clifford algebras to the direct sum of matrix algebras over division algebras.

    The following isomorphisms can be easily verified: $Cl^{0,0}\cong \mathbb{R}$, $Cl^{1,0}\cong \mathbb{C}$ with $e_1\mapsto i$, $Cl^{2,0}\cong \mathbb{H}$ with $e_1\mapsto I$ and $e_2\mapsto J$, $Cl^{0,1}\cong \mathbb{R}\times \mathbb{R}$ with $\tilde{e}_1\mapsto (1,-1)$, $Cl^{1,1}\cong M_2(\mathbb{R})$ with $e_1\mapsto i\sigma_y$, $\tilde{e}_1\mapsto \sigma_z$ ($i\sigma_y$ is a real matrix), and $Cl^{0,2}\cong M_2(\mathbb{R})$ with $\tilde{e}_1\mapsto \sigma_x$, $\tilde{e}_2\mapsto \sigma_z$.

    There are isomorphisms
    \begin{equation}\label{eq:isomorphism of Cl^{p+2 0}}
        Cl^{p+2,0}\cong Cl^{0,p}\otimes Cl^{2,0}
    \end{equation}
    given by $e_i\mapsto \tilde{e}_{i}\otimes e_{1}e_{2}$ for $i\leq p$, and $e_i\mapsto 1\otimes e_{i-p}$ for $i=p+1$ and $p+2$. Thus, we have $Cl^{3,0}\cong \mathbb{H}\times \mathbb{H}$.

    There are isomorphisms
    \begin{equation}\label{eq:isomorphism of Cl^{0 p+2}}
        Cl^{0,p+2}\cong Cl^{p,0}\otimes Cl^{0,2}
    \end{equation}
    given by $\tilde{e}_i\mapsto e_{i}\otimes \tilde{e}_{1}\tilde{e}_{2}$ for $i\leq p$, and $\tilde{e}_i\mapsto 1\otimes \tilde{e}_{i-p}$ for $i=p+1$ and $p+2$. Thus, we have $Cl^{0,3}\cong M_2(\mathbb{C})$.

    If we repeatedly use the above two isomorphisms, we obtain all $Cl^{p,0}$ and $Cl^{0,p}$. They are isomorphic to matrix algebras over division algebras, since we have the isomorphisms $\mathbb{H}\otimes \mathbb{C}\cong M_2(\mathbb{C})$ with $I\otimes 1\mapsto -i\sigma_x$ and $J\otimes 1\mapsto -i\sigma_y$, and $\mathbb{H}\otimes\mathbb{H}\cong M_4(\mathbb{R})$ with $I\otimes 1\mapsto i\sigma_y\otimes\sigma_z$, $J\otimes 1\mapsto i\sigma_y\otimes\sigma_x$, $1\otimes I\mapsto\sigma_z\otimes i\sigma_y$, $1\otimes J\mapsto\sigma_x\otimes i\sigma_y$.

    The remaining $Cl^{p,q}$ can be obtained by the isomorphisms
    \begin{equation}\label{eq:isomorphism of Cl^{p+1 q+1}}
        Cl^{p+1,q+1}\cong Cl^{p,q}\otimes Cl^{1,1}
    \end{equation}
    with $e_i\mapsto e_i\otimes e_1\tilde{e}_1$ and $\tilde{e}_j\mapsto \tilde{e}_j\otimes e_1\tilde{e}_1$ for $i\leq p$ and $j\leq q$, $e_{p+1}\mapsto 1\otimes e_1$, $\tilde{e}_{q+1}\mapsto 1\otimes\tilde{e}_1$.

    The result for complex Clifford algebras $\mathbb{C}l^p$ can be obtained by the complexification $\mathbb{C}l^{p,q}\cong Cl^{p,q}\otimes \mathbb{C}$.
\end{proof}

\subsubsection{\texorpdfstring{$\mathbb{Z}_2$-Graded Simple Algebras as Clifford Algebras}{Z2-Graded Division Algebra and Matrix Algebra as Clifford Algebra}}\label{subsubsection:Z2-Graded Division Algebra and Matrix Algebra as Clifford Algebra}

In this section, we show how to get \cref{table:The isomorphisms between Z2-graded division algebras and Clifford algebras}. In \cref{subsubsection:classification of Z2-graded division algebra}, we have classified all $\mathbb{Z}_2$-graded division algebras over $\mathbb{R}$ and $\mathbb{C}$. It is straightforward to prove that the division algebras in the first column are isomorphic to the Clifford algebras in the second column in \cref{table:The isomorphisms between Z2-graded division algebras and Clifford algebras}.

In \cref{thm:structure theorem of Gr-crossed product algebra}, we have shown that the graded matrix components of a $\Gr$-crossed product algebra are given by $M_{\left|\Gr/E\right|d}(D)(\bar{c})$. The group $E$ is a subgroup of $\Gr$ such that $D_c\not = 0$ for $c\in E$, and $\bar{c}$ is determined by a set of representatives of $\Gr/E$.

For a $\mathbb{Z}_2$-graded division algebra $D$ with $D_1\not =0$, we have $E=\mathbb{Z}_2$; therefore, $M_{\left|\Gr/E\right|d}(D)(\bar{c})$ is just $M_d(D)$.

For a $\mathbb{Z}_2$-graded division algebra $D$ with $D_1=0$, we have $E=\{0\}$ and $\bar{c}=(0,\cdots,0,1,\cdots,1)$, where both $0$ and $1$ appear $d$ times. Therefore, $M_{\left|\Gr/E\right|d}(D)(\bar{c})$ is decomposed into
\begin{equation}
    M_{2d}(D)(\bar{c})\cong M_d(\mathbb{K})\otimes M_2(\mathbb{K})(0,1)\otimes D.
\end{equation}
We have $M_2(\mathbb{R})(0,1)\cong Cl^{1,1}$ and $M_2(\mathbb{C})(0,1)\cong \mathbb{C}l^{1,1}$, which are given by $\sigma_x\mapsto \tilde{e}_1$ and $i\sigma_y\mapsto e_1$. For $D=\mathbb{R}$ and $\mathbb{C}$, we are done. For $D=\mathbb{H}$, we use $Cl^{1,1}\otimes \mathbb{H}$, which is given by \cref{eq:two isomorphism of Cl^{p+2 0}}; this finishes the derivation of the entries in \cref{table:The isomorphisms between Z2-graded division algebras and Clifford algebras}.

In some of the literature, only nongraded matrix algebras over Clifford algebras appear, for the reason encoded in the third column of \cref{table:The isomorphisms between Z2-graded division algebras and Clifford algebras}. However, it should be noted that although using Clifford algebras is computationally convenient, it lacks conceptual uniformity. This leads to the necessity of treating $\mathbb{R}$, $\mathbb{C}$, and $\mathbb{H}$ separately from $Cl^{1,1}$, $\mathbb{C}l^{2}$, and $Cl^{4,0}$. Therefore, graded matrix algebras will still be used in this work.

\subsubsection{\texorpdfstring{Clifford algebras as $\mathbb{Z}_2$-graded semisimple algebras}{Clifford algebra as Z2-graded algebra}}

We first prove the following three isomorphisms, which will be used in a similar way as \cref{eq:isomorphism of Cl^{p+2 0},eq:isomorphism of Cl^{0 p+2},eq:isomorphism of Cl^{p+1 q+1}}. We have the following isomorphisms of $\mathbb{Z}_2$-graded algebras,
\begin{equation}\label{eq:two isomorphism of Cl^{p+2 0}}
    Cl^{p+2,0}\cong Cl^{0,p}\otimes Cl^{2,0}\cong Cl^{p-1,1}\otimes \mathbb{H}.
\end{equation}
The first isomorphism is just \cref{eq:isomorphism of Cl^{p+2 0}}, and it is easy to check that \cref{eq:isomorphism of Cl^{p+2 0}} preserves the grading. The second isomorphism is given by $\tilde{e}_p\otimes e_1\mapsto 1\otimes I$, $\tilde{e}_p\otimes e_2\mapsto 1\otimes J$, $\tilde{e}_p\otimes 1 \mapsto \tilde{e}_1\otimes 1$ and $\tilde{e}_i\otimes e_1e_2\mapsto e_i\otimes 1$ for $i\not = p$. We have
\begin{equation}\label{eq:two isomorphism of Cl^{0 p+2}}
    Cl^{0,q+2}\cong Cl^{q,0}\otimes Cl^{0,2}\cong Cl^{1,q-1}\otimes \mathbb{H}.
 \end{equation}
The first isomorphism is just \cref{eq:isomorphism of Cl^{0 p+2}}, and it is easy to check that \cref{eq:isomorphism of Cl^{0 p+2}} preserves the grading. The second isomorphism is given by $e_q\otimes \tilde{e}_1\mapsto 1\otimes I$, $e_q\otimes \tilde{e}_2\mapsto 1\otimes J$, $e_q\otimes 1 \mapsto e_1\otimes 1$ and $e_i\otimes \tilde{e}_1\tilde{e}_2\mapsto \tilde{e}_i\otimes 1$ for $i\not = q$.

If $p>0$ or $q>0$, we have
\begin{equation}\label{eq:two isomorphism of Cl^{p+1 q+1}}
    Cl^{p+1,q+1}\cong Cl^{p,q}\otimes Cl^{1,1}\cong Cl^{p,q}\otimes M_2(\mathbb{R})\cong M_2(Cl^{p,q}).
\end{equation}
The first isomorphism is just \cref{eq:isomorphism of Cl^{p+1 q+1}}, and it is easy to check that \cref{eq:isomorphism of Cl^{p+1 q+1}} preserves the grading. For $p>0$, the second isomorphism is given by $e_p\otimes \tilde{e}_1\mapsto 1\otimes i\sigma_y$, $e_p\otimes e_1\mapsto 1\otimes \sigma_x$, $e_p\otimes 1\mapsto e_p\otimes 1$, $\tilde{e}_j\otimes e_1\tilde{e}_1\mapsto \tilde{e}_j\otimes 1$ and $e_i\otimes e_1\tilde{e}_1\mapsto e_i\otimes 1$ for $i \not = p$. For $q>0$, the second isomorphism is given by $\tilde{e}_q\otimes e_1\mapsto 1\otimes i\sigma_y$, $\tilde{e}_q\otimes \tilde{e}_1\mapsto 1\otimes \sigma_x$, $\tilde{e}_q\otimes 1\mapsto \tilde{e}_q\otimes 1$, $e_i\otimes e_1\tilde{e}_1\mapsto e_i\otimes 1$ and $\tilde{e}_j\otimes e_1\tilde{e}_1\mapsto \tilde{e}_j\otimes 1$ for $j \not = q$.

Using a method similar to \cref{thm:structure theorem of semisimple graded algebra}, we prove Clifford algebras are semisimple $\mathbb{Z}_2$-graded algebras, which also finishes the derivation of the entries in \cref{table:Clifford algebras}.
\begin{thm}\label{thm:structure theorem of Clifford algebra as Z2-graded algebra}
    $Cl^{p,q}$ and $\mathbb{C}l^{p,q}$ are semisimple $\mathbb{Z}_2$-graded algebras.
\end{thm}
\begin{proof}
    We use \cref{eq:two isomorphism of Cl^{p+2 0},eq:two isomorphism of Cl^{0 p+2},eq:two isomorphism of Cl^{p+1 q+1}} repeatedly to get the result. We first reduce $Cl^{p,q}$ to a matrix algebra over $Cl^{p-q,0}$ for $p>q$, a matrix algebra over $Cl^{0,q-p}$ for $p<q$ and a matrix algebra over $Cl^{1,1}$ for $p=q>0$, using \cref{eq:two isomorphism of Cl^{p+1 q+1}}.

    For $p=0\dots 4$, $Cl^{p,0}$ already has the form of graded matrix algebras over graded division algebras using \cref{table:The isomorphisms between Z2-graded division algebras and Clifford algebras}. For $p=5\dots 8$, we apply \cref{eq:two isomorphism of Cl^{p+2 0}}. We get $Cl^{2,1}\otimes \mathbb{H}$ for $p=5$, $Cl^{0,2}\otimes \mathbb{H}\otimes \mathbb{H}$ for $p=6$, $Cl^{1,2}\otimes \mathbb{H}\otimes \mathbb{H}$ for $p=7$, and $Cl^{2,2}\otimes \mathbb{H}\otimes \mathbb{H}$ for $p=8$. Then we apply \cref{eq:two isomorphism of Cl^{p+1 q+1}}, up to tensoring with a matrix algebra, $Cl^{2,1}$ becomes $Cl^{1,0}$ for $p=5$, $Cl^{1,2}$ becomes $Cl^{0,1}$ for $p=7$, and $Cl^{2,2}$ becomes $Cl^{1,1}$ for $p=8$. For $p=5$, we apply \cref{eq:two isomorphism of Cl^{0 p+2}}, $Cl^{1,0}\otimes \mathbb{H}$ becomes $Cl^{0,3}$. Finally, we use $\mathbb{H}\otimes\mathbb{H}\cong M_4(\mathbb{R})$ which is proved in \cref{thm:structure theorem of Clifford algebra}. The remaining cases are similar. Notice that we have $Cl^{4,0}\cong Cl^{0,4}$, which is also proved by \cref{eq:two isomorphism of Cl^{p+2 0},eq:two isomorphism of Cl^{0 p+2}}.
\end{proof}

\subsubsection{The simple modules over Clifford algebras}\label{subsubsection:modules of Clifford algebras}

It is straightforward to get all simple modules and simple graded modules over a Clifford algebra using \cref{thm:structure theorem of Clifford algebra} and \cref{thm:structure theorem of Clifford algebra as Z2-graded algebra}, since they are simple modules over the direct sums of matrix algebras and graded matrix algebras. Notice that if $D\cong\mathbb{R}$, $\mathbb{C}$, or $\mathbb{H}$, we have two inequivalent $\mathbb{Z}_2$-graded simple $D$-modules $D$ and $D(1)$. We prove two useful results in what follows.

Let us prove the first one. When $p-q=3\, \mathrm{mod}\, 4$, $Cl^{p,q}$ has two matrix components. The simple $Cl^{p,q}$ modules can be constructed from the unique $Cl^{p-1,q}$ module (or $Cl^{p,q-1}$ module). We denote by $\varepsilon$ the element $e_1\dots e_{p-1}\tilde{e}_1\dots\tilde{e}_q$ (or $e_1\dots e_{p}\tilde{e}_1\dots\tilde{e}_{q-1}$). Then $e_p=\pm\varepsilon$ (or $\tilde{e}_q=\pm\varepsilon$) defines two inequivalent simple $Cl^{p,q}$ modules. To prove that the two modules are simple, we only need to examine their dimensions. To prove that the two modules are inequivalent, notice that
if there is an isomorphism $f$ between them, we have $f(\varepsilon x)=-\varepsilon f(x)$, but $\varepsilon$ is a product of elements of $Cl^{p-1,q}$ (or $Cl^{p,q-1}$); therefore, we should have $f(\varepsilon x)=\varepsilon f(x)$, a contradiction. There is a similar result for $\mathbb{C}l^{p,q}$ when $p-q=1\, \mathrm{mod}\, 2$.

The second result is $(Cl^{p,q})_0\cong Cl^{p-1,q}\cong Cl^{q-1,p}$ and $(\mathbb{C}l^{p})_0 \cong \mathbb{C}l^{p-1}$. For a generator $e\in Cl^{p-1,q}$, the isomorphism from $Cl^{p-1,q}$ to $(Cl^{p,q})_0$ is given by $e\mapsto ee_p$. The other two isomorphisms can be proved similarly. Combining with \cref{prop:equivalence of categories between right A-modules and A_0-modules}, we get another way to construct all $\mathbb{Z}_2$-graded simple modules over Clifford algebras.

\subsection{\texorpdfstring{$\mathbb{Z}_2^N$-graded algebra}{Z2N-graded algebra}}\label{subsection:Z2N-Graded Algebra in Appendix}

\subsubsection{Grading forgetting}\label{subsubsection:grading forgetting}

For a $\Gr$-graded algebra $A$, a group homomorphism $f:\Gr\to\Gr'$ induces a $\Gr'$-grading structure on $A$. When $\Gr = \mathbb{Z}_2^N$ and $\Gr' = \mathbb{Z}_2^{N'}$, the problem reduces to the case when $f$ is a surjective map from $\mathbb{Z}_2^{N}$ to $\mathbb{Z}_2^{N-1}$. This can be seen by diagonalizing the linear map $f$ over the field $\mathbb{Z}_2$.

In this section, we give an algorithm to calculate the Wedderburn--Artin decomposition of a $\mathbb{Z}_2^{N-1}$-graded algebra that is induced from a semisimple $\mathbb{Z}_2^N$-crossed product algebra by $f$. In what follows, without loss of generality, we focus on a single $\mathbb{Z}_2^N$-graded Wedderburn--Artin component $M_{\left|\Gr/E\right|d}(D)(\bar{c})$.

We need to forget the grading $\mathrm{ker}(f)$ of the graded matrix algebra $M_{\left|\Gr/E\right|d}(D)(\bar{c})$. Recall that we have an isomorphism of $\mathbb{Z}_2^N$-graded algebras,
\begin{equation}\label{eq:decomposition of graded matrix algebra}
    \begin{aligned}
        M_{\left|\Gr/E\right|d}(D)(\bar{c}) & \cong M_{\left|\Gr/E\right|d}(\mathbb{K})(\bar{c})\otimes D \\
        & \cong M_d(\mathbb{K})\otimes \left(\bigotimes_{c_i}M_2(\mathbb{K})(0,c_i)\right) \otimes D.
    \end{aligned}
\end{equation}
If $\mathrm{ker}(f)$ is contained in $E$, the problem reduces to forgetting the grading of $D$, and the pure matrix algebra part is easy to deal with. If $\mathrm{ker}(f)$ is not contained in $E$, we can find a new decomposition $\Gr \cong E\oplus \Gr/E$ such that $\mathrm{ker}(f)$ is contained in $\Gr/E$. The decomposition is determined by taking another set of $c_i$. Then the problem reduces to forgetting the grading of the pure matrix algebra part.

It remains to figure out how to forget the grading of $D$. We denote the generators of $D$ by $\theta_i$, where $\theta_i$ has grading $z_i$. We need to consider the following three cases according to the structure of $D_0$.
\paragraph{1. $D_0=\mathbb{R}$} \,

If $\theta_1^2=1$ and $\theta_1$ commutes with all other $\theta_i$, $D\cong D'\oplus D'$, where $D'$ is $D$ without the generator $\theta_1$.

If $\theta_1^2=1$, $\theta_1$ anticommutes with $\theta_2$ and commutes with all the other $\theta_i$. Notice that the general case in which $\theta_1$ anticommutes with several $\theta_i$ can be reduced to this case by changing the basis of $\mathbb{Z}_2^N$: we exchange $\theta_2$ with some $\theta_i$ that anticommutes with $\theta_1$, and multiply $\theta_2$ by the other $\theta_i$ that anticommute with $\theta_1$. This technique will be used again, and we will not mention it again. $\theta_1\theta_2$ and $\theta_2$ form an algebra $M_2(\mathbb{K})(0,z_2)$. If $\theta_2$ anticommutes with some $\theta_i$, we use $\theta_1\theta_i$ to replace $\theta_i$. Now $M_2(\mathbb{K})(0,z_2)$ is decoupled from the other $\theta_i$. We get $D\cong M_2(\mathbb{K})(0,z_2)\otimes D''$, where $D''$ is $D$ without the generators $\theta_1$ and $\theta_2$.

If $\theta_1^2=-1$, $D$ becomes a $\mathbb{Z}_2^{N-1}$-graded division algebra with $D_0=\mathbb{C}$. If $\theta_i$ commutes with $\theta_1$, then $D_i\cong \mathbb{C}l^{1}$. If $\theta_i$ anticommutes with $\theta_1$, then $D_i\cong Cl^{0,2}$ for $\theta_i^2=1$, and $D_i\cong Cl^{2,0}$ for $\theta_i^2=-1$.

\paragraph{2. $D_0=\mathbb{H}$} \,

Notice that $\theta_i$ commutes with the elements of $D_0$; therefore, we can decompose $D$ into $D_0\otimes D_{\mathbb{R}}$, where $D_{\mathbb{R}}$ has $\left(D_{\mathbb{R}}\right)_0=\mathbb{R}$. The structure of $D_{\mathbb{R}}$ after forgetting the grading is already obtained above; therefore, we only need to calculate the structure of $\mathbb{H}\otimes D_{\mathbb{R}}$.

The only nontrivial case is $\theta_1^2=-1$; hence, we need to find out the structure of $\mathbb{H}\otimes D$ for a $\mathbb{Z}_2^{N-1}$-graded division algebra $D$ with $D_0\cong \mathbb{C}$. We use a new set of generators $I\otimes i$, $J\otimes i$, $K\otimes \theta_i$ for $\theta_i$ anticommuting with $i$, and $1\otimes \theta_i$ for $\theta_i$ commuting with $i$. $I\otimes i$ and $J\otimes i$ form the matrix algebra $M_2(\mathbb{R})$, and all generators with nontrivial grading commute with them. Thus, we get $\mathbb{H}\otimes D\cong M_2(\mathbb{R})\otimes D'$, where $(D')_i\cong Cl^{2,0}$ if $D_i\cong Cl^{0,2}$, $(D')_i\cong Cl^{0,2}$ if $D_i\cong Cl^{2,0}$, and $(D')_i\cong \mathbb{C}l^{1}$ if $D_i\cong \mathbb{C}l^{1}$. The commutation relations of $\theta_i$ are not changed.

\paragraph{3. $D_0=\mathbb{C}$ (including $\mathbb{K}=\mathbb{C}$)} \,

The case $\mathbb{K}=\mathbb{C}$ is a special case of $\mathbb{K}=\mathbb{R}$, $D_0=\mathbb{C}$, and arises when all $\theta_i$ commute with the elements of $D_0$.

If $\theta_1$ commutes with $i\in D_0$, and $\theta_1$ commutes with all other $\theta_i$, then $D\cong D'\oplus D'$, where $D'$ is $D$ without the generator $\theta_1$.

If $\theta_1$ commutes with $i\in D_0$, and $\theta_1$ anticommutes with $\theta_2$, then if $\theta_2$ commutes with $i\in D_0$, $\theta_2$ and $\theta_1\theta_2$ generate $M_2(\mathbb{R})(0,z_2)$. If there exists a $\theta_i$ that anticommutes with $\theta_2$, we choose the new generator to be $\theta_1\theta_i$. Then $D\cong M_2(\mathbb{R})(0,z_2)\otimes D''$, where $D''$ is $D$ without the generators $\theta_1$ and $\theta_2$. If $\theta_2$ anticommutes with $i\in D_0$, then $\theta_2$ and $\theta_1\theta_2$ generate $M_2(\mathbb{R})$, but we need to use $i\theta_1$ to replace $i$. If there exists a $\theta_i$ that anticommutes with $\theta_2$, we choose the new generator to be $\theta_1\theta_i$. Then $D\cong M_2(\mathbb{R})(0,z_2)\otimes D''$, where $D''$ is $D$ without the generators $\theta_1$ and $\theta_2$.

If $\theta_1$ anticommutes with $i\in D_0$, and $\theta_1^2=1$, then $\theta_1$ and $i$ generate $M_2(\mathbb{R})$. We show by redefinition of $\theta_i$ that $M_2(\mathbb{R})$ commutes with all the other $\theta_i$; therefore, $D\cong M_2(\mathbb{R})\otimes D'$, where $D'$ is a $\mathbb{Z}_2^{N-1}$-graded division algebra with $D_0\cong \mathbb{R}$. For $D_i\cong \mathbb{C}l^{1}$, if $\theta_i$ commutes with $\theta_1$, then $D'_i\cong Cl^{0,1}$. If $\theta_i$ anticommutes with $\theta_1$, we choose a new generator to be $i\theta_i$; then $D'_i\cong Cl^{1,0}$. For $D_i\cong Cl^{0,2}$, if $\theta_i$ commutes with $\theta_1$, we choose a new generator to be $\theta_1\theta_i$; then $D'_i\cong Cl^{0,1}$. If $\theta_i$ anticommutes with $\theta_1$, we choose a new generator to be $i\theta_1\theta_i$; then $D'_i\cong Cl^{0,1}$. For $D_i\cong Cl^{2,0}$, if $\theta_i$ commutes with $\theta_1$, we choose a new generator to be $\theta_1\theta_i$; then $D'_i\cong Cl^{1,0}$. If $\theta_i$ anticommutes with $\theta_1$, we choose a new generator to be $i\theta_1\theta_i$; then $D'_i\cong Cl^{1,0}$.

If $\theta_1$ anticommutes with $i\in D_0$, and $\theta_1^2=-1$, then $\theta_1$ and $i$ generate $\mathbb{H}$. We show by redefinition of $\theta_i$ that $\mathbb{H}$ commutes with all the other $\theta_i$; therefore, $D$ becomes a $\mathbb{Z}_2^{N-1}$-graded division algebra with $D_0\cong \mathbb{H}$, which we denote by $D'$. For $D_i\cong \mathbb{C}l^{1}$, if $\theta_i$ commutes with $\theta_1$, then $D'_i\cong Cl^{3,0}$. If $\theta_i$ anticommutes with $\theta_1$, we choose a new generator to be $i\theta_i$; then $D'_i\cong Cl^{0,3}$. For $D_i\cong Cl^{0,2}$, if $\theta_i$ commutes with $\theta_1$, we choose a new generator to be $\theta_1\theta_i$; then $D'_i\cong Cl^{3,0}$. If $\theta_i$ anticommutes with $\theta_1$, we choose a new generator to be $i\theta_1\theta_i$; then $D'_i\cong Cl^{3,0}$. For $D_i\cong Cl^{2,0}$, if $\theta_i$ commutes with $\theta_1$, we choose a new generator to be $\theta_1\theta_i$; then $D'_i\cong Cl^{0,3}$. If $\theta_i$ anticommutes with $\theta_1$, we choose a new generator to be $i\theta_1\theta_i$; then $D'_i\cong Cl^{0,3}$.

\subsubsection{\texorpdfstring{Graded tensor product of $\mathbb{Z}_2^N$-graded algebras}{Graded tensor product of Z2N-graded algebras}}\label{subsubsection:Graded Tensor Product of Z2N-Graded Algebras}

Similar to the previous section, this section only considers graded matrix algebras over graded division algebras. First, we define an extended graded tensor product $\tilde{\otimes}$. Given a $\mathbb{Z}_2^N$-graded space $V$ and a $\mathbb{Z}_2^{N'}$-graded space $V'$, $V \tilde{\otimes} V'$ is a $\mathbb{Z}_2^{N+N'}$-graded space. Similar to the braiding structure in \cref{subsubsection:Braiding in Graded Vector Space}, a map $b:\mathbb{Z}_2^N\times \mathbb{Z}_2^{N'}\to \mathbb{Z}_2$ is used to describe the sign acquired by exchanging the two spaces. The graded tensor product defined in \cref{subsubsection:graded vector space} requires the further step of forgetting the off-diagonal grading, which can be accomplished using the method described in Appendix \ref{subsubsection:grading forgetting}.

Given a $\mathbb{Z}_2^N$-graded algebra $A$ and a $\mathbb{Z}_2^{N'}$-graded algebra $A'$, $A\tilde{\otimes}A'$ is a $\mathbb{Z}_2^{N+N'}$-graded algebra. Suppose the graded matrix algebra $A$ is isomorphic to $\mathrm{End}_{D}(V)$ and $A'$ is isomorphic to $\mathrm{End}_{D'}(V')$, where $V$ and $V'$ are finitely generated $D$-graded and $D'$-graded modules, respectively. We have the following formula:
\begin{equation}
    \mathrm{End}_{D}(V) \tilde{\otimes} \mathrm{End}_{D'}(V') \cong \mathrm{End}_{D \tilde{\otimes} D'}(V \tilde{\otimes} V').
\end{equation}
Therefore, it suffices to study $D \tilde{\otimes} D'$.

Suppose $D$ is generated by $D_0$ and the $\theta_i$, and $D'$ by $D'_0$ and the $\theta'_i$; then $D \tilde{\otimes} D'$ is generated by $D_0\otimes D'_0$, the $\theta_i$, and the $\theta'_i$. The algebraic relations between the $\theta_i$ and the $\theta'_i$ can be easily obtained from $b$. $D_0\otimes D'_0$ is the tensor product of nongraded division algebras, and its result is also easy to obtain. When $D_0$ and $D_0'$ are each $\mathbb{R}$ or $\mathbb{H}$, $D_0\otimes D'_0$ commutes with the generators $\theta_i$ and $\theta'_i$. When exactly one of $D_0$ and $D_0'$ is $\mathbb{R}$ and the other is $\mathbb{C}$, the result is also obvious. When exactly one of $D_0$ and $D_0'$ is $\mathbb{H}$ and the other is $\mathbb{C}$, their tensor product is $M_2(\mathbb{C})$, and $I$, $J$, and $K$ are mapped to $-i\sigma_x$, $-i\sigma_y$, and $-i\sigma_z$ respectively, where the $\sigma$'s are the $2\times 2$ Pauli matrices. For an antilinear $\theta'_i$, it commutes with $-i\sigma_y$ and anticommutes with $-i\sigma_x$ and $-i\sigma_z$; therefore, by choosing $-i\sigma_y\theta'_i$ as a new generator, its conjugation action on $M_2(\mathbb{C})$ is the same as complex conjugation.

Finally, if both $D_0$ and $D_0'$ are $\mathbb{C}$, their tensor product is $\mathbb{C}\times \mathbb{C}$, where $z\otimes 1$ is mapped to $(z, \bar{z})$ and $1\otimes z$ is mapped to $(z, z)$. If neither $D$ nor $D'$ has antilinear generators, their tensor product is the direct sum of two graded division algebras, each generated by $\mathbb{C}$, the $\theta_i$, and the $\theta'_i$. If there are antilinear generators in both $D$ and $D'$, and we assume the generator $\theta_1$ of $D$ is antilinear while all other $\theta_i$ are linear, then in $D \tilde{\otimes} D'$, the elements $\theta_1 \tilde{\otimes} 1$, $i\theta_1 \tilde{\otimes} 1$, and $1 \tilde{\otimes} i$ form $M_2(\mathbb{C})(0,1)$. Then, it is sufficient to replace the $\theta_i \tilde{\otimes} 1$ or $1 \tilde{\otimes} \theta'_i$ that anticommute with $\theta_1 \tilde{\otimes} 1$ by $i\theta_i \tilde{\otimes} 1$ or $i \tilde{\otimes} \theta'_i$.

For the super tensor product $\hat{\otimes}$ of real Clifford algebras and the ten $\mathbb{Z}_2$-graded division algebras, one can obtain the Clifford algebras by utilizing $Cl^{p,q}\hat{\otimes}\allowbreak Cl^{p',q'}\cong Cl^{p+p',q+q'}$, $Cl^{p,q}\hat{\otimes}\mathbb{C}l^{p'}\cong \mathbb{C}l^{p+p'+q}$, $Cl^{p,q}\hat{\otimes}\mathbb{H}\cong Cl^{p+3,p-1}$ and $Cl^{p,q}\hat{\otimes}\mathbb{H}\cong Cl^{p-1,p+3}$, and then perform the Wedderburn--Artin decomposition on these Clifford algebras.

%% file: a3.group_algebra.tex
\section{\texorpdfstring{Group Algebras over $\mathbb{C}$ and $\mathbb{R}$}{Group Algebras over C and R}}\label{section:Group Algebra over C and R}

We summarize several results in group algebra theory. Starting with the group algebra $\mathbb{C}[G]$ over $\mathbb{C}$, we note that the theories of the group algebra over $\mathbb{R}$, the twisted group algebra, and the $\mathbb{Z}_2$-graded group algebra are all based on the theory of $\mathbb{C}[G]$. The basic knowledge of group representations and group algebras over $\mathbb{R}$ and $\mathbb{C}$ can be found in \cite{serreLinearRepresentationsFinite1977}.

The main purpose of this appendix is to introduce \cref{eq:Maschke's theorem complex,Maschke's theorem real,Maschke's theorem twisted}, and their inverses \cref{inverse of Phi complex,eq:inverse of Phi real,eq:inverse of Phi twisted real,eq:inverse of Phi twisted complex}. Most other results follow as direct consequences of these isomorphisms.

\subsection{\texorpdfstring{Group algebra $\mathbb{C}[G]$}{Group algebra C[G]}}\label{subsection:Group Algebra over C}

\subsubsection{The Wedderburn--Artin Decomposition}\label{subsubsection:Decomposition into Matrix Algebras}

The well-known Maschke's theorem shows that $\mathbb{C}[G]$ is a semisimple $\mathbb{C}$ algebra, which means there is an isomorphism of algebras $\Phi=(\rho_i)_{i\in \mathrm{Irr}(G)}$
\begin{equation}\label{eq:Maschke's theorem complex}
    \mathbb{C}[G] \cong \prod_{i \in \mathrm{Irr}(G)} M_{d_i}(\mathbb{C}),
\end{equation}
where $\mathrm{Irr}(G)$ is the set of isomorphism classes of simple modules, and each component of $\Phi$ is given by the irreducible representation $\rho_i:\mathbb{C}[G]\to M_{d_i}(\mathbb{C})$. This is the Wedderburn--Artin decomposition of $\mathbb{C}[G]$.

\subsubsection{\texorpdfstring{The inverse of $\Phi$}{The inverse of Phi}}\label{subsubsection:Inverse of Phi complex}

The inverse of $\Phi$ can be constructed explicitly; let $f$ be an element of $\prod_i M_{d_i}(\mathbb{C})$ that is nonzero only in the component $M_{d_i}(\mathbb{C})$. The inverse $\Phi^{-1}(f)$ is given by
\begin{equation}\label{inverse of Phi complex}
    \frac{d_i}{|G|} \sum_{g \in G} \mathrm{Tr}\left( \rho_i(g^{-1})f \right) \, g.
\end{equation}
To prove this is the inverse, we assume $\Phi^{-1}(f) = \sum_{g\in G} a_g\, g \in \mathbb{C}[G]$. Notice that we have
\begin{equation}
    \begin{aligned}
        a_g & = \frac{1}{|G|}\, \mathrm{Tr}\left(g^{-1}\Phi^{-1}(f)\cdot -\right) \\
        & = \frac{d_i}{|G|}\, \mathrm{Tr}\left(\rho_i(g^{-1})f\right)
    \end{aligned}
\end{equation}
The trace in the first line is calculated as a linear map on $\mathbb{C}[G]$, and the trace in the second line is calculated as a linear map on $M_{d_i}(\mathbb{C})$.

Several familiar results follow straightforwardly from the explicit form of $\Phi^{-1}$. For example, if we take $f = e^i_{ba}\in M_{d_i}(\mathbb{C})$ to be a matrix entry generator, then the equation $\Phi\left(\Phi^{-1}\left(e^i_{ba}\right)\right) = e^i_{ba}$ becomes the great orthogonality theorem
\begin{equation}\label{inverse of Phi on matrix entries complex}
    \frac{d_i}{|G|} \sum_{g \in G} \rho_{i}(g^{-1})_{ab} \, \rho_j(g)_{cd} = \delta_{ij}\delta_{ad}\delta_{bc}.
\end{equation}
If we take $f = \mathrm{id}\in M_{d_i}(\mathbb{C})$, then the equation $\mathrm{Tr}\left( \Phi\left(\Phi^{-1}\left(f\right)\right) \right) = \mathrm{Tr}(f)$ becomes $(\chi_i,\chi_j) = 1/|G| \, \delta_{ij}$, where
\begin{equation}\label{eq:inner product of class function}
    (\chi,\chi') = \frac{1}{|G|} \sum_{g \in G} \chi(g^{-1}) \, \chi'(g),
\end{equation}
i.e., the characters $\chi_i$ form an orthonormal basis of the vector space of class functions.

\subsubsection{Constructing isomorphisms between two modules}\label{subsubsection:Constructing isomorphisms}

Let $V$ and $W$ be two isomorphic simple $\mathbb{C}[G]$ modules. We construct a $\mathbb{C}[G]$ isomorphism between them\footnote{The most efficient way to construct a $\mathbb{C}[G]$-isomorphism is probably to start with an arbitrary linear isomorphism and average it over $G$. The Schur lemma ensures that the result of the averaging is either 0 or a $\mathbb{C}[G]$-isomorphism. The result is a $\mathbb{C}[G]$-isomorphism for almost all choices of linear isomorphisms. However, we want to use a method that doesn't involve a random choice.}.

Notice that there is an isomorphism $\psi:\mathrm{End}(V)\xrightarrow{\Phi^{-1}} \mathbb{C}[G] \xrightarrow{\Phi} \mathrm{End}(W)$. The first map does not preserve the unit, but the composition is an algebra isomorphism. To be explicit, let $f\in \mathrm{End}(V)$. The image $\psi(f)$ is
\begin{equation}\label{isomorphism of EndV and EndW}
    \frac{d}{|G|} \sum_{g \in G} \mathrm{Tr}\left(\rho_V(g^{-1})f\right) \rho_W(g)\in \mathrm{End}(W),
\end{equation}
where $d$ is the dimension of $V$. According to the Skolem--Noether theorem, the isomorphism must have the form $\alpha f \alpha^{-1}$ for some linear isomorphism $\alpha$, which is unique up to a scalar. The isomorphism $\alpha$ is $G$-equivariant. Notice that if we set $f=\rho_V(h)$, \cref{isomorphism of EndV and EndW} becomes $\rho_W(h)$, i.e., $\alpha \rho_V(h) \alpha^{-1}=\rho_W(h)$. $\alpha$ seems to be a suitable choice of the $\mathbb{C}[G]$ isomorphism; it only remains to construct $\alpha$ explicitly.

Let $r\in \mathrm{End}(V)$ be a diagonalizable rank-1 linear map, and let $v$ be an eigenvector of $r$. The image $\psi(r)$ also has rank 1, and $\alpha(v)$ must be an eigenvector of $\psi(r)$. The choice of $\alpha(v)$ is unique up to a scalar, which is the same scalar appearing in the Skolem--Noether theorem. $\alpha$ is determined by $\alpha(v)$, since $\alpha$ is $G$ equivariant, and any $f\in \mathrm{End}(V)$ can be written as a linear combination of $\rho_V(g)$, i.e., we have
\begin{equation}\label{isomorphism of V and W}
    \alpha\left(f(v)\right)=\frac{d}{|G|} \sum_{g \in G} \mathrm{Tr}\left(\rho_V(g^{-1})f\right)\rho_W(g)\alpha(v).
\end{equation}
In practice, if there is a basis on $V$, we can choose $r$ to be the canonical projection on a basis vector, and use \cref{isomorphism of V and W} to calculate the action of $\alpha$ on other basis vectors.

\subsection{\texorpdfstring{Group algebra $\mathbb{R}[G]$}{Group algebra R[G]}}

\subsubsection{The Wedderburn--Artin decomposition}\label{subsubsection:Decomposition into Matrix Algebras real}

Maschke's theorem also shows that $\mathbb{R}[G]$ is a semisimple $\mathbb{R}$ algebra, which means there is an isomorphism of algebras $\Phi=(\rho_i)$,
\begin{equation}\label{Maschke's theorem real}
    \mathbb{R}[G] \cong \prod_{i \in \mathrm{Irr}_{\mathbb{R}}(G)} M_{d_i'}(D_i).
\end{equation}
Here, $i$ takes values in the isomorphism classes of simple $\mathbb{R}[G]$ modules. $D_i$ is a division $\mathbb{R}$ algebra. Moreover, $d_i' = d_i/\mathrm{dim}(D_i)$, where $d_i$ is the dimension of the simple module. Moreover, the division algebra $D_i$ is the opposite algebra of the endomorphism algebra of the simple $\mathbb{R}[G]$ module.

A simple module is of $D$ type if the endomorphism algebra of the module is $D$, and the endomorphism algebra becomes a $G$-equivariant $D$ structure on the simple module. It is well known that there are only three isomorphism classes of division $\mathbb{R}$ algebras, which are $\mathbb{R}$, $\mathbb{C}$, and $\mathbb{H}$.

In what follows, we show how to construct $\Phi$ explicitly from the complex representations.

\subsubsection{\texorpdfstring{The relationship between $\mathbb{R}[G]$ and $\mathbb{C}[G]$ modules}{subsubsection:The relationship between R[G] and C[G] modules}}\label{subsubsection:The Relationship Between R[G] and C[G] Modules}

If we apply the complexification functor $-\otimes_{\mathbb{R}}\mathbb{C}$ to both sides of \cref{Maschke's theorem real}, the left-hand side becomes $\mathbb{C}[G]$, and each component of the right-hand side becomes one or two components in \cref{eq:Maschke's theorem complex}, i.e.,
$M_{d}(\mathbb{R})\otimes_{\mathbb{R}}\mathbb{C} \cong M_{d}(\mathbb{C})$,
$M_{d}(\mathbb{C})\otimes_{\mathbb{R}}\mathbb{C} \cong M_{d}(\mathbb{C})\times M_{d}(\mathbb{C})$
and $M_{d}(\mathbb{H})\otimes_{\mathbb{R}}\mathbb{C} \cong M_{2d}(\mathbb{C})$.
The first isomorphism is obvious; the second is given by $f\otimes_{\mathbb{R}}1\mapsto (f,\bar{f})$; and the third is induced by mapping $I$, $J$, and $K$ to $-i\sigma_x$, $-i\sigma_y$, and $-i\sigma_z$, respectively, where the $\sigma$'s are the $2\times 2$ Pauli matrices.

This gives a correspondence between real and complex representations. For $D=\mathbb{R}$ and $\mathbb{H}$, there is a one-to-one correspondence between simple $\mathbb{R}[G]$ modules and simple $\mathbb{C}[G]$ modules. For $D=\mathbb{C}$, there is a one-to-two correspondence between simple $\mathbb{R}[G]$ modules and simple $\mathbb{C}[G]$ modules.

Let $W$ be a simple $\mathbb{R}[G]$ module, and let $V$ and $\bar{V}$ be the corresponding simple $\mathbb{C}[G]$ modules. For $D=\mathbb{R}$ or $\mathbb{H}$, we have $V\cong\bar{V}$. We say that the simple $\mathbb{C}[G]$ module $V$ has $D$ type if $W$ does.

Notice that $V$ is not necessarily the complexification of $W$. This is true for $\mathbb{R}$-type modules: we have $W\otimes_{\mathbb{R}}\mathbb{C}\cong V$. For $\mathbb{C}$-type modules, we have $W\otimes_{\mathbb{R}}\mathbb{C}\cong V\oplus \bar{V}$. For $\mathbb{H}$-type modules, we have $W\otimes_{\mathbb{R}}\mathbb{C}\cong V\oplus V$.

Using the correspondence between $\mathbb{R}[G]$ modules and $\mathbb{C}[G]$ modules, it is easy to show $V \not\cong \bar{V}$ for $\mathbb{C}$-type simple $\mathbb{C}[G]$ modules, and there exists an isomorphism $\alpha:V\cong \bar{V}$ for $\mathbb{R}$- and $\mathbb{H}$-type simple $\mathbb{C}[G]$ modules. To further distinguish the $\mathbb{R}$ and $\mathbb{H}$ types, notice that $\bar{\alpha}\circ \alpha=\alpha\circ \bar{\alpha}=z\cdot -$ for some nonzero $z\in \mathbb{C}$, and $\bar{\alpha}\circ \alpha \circ \bar{\alpha}\circ \alpha=z^2\,\mathrm{id}=|z|^2\,\mathrm{id}$; hence, $z\in \mathbb{R}$. If $V$ has $\mathbb{R}$ type, then $z$ is positive; therefore, $\alpha$ is a $G$-equivariant real structure up to a real scalar. If $V$ is $\mathbb{H}$ type, $z$ is negative; therefore, $\alpha$ is a $G$-equivariant quaternion structure up to a real scalar.

There is a simple method to distinguish the three types by the Frobenius--Schur indicator. For $i\in \mathrm{Irr}(G)$,
\begin{equation}\label{Frobenius--Schur indicator}
    FS(i) = \frac{1}{|G|} \sum_{g \in G} \chi_i(g^2).
\end{equation}
If $i$ is of $\mathbb{R}$, $\mathbb{C}$, or $\mathbb{H}$ type, then $FS(i)$ is $1$, $0$, and $-1$, respectively. The proof is in Appendix \ref{subsubsection:Characters and Great Orthogonality Theorem real}.

In summary, assume that all complex irreducible representations are known. To get all components $\rho_i$ of $\Phi$, we first calculate the indicator $FS$ for each complex irreducible representation $\rho$.
\begin{itemize}
    \item If $FS=1$, $\rho$ has $\mathbb{R}$ type. There exists $\alpha:V\cong\bar{V}$ such that $\bar{\alpha}\circ \alpha=\mathrm{id}$. The isomorphism $\alpha$ can be constructed by \cref{isomorphism of V and W}. $\alpha$ is a $G$-equivariant real structure; by restricting to the eigenspace of $\alpha$, we get $\rho_i$.

    \item If $FS=0$, $\rho$ has $\mathbb{C}$ type. In this case, $\rho_i$ is just $\rho$.

    \item If $FS=-1$, $\rho$ has $\mathbb{H}$ type. There exists an isomorphism $\alpha:V\cong\bar{V}$ such that $\bar{\alpha}\circ \alpha=-\mathrm{id}$. The isomorphism $\alpha$ can be constructed by \cref{isomorphism of V and W}, and $\alpha$ is a $G$-equivariant quaternion structure, which gives $\rho_i$.
\end{itemize}

\subsubsection{\texorpdfstring{The inverse of $\Phi$}{The inverse of Phi}}\label{subsubsection:The Inverse of Phi real}

Notice that there is a commutative diagram
\begin{widetext}
\begin{equation}
    \begin{tikzcd}
        \mathbb{R}[G] \ar[r,"\Phi"]\ar[d,hook] & \prod_{i_{\mathbb{R}}} M_{d_{i_\mathbb{R}}}(\mathbb{R}) \times \prod_{j_{\mathbb{C}}} M_{d_{j_\mathbb{C}}/2}(\mathbb{C}) \times \prod_{k_{\mathbb{H}}} M_{d_{k_\mathbb{H}}/4}(\mathbb{H}) \ar[d,hook] \\
        \mathbb{C}[G] \ar[r,"\Phi"] & \prod_{i_{\mathbb{R}}} M_{d_{i_\mathbb{R}}}(\mathbb{C}) \times \prod_{j_{\mathbb{C}}} M_{d_{j_\mathbb{C}}/2}(\mathbb{C})\times M_{d_{j_\mathbb{C}}/2}(\mathbb{C}) \times \prod_{k_{\mathbb{H}}} M_{d_{k_\mathbb{H}}/2}(\mathbb{C}),
    \end{tikzcd}
\end{equation}
\end{widetext}
where $i_D$ takes values in the isomorphism classes of $D$-type simple modules. The vertical maps are induced by the canonical injections to the complexified spaces. In order to construct $\Phi^{-1}$ for $\mathbb{R}[G]$, we simply use $\Phi^{-1}$ for $\mathbb{C}[G]$ and restrict it to the subspace. To be specific, we have
\begin{itemize}
    \item If $f\in M_{d_{i_\mathbb{R}}}(\mathbb{R})$,
    \begin{equation}\label{inverse of phi real R-type}
        \Phi^{-1}(f)=\frac{d_{i_{\mathbb{R}}}}{|G|} \sum_{g \in G} \mathrm{Tr}\left(\rho_{i_{\mathbb{R}}}(g^{-1})f\right) \, g.
    \end{equation}

    \item If $f\in M_{d_{j_\mathbb{C}}/2}(\mathbb{C})$,
    \begin{equation}\label{inverse of phi real C-type}
        \begin{aligned}
            \Phi^{-1}(f)&=\frac{d_{j_\mathbb{C}}}{2|G|} \sum_{g \in G} \mathrm{Tr}\left( \rho_{j_\mathbb{C}}(g^{-1})f + \bar{\rho}_{j_\mathbb{C}}(g^{-1}) \bar{f} \right) \, g \\
            &=\frac{d_{j_\mathbb{C}}}{|G|} \sum_{g \in G} \mathrm{Tr}\left( \mathrm{Re}\left(\rho_{j_\mathbb{C}}(g^{-1})f\right)\right) \, g.
        \end{aligned}
    \end{equation}

    \item If $f\in M_{d_{k_\mathbb{H}}/4}(\mathbb{H})$, recall that the isomorphism $M_{d_{k_\mathbb{H}}/4}(\mathbb{H})\otimes_{\mathbb{R}}\mathbb{C} \cong M_{d_{k_\mathbb{H}}/2}(\mathbb{C})$ is given by mapping $I$, $J$, and $K$ to $-i\sigma_x$, $-i\sigma_y$, and $-i\sigma_z$, respectively. We denote its restriction to $M_{d_{k_\mathbb{H}}/4}(\mathbb{H})$ by $\phi_{k_\mathbb{H}}$. Then we have
    \begin{equation}\label{inverse of phi real H-type}
        \begin{aligned}
            \Phi^{-1}(f)&=\frac{d_{k_\mathbb{H}}}{2|G|} \sum_{g \in G} \mathrm{Tr}\left(\phi_{k_\mathbb{H}}\left(\rho_{k_\mathbb{H}}(g^{-1})f\right)\right) \, g \\
            &=\frac{d_{k_\mathbb{H}}}{|G|} \sum_{g \in G} \mathrm{Tr}\left(\mathrm{Re}\left(\rho_{k_\mathbb{H}}(g^{-1})f\right)\right) \, g.
        \end{aligned}
    \end{equation}
    Notice that $\mathrm{Tr}\circ \phi_{k_\mathbb{H}}=2\mathrm{Tr}\circ\mathrm{Re}$, where $\mathrm{Re}$ denotes the projection that maps a quaternion number to its coefficient of the unit element.
\end{itemize}

There is a more technical way to derive $\Phi^{-1}$ using the reduced trace. Notice that there are three different kinds of traces for a $D^{\mathrm{op}}$-linear map $f\in \mathrm{End}_{D^{\mathrm{op}}}\left( V \right) \cong M_d(D) $, where $D$ is a division $\mathbb{K}$ algebra and $V$ is a $D^{\mathrm{op}}$ module.
\begin{enumerate}
    \item The trace of $f$ as a $\mathbb{K}$-linear map, i.e., $\mathrm{Tr}_{\mathbb{K}}(f)$.
    \item The trace of $f$ as a $Z$-linear map, where $Z$ is the center of $D$, i.e., $\mathrm{Tr}_Z(f)$.
    \item The reduced trace of $f$ as a $D$-linear map, i.e., $\mathrm{Trd}(f) = \mathrm{Tr}_{\overline{Z}}\left(\phi(f\otimes_{Z} 1)\right)$, where $\phi$ is an isomorphism from $M_d(D)\otimes_{Z}\overline{Z}$ to $M_{dm}(\overline{Z})$, $\overline{Z}$ is the algebraic closure of the center $Z$, $m$ is called the Schur index of $D$. $\mathrm{Trd}(f)$ takes values in $Z$.
\end{enumerate}
These traces are related by the formulas
\begin{equation}
    \mathrm{Tr}_{\mathbb{K}}(f) = \mathrm{Tr}_{Z/\mathbb{K}} \circ \mathrm{Tr}_{Z}(f),
\end{equation}
where $\mathrm{Tr}_{Z/\mathbb{K}}$ is the relative trace of the field extension, and
\begin{equation}
    \mathrm{Tr}_Z(f) = m \cdot \mathrm{Trd}(f),
\end{equation}
where $m$ is the Schur index of $D$.

For $f$ an element of the component $M_{d_i'}(D_i)$, we have the formula
\begin{equation}
    \Phi^{-1}(f) = \frac{d_i'm_i}{|G|}\sum_{g \in G} \mathrm{Tr}_{Z/\mathbb{K}} \circ \mathrm{Trd}(\rho_i(g^{-1})f)\, g,
\end{equation}
where $m_i$ is the Schur index of $D_i$.
To prove this formula, we assume $\Phi^{-1}(f) = \sum_{g\in G} a_g\, g \in \mathbb{R}[G]$. Notice that we have
\begin{equation}
    \begin{aligned}
        a_g & = \frac{1}{|G|}\, \mathrm{Tr}_{\mathbb{K}}\left(g^{-1}\Phi^{-1}(f)\cdot -\right) \\
        & = \frac{1}{|G|}\, \mathrm{Tr}_{\mathbb{K}}\left(\rho_i(g^{-1})f\cdot -\right) \\
        & = \frac{d_i'm_i}{|G|}\, \mathrm{Tr}_{Z/\mathbb{K}} \circ \mathrm{Trd}\left(\rho_i(g^{-1})f\right).
    \end{aligned}
\end{equation}

Finally, when $\mathbb{K}=\mathbb{R}$, notice that for $D_i = \mathbb{C}$, we have $\mathrm{Tr}_{Z/\mathbb{K}} = 2\mathrm{Re}$; for $D_i = \mathbb{H}$, we have $\mathrm{Trd} = 2\mathrm{Re}$. In summary, for $f\in M_{d_i'}(D_i)$, we have
\begin{equation}\label{eq:inverse of Phi real}
    \Phi^{-1}(f) = \frac{d_i}{|G|} \sum_{g \in G} \mathrm{Tr}\left(\mathrm{Re}\left(\rho_{i}(g^{-1})f\right)\right) \, g.
\end{equation}

\subsubsection{Characters and great orthogonality theorem}\label{subsubsection:Characters and Great Orthogonality Theorem real}

We first prove that the Frobenius--Schur indicator can distinguish the three $D$ types. It is well known that there is a $G$-invariant inner product on any $\mathbb{C}[G]$ module, since we can start with an arbitrary inner product and then average it over $G$. The inner product induces an isomorphism $\bar{V}\cong V^*$, where $V^*$ is the dual module. Thus, we can replace $V\cong \bar{V}$ by $V\cong V^*$. The isomorphism $V\cong V^*$ is equivalent to the existence of a $G$-invariant nondegenerate bilinear form $B$. Notice that $B$ is unique up to a scalar, because of the Schur lemma. $B$ can be decomposed into $S+A$, where $S$ is symmetric and $A$ is skew symmetric; both $S$ and $A$ are $G$ invariant. Thus, $B$ must equal either $S$ or $A$. It can be proved that if $\bar{\alpha}\circ \alpha$ is positive, then $B=S$. If $\bar{\alpha}\circ \alpha$ is negative, then $B=A$. Finally, we use the following identities:
\begin{equation}
    2\chi_S(g)=\chi(g)^2+\chi(g^2) \quad \textrm{and} \quad 2\chi_A(g)=\chi(g)^2-\chi(g^2),
\end{equation}
where $\chi_S(g)$ and $\chi_A(g)$ are the characters of the symmetric and skew-symmetric subspaces of $V\otimes V$, respectively, which are the duals of the $G$-invariant symmetric and skew-symmetric bilinear forms on $V$, respectively. Subtracting the two identities above and taking the inner product with the trivial representation, we get $FS(\rho)=s-a$, where $s$ is the number of trivial representations contained in the symmetric subspace of $V\otimes V$. We have $s=1$ if $B$ exists and equals $S$; otherwise, $s=0$. Similarly, $a$ is the number of trivial representations contained in the skew-symmetric subspace of $V\otimes V$. We have $a=1$ if $B$ exists and equals $A$; otherwise, $a=0$. This completes the proof.

The character $\xi$ of a real representation is the same as the character of its complexification. Recall that there is a correspondence between the simple $\mathbb{R}[G]$ modules and simple $\mathbb{C}[G]$ modules. Let $i$ denote the $\mathbb{C}[G]$ module obtained from $i_{\mathbb{R}}$; a similar rule applies to $\mathbb{C}$ and $\mathbb{H}$. We have
\begin{itemize}
    \item $\xi_{i_{\mathbb{R}}}=\chi_i$.

    \item $\xi_{j_{\mathbb{C}}}=\chi_j+\bar{\chi}_j$.

    \item $\xi_{k_{\mathbb{H}}}=2\chi_k$.
\end{itemize}

Finally, the great orthogonality theorem is obtained by setting $f$ to be matrix entries.
\begin{itemize}
    \item If $e_{ab}^{i_\mathbb{R}}\in M_{d_{i_\mathbb{R}}}(\mathbb{R})$,
    \begin{equation}
        \Phi^{-1}(e_{ab}^{i_\mathbb{R}})=\frac{d_{i_{\mathbb{R}}}}{|G|} \sum_{g \in G} \rho_{i_{\mathbb{R}}}(g^{-1})_{ba} \, g.
    \end{equation}

    \item If $e_{ab}^{j_\mathbb{C}}\in M_{d_{j_\mathbb{C}}/2}(\mathbb{C})$,
    \begin{equation}
        \Phi^{-1}(e_{ab}^{j_\mathbb{C}})=\frac{d_{j_\mathbb{C}}}{|G|} \sum_{g \in G} \mathrm{Re}(\rho_{j_\mathbb{C}}(g^{-1})_{ba})  \, g,
    \end{equation}
    \begin{equation}
        \Phi^{-1}(i\cdot e_{ab}^{j_\mathbb{C}})=-\frac{d_{j_\mathbb{C}}}{|G|} \sum_{g \in G} \mathrm{Im}(\rho_{j_\mathbb{C}}(g^{-1})_{ba})  \, g.
    \end{equation}

    \item If $e_{ab}^{k_\mathbb{H}}\in M_{d_{k_\mathbb{H}}/4}(\mathbb{H})$,
    \begin{equation}
        \Phi^{-1}(e_{ab}^{k_\mathbb{H}})=\frac{d_{k_\mathbb{H}}}{|G|} \sum_{g \in G} \mathrm{Re}\left(\rho_{k_\mathbb{H}}(g^{-1})_{ba}\right) \, g,
    \end{equation}
    \begin{equation}
        \Phi^{-1}(I\cdot e_{ab}^{k_\mathbb{H}})=\frac{d_{k_\mathbb{H}}}{|G|} \sum_{g \in G} \mathrm{Re}\left(I\cdot \rho_{k_\mathbb{H}}(g^{-1})_{ba}\right) \, g.
    \end{equation}
    $J$ and $K$ are similar.
\end{itemize}
All orthogonality relations can be extracted from the identities above; for example, we have
\begin{equation}
    \frac{d_{j_\mathbb{C}}}{|G|} \sum_{g \in G} \mathrm{Re}(\rho_{j_\mathbb{C}}(g^{-1})_{ba})  \, \rho_{j_\mathbb{C}}(g)_{cd}=\delta_{ac}\delta_{bd}.
\end{equation}
As a direct consequence, we have
\begin{equation}
    \frac{d_{j_\mathbb{C}}}{|G|} \sum_{g \in G} \mathrm{Re}(\rho_{j_\mathbb{C}}(g^{-1})_{ba})  \, \mathrm{Re}(\rho_{j_\mathbb{C}}(g)_{cd})=\delta_{ac}\delta_{bd},
\end{equation}
\begin{equation}
    \frac{d_{j_\mathbb{C}}}{|G|} \sum_{g \in G} \mathrm{Re}(\rho_{j_\mathbb{C}}(g^{-1})_{ba})  \, \mathrm{Im}(\rho_{j_\mathbb{C}}(g)_{cd})=0.
\end{equation}

\subsubsection{Constructing isomorphisms between two modules}\label{subsubsection:Constructing isomorphisms real}

Let $V$ and $W$ be isomorphic simple $\mathbb{R}[G]$ modules with $D^{\mathrm{op}}$ type. We construct a $\mathbb{R}[G]$ isomorphism explicitly. As in Appendix \ref{subsubsection:Constructing isomorphisms}, we have an isomorphism $\psi:\mathrm{End}_{D^{\mathrm{op}}}(V)\xrightarrow{\Phi^{-1}} \mathbb{R}[G] \xrightarrow{\Phi} \mathrm{End}_{D^{\mathrm{op}}}(W)$. $\mathrm{End}_{D^{\mathrm{op}}}(V)$ and $\mathrm{End}_{D^{\mathrm{op}}}(W)$ are isomorphic to $M_d(D)$.

The construction for $\mathbb{R}$ type and $\mathbb{H}$ type is analogous to that in Appendix \ref{subsubsection:Constructing isomorphisms}, except that we use \cref{eq:inverse of Phi real} for $\Phi^{-1}$. Let $f\in \mathrm{End}_{D}(V)$. Since $M_d(D)$ is a central simple algebra for $D=\mathbb{R}$ and $\mathbb{H}$, we have $\psi=\alpha f \alpha^{-1}$ and $\alpha:V\to W$ is a $D$-linear map. For $\mathbb{H}$-type, $\alpha$ is only unique up to multiplication by an element of $D$. Other choices of $\alpha$ are given by composing $\alpha$ with an automorphism of $V$.

For $\mathbb{C}$ type, $M_d(\mathbb{C})$ is no longer a central simple $\mathbb{R}$ algebra. The automorphism of $M_d(\mathbb{C})$ as an $\mathbb{R}$ algebra is not always an inner automorphism; however, it is an inner automorphism up to complex conjugation. Since the automorphism of $M_d(\mathbb{C})$ induces an automorphism of its center, which must be either the identity or complex conjugation, we still have $\psi=\alpha f \alpha^{-1}$; however, $\alpha$ may not be a $\mathbb{C}$-linear map, and it could also be a $\mathbb{C}$-antilinear map. The construction in Appendix \ref{subsubsection:Constructing isomorphisms} is still valid, except that we need to use \cref{eq:inverse of Phi real} for $\Phi^{-1}$. Other choices of $\alpha$ are given by composing $\alpha$ with an automorphism of $V$.

\subsection{\texorpdfstring{Twisted group algebra $\mathbb{K}[G,\omega]$}{Twisted group algebra K[G,Omega]}}\label{subsection:Twisted Group Algebra}

\subsubsection{Definition}

Let $\omega\in H^2(G,\mathbb{K}^{\times})$ be a two-cocycle, where $\mathbb{K}^{\times}$ is the group of invertible elements of $\mathbb{K}$. For $\mathbb{K}=\mathbb{R}$ or $\mathbb{C}$, without loss of generality, $\mathbb{K}^{\times}$ can be replaced by $U(1)$ and $\mathbb{Z}_2=\{1,-1\}$, respectively. Moreover, $\omega\in H^2(G,\mathbb{K}^{\times})$ can be chosen to satisfy $\omega(1,g)=\omega(g,1)=1$ for all $g\in G$.

The twisted group algebra $\mathbb{K}[G,\omega]$ is spanned by its basis $\{g\}_{g\in G}$, and the multiplication is given by $g\cdot h=\omega(g,h)gh$. Notice that $\omega\in H^2(G,\mathbb{K}^{\times})$ also determines a central extension of $G$ by $\mathbb{K}^{\times}$. Moreover, for $\mathbb{K}=\mathbb{C}$, the cocycle $\omega$ takes values in $\mathbb{Z}_N$ for some $N$, where $\mathbb{Z}_N$ is viewed as a subgroup of $U(1)$ with the injection $n\mapsto e^{2\pi i n/N}$. Thus, we have a group extension
\begin{equation}
    1\to \mathbb{Z}_N \to \tilde{G} \to G \to 1.
\end{equation}
The group elements of $\tilde{G}$ have the form $(a,g)$, where $a\in \mathbb{Z}_N$ and $g\in G$, and the multiplication is given by $(a,g)\cdot (b,h)=(ab\,\omega(g,h),gh)$. We have the following isomorphism of algebras:
\begin{equation}\label{eq:Decomposition of twisted group algebra}
    \mathbb{K}[\tilde{G}]\cong \mathbb{K}[G]\times \mathbb{K}[G,\omega] \times \mathbb{K}[G,\omega^2] \times \cdots \times \mathbb{K}[G,\omega^{N-1}].
\end{equation}
The isomorphism is given by the assignment $(a,g)\mapsto (g,ag,a^2g,\cdots,a^{N-1}g)$.

\subsubsection{The Wedderburn--Artin decomposition}

Using \cref{eq:Decomposition of twisted group algebra}, it is easy to get the Wedderburn--Artin decomposition of $\mathbb{K}[G,\omega]$ from the Wedderburn--Artin decomposition of $\mathbb{K}[\tilde{G}]$. Assume that $\mathbb{K}[\tilde{G}]$ has been decomposed into a direct sum of matrix algebras; that is, all irreducible representations of $\mathbb{K}[\tilde{G}]$ are known. If $\rho$ is a component of $\mathbb{K}[\tilde{G}]$ with $\rho:(a,1)\mapsto a^n$, then it must be a component of $\mathbb{K}[G,n\omega]$. Thus, we have the decomposition $\Phi=(\rho_{i})$
\begin{equation}\label{Maschke's theorem twisted}
    \mathbb{K}[G,\omega]\cong \prod_{i} M_{d_{i}'}(D_i).
\end{equation}
Here, $i$ takes values in the isomorphism classes of simple $\mathbb{K}[\tilde{G}]$ modules such that $\tilde{\rho}_{i}(a,1)=a$ with $\tilde{\rho}_{i}$ being the corresponding component of $\mathbb{K}[\tilde{G}]$. The component $\rho_{i}$ is the restriction of $\tilde{\rho}_{i}$ to the subalgebra $\mathbb{K}[G,\omega]$. The dimension is $d_{i}'=d_{i}/\mathrm{dim}(D_i)$, where $d_{i}$ is the dimension of the simple module $i$.

\subsubsection{\texorpdfstring{The inverse of $\Phi$}{The inverse of Phi}}

The inverse $\Phi^{-1}$ is obtained by $M_{d_{i}'}(D_i)\to \mathbb{K}[\tilde{G}] \to \mathbb{K}[G,\omega]$. To be explicit, for $\mathbb{K}=\mathbb{C}$,
\begin{equation}\label{eq:inverse of Phi twisted complex}
    \begin{aligned}
        \Phi^{-1}(f)&=\frac{d_{i}}{|\tilde{G}|} \sum_{(a,g) \in \tilde{G}} \mathrm{Tr}\left(\tilde{\rho}_{i}\left((a,g)^{-1}\right)f\right) \, ag\\
        &=\frac{d_i}{|G|} \sum_{g \in G} \omega(g,g^{-1})^{-1}\mathrm{Tr}\left(\rho_i(g^{-1})f\right) \, g,
    \end{aligned}
\end{equation}
for $\mathbb{K}=\mathbb{R}$,
\begin{equation}\label{eq:inverse of Phi twisted real}
    \begin{aligned}
        \Phi^{-1}(f)&=\frac{d_{i}}{|\tilde{G}|} \sum_{(a,g) \in \tilde{G}} \mathrm{Tr}\left(\mathrm{Re}\left(\tilde{\rho}_{i}\left((a,g)^{-1}\right)f\right)\right) \, ag\\
        &=\frac{d_{i}}{|G|} \sum_{g \in G} \omega(g,g^{-1})^{-1}\mathrm{Tr}\left(\mathrm{Re}\left(\rho_{i}(g^{-1})f\right)\right) \, g,
    \end{aligned}
\end{equation}
where $(a,g)^{-1}=(a^{-1}\omega(g,g^{-1})^{-1},g^{-1})$. Notice that we have $\omega(g^{-1},g)\omega(g,1)=\omega(1,g)\omega(g,g^{-1})$, which means $\omega(g,g^{-1})=\omega(g^{-1},g)$.

\subsubsection{Characters and great orthogonality theorem}

The Frobenius--Schur indicator of a $\mathbb{C}[G,\omega]$ module is given by
\begin{equation}\label{eq:Frobenius--Schur indicator twisted}
    \begin{aligned}
        & \frac{1}{|\tilde{G}|} \sum_{(a,g) \in \tilde{G}} \chi((a,g)^2) \\
        = & \frac{1}{|G|} \sum_{g \in G} \omega(g,g)\chi(g^2),
    \end{aligned}
\end{equation}
where we assume that $\omega$ takes values in $\mathbb{Z}_2$, and $\chi$ is the character of the module.

Using the character theory of $\mathbb{C}[\tilde{G}]$, a direct computation shows that the simple character $\chi_i(g)=\mathrm{Tr}\left(\rho_i(g)\right)$ of $\mathbb{C}[G,\omega]$ is an orthonormal basis under the inner product
\begin{equation}\label{eq:inner product of characters twisted}
    (\chi_i,\chi_j)=\frac{1}{|G|}\sum_{g\in G}\omega(g,g^{-1})^{-1}\chi_i(g^{-1})\chi_j(g)=\delta_{ij}.
\end{equation}
Notice that we have $\omega(g,g^{-1})^{-1}\chi_i(g^{-1})=\bar{\chi}_i(g)$.

The great orthogonality theorem is obtained by setting $f$ to be matrix entries in \cref{eq:inverse of Phi twisted complex}.

\subsubsection{Constructing isomorphisms between two modules}\label{subsubsection:constructing isomorphisms between two modules twisted}

Apart from the use of \cref{eq:inverse of Phi twisted real,eq:inverse of Phi twisted complex}, the approach remains unchanged.

%% file: a4.graded_group_algebra.tex
\section{Graded Group Algebras}\label{section:graded group algebra}

\subsection{\texorpdfstring{$\mathbb{Z}_2$-graded group algebra over $\mathbb{R}$ and $\mathbb{C}$}{Z2-graded group algebra over R and C}}\label{subsection:Z2 graded group algebra over R and C}

\subsubsection{The Wedderburn--Artin decomposition}\label{subsubsection:decomposition into matrix algebras for graded group algebras}

A $\mathbb{Z}_2$-graded group $G$ is a group equipped with a group homomorphism $s:G\to\mathbb{Z}_2$. $G_0$ denotes the kernel of $s$, and $G_1$ is the preimage of $1\in \mathbb{Z}_2$. The group algebra $\mathbb{K}[G]$ is a $\mathbb{Z}_2$-graded algebra. We assume that $G_1$ is nonempty, and we fix a grading 1 group element $g_1\in G_1$.

Notice that $\mathbb{K}[G]$ is a $\mathbb{Z}_2$-crossed product algebra. According to \cref{prop:equivalence of categories between right A-modules and A_0-modules}, the representation theory of $\mathbb{K}[G]$ is determined by the representation theory of $\mathbb{K}[G_0]$. Assume that the Wedderburn--Artin decomposition of $\mathbb{K}[G_0]$  is known, that is, there is an isomorphism $\Phi_0=(\rho_{i_0})$,
\begin{equation}\label{eq:decomposition of K[G_0]}
    \mathbb{K}[G_0]\cong \prod_{i_{0}} M_{d_{i_{0}}'}(D_{i_0}).
\end{equation}
Here, $i_{0}$ takes values in the isomorphism classes of simple $\mathbb{K}[G_0]$-modules. $d_{i_{0}}'=d_{i_{0}}/\mathrm{dim}(D_{i_0})$, where $d_{i_{0}}$ is the dimension of the simple module.

As in \cref{thm:structure theorem of Gr-crossed product algebra}, the $\mathbb{Z}_2$-graded Wedderburn--Artin decomposition of $\mathbb{K}[G]$ is given by an isomorphism of $\mathbb{Z}_2$-graded algebras $\Phi=(\rho_{i})$,
\begin{equation}\label{eq:decomposition Phi of graded group algebra}
    \mathbb{K}[G]\cong \prod_{i} M_{d_{i}'}(D_i)(\bar{c}_{i}).
\end{equation}
Here, $i$ takes values in the graded equivalence classes of simple graded $\mathbb{K}[G]$ modules. $d_{i}'=|\Gr/E_{i}|d_{i_{0}}'=d_{i}/\mathrm{dim}(D_i)$, where $d_{i}$ is the dimension of the $\mathbb{Z}_2$-graded module, $E_{i}$ and $\bar{c}_{i}$ are determined by $D_i$ as in \cref{thm:structure theorem of Gr-crossed product algebra}.

To get $\rho_i$, we fix a simple $\mathbb{K}[G_0]$-module $V$ corresponding to $i_{0}$. By \cref{prop:equivalence of categories between right A-modules and A_0-modules}, $\mathbb{K}[G]\otimes_{G_0}V$ is the simple $\mathbb{K}[G]$ module corresponding to $i$. For simplicity, we omit $i$ in the following. It remains to determine $D$ and a $G$-equivariant $D$ structure on $\mathbb{K}[G]\otimes_{G_0}V$.

To be explicit, let $\rho_0$ be the representation of $V$. Then the representation of $\mathbb{K}[G]\otimes_{G_0}V\cong 1\otimes_{G_0} V\oplus g_1\otimes_{G_0} V\cong V\oplus V$ is given by
\begin{equation}\label{eq:representation of K[G] otimes N grading 0}
    \rho(g)=
    \begin{bmatrix}
        \rho_0(g) & 0 \\
        0 & \rho_0(g_1^{-1}gg_1) \\
    \end{bmatrix}, \textrm{ for } g\in G_0,
\end{equation}
\begin{equation}\label{eq:representation of K[G] otimes N grading 1}
    \rho(g)=
    \begin{bmatrix}
        0 & \rho_0(gg_1) \\
        \rho_0(g_1^{-1}g)  & 0 \\
    \end{bmatrix}, \textrm{ for } g\in G_1,
\end{equation}
recall that $g_1\in G_1$ is the fixed grading 1 group element.

The $\mathbb{Z}_2$-graded division algebra $D$ is the endomorphism algebra of $\mathbb{K}[G]\otimes_{G_0}V$, and the action of the subalgebra $D_0$ is already known, which is the diagonal action on $V\oplus V$. To get $D_1$, we only need to construct an invertible element in $D_1$ with grading 1. Notice that $D_1$ is nonzero iff there exists an isomorphism $f:V\cong V^{g_1}$, where $V^{g}$ is a $\mathbb{K}[G_0]$ module defined by $h\cdot_{V^{g}} v=(g^{-1}hg)\cdot v$. Recall from \cref{subsection:Gr-Crossed Product Algebra} that we defined a right module version of $V^{g}$ in a previous section. The grading 1 element is given by
\begin{equation}\label{eq:grading 1 element of D}
    \theta=
    \begin{bmatrix}
        0 & \rho_0(g_1^2)f \\
        f & 0
    \end{bmatrix}.
\end{equation}
It is easy to show that $\theta$ is a $\mathbb{K}[G]$-module map.

We choose an $f$ that is compatible with the $D_0$-structure in the following sense. For $\mathbb{K}=\mathbb{C}$, $f$ always commutes with $D_0$. Using the construction in Appendix \ref{subsubsection:Constructing isomorphisms real}, for $\mathbb{K}=\mathbb{R}$ and $D_0=\mathbb{R}$ or $\mathbb{H}$, we can get an $f$ that commutes with $D_0$. For $\mathbb{K}=\mathbb{R}$ and $D_0=\mathbb{C}$, we can get an $f$ that is either linear or antilinear.

The $\mathbb{Z}_2$-graded division algebras over $\mathbb{R}$ and $\mathbb{C}$ are listed in \cref{table:The isomorphisms between Z2-graded division algebras and Clifford algebras}. For $\mathbb{K}=\mathbb{C}$, we have two cases $D\cong\mathbb{C}$ for $V \not \cong V^{g_1}$, and $D\cong\mathbb{C}l^1$ for $f:V\cong V^{g_1}$. Notice that $\rho_0(g_1^2)f^2$ is a $\mathbb{C}[G_0]$ map; therefore, it is the identity up to a scalar. We can redefine $f$ such that $\theta^2=\mathrm{id}$.

For $\mathbb{K}=\mathbb{R}$, $D_0=\mathbb{R}$ or $\mathbb{H}$, the only difference is that the redefinition of $f$ cannot change the sign of $\theta^2$; therefore, we have two choices $\theta^2=\pm 1$. For $D_0=\mathbb{C}$, $f$ is either a linear map or an antilinear map. For both cases, we have $\theta^2\in Z(D)=\mathbb{C}$. However, if $f$ is linear, then we can choose an $f$ such that $\theta^2=1$. If $f$ is antilinear, then we have two choices $\theta^2=\pm 1$. This finishes the construction of $D$ and of the $G$-equivariant $D$ structure on $\mathbb{K}[G]\otimes_{G_0}V$, which gives $\rho_{i}$.

In summary, to construct all $\rho_{i}$, we start with $\rho_{i_{0}}$ of a simple $\mathbb{K}[G_0]$ module $V$. The induced $\mathbb{K}[G]$ module $\mathbb{K}[G]\otimes_{G_0}V$ is constructed by \cref{eq:representation of K[G] otimes N grading 0,eq:representation of K[G] otimes N grading 1}.
\begin{enumerate}
    \item If $V\not\cong V^{g_1}$, then $D\cong D_0$.
    \item If $V\cong V^{g_1}$, $D$ is constructed by \cref{eq:grading 1 element of D}, where $f:V\cong V^{g_1}$ is constructed using the methods in Appendixes \ref{subsubsection:Constructing isomorphisms} and \ref{subsubsection:Constructing isomorphisms real}.
    \begin{enumerate}
        \item For $\mathbb{K}=\mathbb{C}$, $D\cong \mathbb{C}l^1$.
        \item For $\mathbb{K}=\mathbb{R}$ and
        \begin{enumerate}
            \item $D_0=\mathbb{R}$, if $\theta^2=1$, then $D\cong Cl^{0,1}$. If $\theta^2=-1$, then $D\cong Cl^{1,0}$.
            \item $D_0=\mathbb{C}$, if $f$ is linear, then $D\cong \mathbb{C}l^1$. If $f$ is antilinear, then $D\cong Cl^{0,2}$ for $\theta^2=1$, and $D\cong Cl^{2,0}$ for $\theta^2=-1$.
            \item $D_0=\mathbb{H}$, if $\theta^2=1$, then $D\cong Cl^{3,0}$. If $\theta^2=-1$, then $D\cong Cl^{0,3}$.
        \end{enumerate}
    \end{enumerate}
\end{enumerate}
The characters suffice to distinguish all these cases, as will be shown in Appendix \ref{subsubsection:Character theory graded}.

\subsubsection{\texorpdfstring{The inverse of $\Phi$}{The inverse of Phi}}\label{subsubsection:Inverse of Phi graded}

Let $f$ be an element in $\mathrm{End}_{D^{\mathrm{op}}}(V) \cong M_d(D)(\bar{c})$, where $D$ is a $\Gr$-graded division algebra, and $V$ is a graded $D^{\mathrm{op}}$ module. We can take the trace of $f$ as a $\mathbb{K}$-linear map, i.e., $\mathrm{Tr}_{\mathbb{K}}(f)$. Notice that the matrix entries in $D_c$ for $c\neq 0$ do not contribute to the value of the trace. Therefore, we have $\mathrm{Tr}_{\mathbb{K}}(f) = \mathrm{dim}(D/D_0) \mathrm{Tr}_{\mathbb{K}}(f_0)$, where $f_0\in M_d(D_0)$ is obtained from $f$ by projecting every matrix entry to $D_0$.

Following the method in Appendix \ref{subsubsection:The Inverse of Phi real}, for $f\in M_{d_{i}'}(D_i)(\bar{c}_{i})$ and $\mathbb{K} = \mathbb{R}$, we have
\begin{equation}\label{eq:inverse of Phi graded}
    \Phi^{-1}(f) = \frac{d_{i}}{|G|} \sum_{g \in G} \mathrm{Tr}\left(\mathrm{Re}\left(\rho_{i}(g^{-1})f\right)\right) \, g,
\end{equation}
where $\mathrm{Re}$ is the projection from $D$ to $D_0$ composed with the $\mathrm{Re}$ for $D_0$, i.e., $\mathrm{Re}$ is the projection from $D$ to the subspace $\mathbb{R}$ spanned by the unit. For $\mathbb{K} = \mathbb{C}$, $\mathrm{Re}$ is just the projection from $D$ to $D_0 = \mathbb{C}$. This formula actually works for any grading group $\Gr$.

\subsubsection{Character theory}\label{subsubsection:Character theory graded}

For $\mathbb{K}=\mathbb{C}$, let $V$ be a simple $\mathbb{C}[G_0]$ module, and let $\chi$ be its character. To determine whether $V\cong V^{g_1}$, we use the inner product formula of characters,
\begin{equation}\label{eq:criteria of N=N^g1}
    \frac{1}{|G_0|} \sum_{g \in G_0} \chi(g_1^{-1}g^{-1}g_1) \, \chi(g).
\end{equation}
If the inner product is 1, then $V\cong V^{g_1}$, otherwise $V\not\cong V^{g_1}$.

For $\mathbb{K}=\mathbb{R}$, in what follows, we determine $D$ using characters. Let $W$ be a simple $\mathbb{R}[G_0]$ module, and let $\chi$ be its character. We still use \cref{eq:criteria of N=N^g1} to determine whether $W\cong W^{g_1}$, but the value is not always 0 or 1. If $W\not\cong W^{g_1}$, the Frobenius--Schur indicator \cref{Frobenius--Schur indicator} of $W$ is enough to determine $D$. If $W\cong W^{g_1}$, we also need to use the Frobenius--Schur indicator of $\mathbb{R}[G]\otimes_{G_0}W$. According to \cref{table:Clifford algebras,table:The isomorphisms between Z2-graded division algebras and Clifford algebras}, $D$ is uniquely determined by these data. Using \cref{eq:representation of K[G] otimes N grading 0,eq:representation of K[G] otimes N grading 1}, we obtain
\begin{equation}\label{eq:Frobenius--Schur indicator of C[G] tensor G0 N}
    FS(\mathbb{R}[G]\otimes_{G_0}W)=\frac{1}{|G_0|}\sum_{g\in G} \chi(g^2).
\end{equation}
Notice that $FS(\mathbb{R}[G]\otimes_{G_0}W)=FS(W)+w$, where
\begin{equation}\label{eq:W}
    w=\frac{1}{|G_0|}\sum_{g\in G_1} \chi(g^2).
\end{equation}
Equivalently, we use $FS(W)$ and $w$ to replace $FS(W)$ and $FS(\mathbb{R}[G]\otimes_{G_0}W)$. They have the advantage that it suffices to know whether $FS(W)$ and $w$ are positive, negative, or zero. We record the results in \cref{table:character theory and Z2-graded type of module from real module}.

For a simple $\mathbb{R}[G_0]$ module $W$, let $V$ be the corresponding simple $\mathbb{C}[G_0]$ module, as discussed in Appendix \ref{subsubsection:The Relationship Between R[G] and C[G] Modules}. A similar argument can be used to derive the entries of \cref{table:character theory and Z2-graded type of module}.

\begin{table}
    \centering
    \begin{tabular}{|c|c|c|c|}
        \hline
        $V\cong V^{g_1}$ & & & $D$ over $\mathbb{C}$ \\
        \hline
        0 & & & $\mathbb{C}$ \\
        1 & & & $\mathbb{C}l^{1}$ \\
        \hline
        \hline
        $V\cong V^{g_1}$ & $FS$ & $w$ & $D$ over $\mathbb{R}$ \\
        \hline
        0 & 1 & 0 & $\mathbb{R}$ \\
        0 & 0 & 0 & $\mathbb{C}$ \\
        0 & -1 & 0 & $\mathbb{H}$ \\
        1 & 1 & 1 & $Cl^{0,1}$ \\
        1 & 1 & -1 & $Cl^{1,0}$ \\
        1 & 0 & 0 & $\mathbb{C}l^{1}$ \\
        1 & 0 & 1 & $Cl^{0,2}$ \\
        1 & 0 & -1 & $Cl^{2,0}$ \\
        1 & -1 & -1 & $Cl^{3,0}$ \\
        1 & -1 & 1 &  $Cl^{0,3}$ \\
        \hline
    \end{tabular}
    \caption{The first column is calculated by \cref{eq:criteria of N=N^g1} to determine whether $V\cong V^{g_1}$ or $\bar{V}^{g_1}$. The second column is the Frobenius--Schur indicator \cref{Frobenius--Schur indicator} of $V$. The third column is $w$ \cref{eq:W} using the characters of $V$.}\label{table:character theory and Z2-graded type of module}
\end{table}

\begin{table}
    \centering
    \begin{tabular}{|c|c|c|c|}
        \hline
        $W\cong W^{g_1}$ & $FS$ & $w$ & $D$ over $\mathbb{R}$ \\
        \hline
        0 & 1 & 0 & $\mathbb{R}$ \\
        0 & 0 & 0 & $\mathbb{C}$ \\
        0 & -2 & 0 & $\mathbb{H}$ \\
        1 & 1 & 1 & $Cl^{0,1}$ \\
        1 & 1 & -1 & $Cl^{1,0}$ \\
        2 & 0 & 0 & $\mathbb{C}l^{1}$ \\
        2 & 0 & 2 & $Cl^{0,2}$ \\
        2 & 0 & -2 & $Cl^{2,0}$ \\
        4 & -2 & -2 & $Cl^{3,0}$ \\
        4 & -2 & 2 &  $Cl^{0,3}$ \\
        \hline
    \end{tabular}
    \caption{The first column is calculated by \cref{eq:criteria of N=N^g1}. The second column is the Frobenius--Schur indicator \cref{Frobenius--Schur indicator} of $W$. The third column is $w$ \cref{eq:W} using the characters of $W$.}\label{table:character theory and Z2-graded type of module from real module}
\end{table}

\subsection{\texorpdfstring{Twisted group algebra $\mathbb{C}[G,\omega,s]$ with antilinear group elements}{Twisted group algebra C[G,omega,s] with antilinear group elements}}\label{subsection:twisted group algebra with antilinear elements}

\subsubsection{Definition}\label{subsubsection:definition of twisted group algebra with antilinear elements}

Let $\mathbb{C}[G,\omega,s]$ be the twisted group algebra with antilinear elements, following the notation in \cref{subsection:symmetry of the free fermionic systems}. The cocycle $\omega\in H^2(G,U(1))$ defines a $\mathbb{Z}_2$-graded group extension of $G$ by $U(1)$. If $G$ is finite, $\omega$ takes values in $\mathbb{Z}_N$ for some $N$, where $\mathbb{Z}_N$ is viewed as a subgroup of $U(1)$ determined by the map $n\mapsto e^{2\pi i n/N}$, i.e.,
\begin{equation}
    \begin{tikzcd}
        1 \ar[r] & \mathbb{Z}_N \ar[r]\ar[d, equals] & \tilde{G}_0 \ar[r]\ar[d, hook] & G_0 \ar[r]\ar[d, hook] & 1\\
        1 \ar[r] & \mathbb{Z}_N \ar[r] & \tilde{G} \ar[r] & G \ar[r] & 1
    \end{tikzcd}
\end{equation}
The group elements of $\tilde{G}$ are $(a,g)$ for $a\in \mathbb{Z}_N$ and $g\in G$. The multiplication is defined by $(a,g)\cdot (b,h)=(a\,{}^gb\,\omega(g,h),gh)$. Notice that $\tilde{G}$ is also a $\mathbb{Z}_2$-graded group; we denote the grading map by $\tilde{s}$. We have the following isomorphism of algebras:
\begin{equation}
    \begin{aligned}
        \mathbb{C}[\tilde{G},1,\tilde{s}] & \cong \mathbb{C}[G,1,s] \times \mathbb{C}[G,\omega,s]\\
        & \phantom{{}\cong{}} \quad \times \mathbb{C}[G,\omega^2,s] \times \cdots \times \mathbb{C}[G,\omega^{N-1},s].
    \end{aligned} 
\end{equation}
The isomorphism is given by the assignment $(a,g)\mapsto (g,ag,a^2g,\cdots,a^{N-1}g)$.

\subsubsection{The Wedderburn--Artin decomposition}\label{subsubsection:decomposition into matrix algebras twisted with antilinear}

We follow a method similar to that in Appendix \ref{subsection:Z2 graded group algebra over R and C}. The grading 0 subspace of $\mathbb{C}[G,\omega,s]$ is a twisted group algebra $\mathbb{C}[G_0,\omega]$ without antilinear group elements. Assume that we already have the Wedderburn--Artin decomposition of $\mathbb{C}[G_0,\omega]$ using the methods in Appendix \ref{subsection:Twisted Group Algebra}. Since $\mathbb{C}[G,\omega,s]$ is a $\mathbb{Z}_2$-crossed product algebra, the $\mathbb{Z}_2$-graded simple modules over $\mathbb{C}[G,\omega,s]$ are determined by the simple modules over $\mathbb{C}[G_0,\omega]$, according to \cref{prop:equivalence of categories between right A-modules and A_0-modules}.

Let $V$ be a simple $\mathbb{C}[G_0,\omega]$ module. We construct the corresponding $\mathbb{Z}_2$-graded simple module 
\begin{equation}
    \mathbb{C}[G,\omega,s]\otimes_{\mathbb{C}[G_0,\omega]}V.
\end{equation}
To be explicit, let $\rho_0$ be the representation of $V$, and $\rho$ be the representation of $\mathbb{C}[G,\omega,s]\otimes_{\mathbb{C}[G_0,\omega]}V$. The representation $\rho$ has
\begin{widetext}
\begin{equation}\label{eq:representation of C[G omega s] otimes N complex number}
    \rho(z)=
    \begin{bmatrix}
        \rho_0(z) & 0 \\
        0 & \rho_0(\bar{z}) \\
    \end{bmatrix}, \textrm{ for } z\in \mathbb{C},
\end{equation}
\begin{equation}\label{eq:representation of C[G omega s] otimes N grading 0}
    \rho(g)=
    \begin{bmatrix}
        \rho_0(g) & 0 \\
        0 & \bar{\omega}(g_1,g_1^{-1})^{-1}\bar{\omega}(g,g_1)\omega(g_1^{-1},gg_1)\rho_0(g_1^{-1}gg_1) \\
    \end{bmatrix}, \textrm{ for } g\in G_0,
\end{equation}
\begin{equation}\label{eq:representation of C[G omega s] otimes N grading 1}
    \rho(g)=
    \begin{bmatrix}
        0 & \omega(g,g_1)\rho_0(gg_1) \\
        \bar{\omega}(g_1,g_1^{-1})^{-1}\omega(g_1^{-1},g)\rho_0(g_1^{-1}g)  & 0 \\
    \end{bmatrix}, \textrm{ for } g\in G_1,
\end{equation}
\end{widetext}
where $g_1\in G_1$ is the fixed grading 1 group element.

In principle, $V$ is only a real vector space; however, if we take $\rho_0(i)$ as a complex structure, then $V$ becomes a complex vector space. Similarly, if we take $\rho(i)$ as a complex structure on $V\oplus V$, then the induced complex vector space is isomorphic to $V\oplus \bar{V}$. Notice that \cref{eq:representation of C[G omega s] otimes N grading 0} is actually an antilinear map. If we choose a basis of the complex vector space $V$, then there is a canonical antilinear map on $V$ which is denoted by $K$. $\rho(g)$ can be written as matrices over $\mathbb{C}$ in this basis,
\begin{widetext}
\begin{equation}\label{eq:matrix of C[G omega s] otimes N grading 0}
    \rho(g)=
    \begin{bmatrix}
        \rho_0(g) & 0 \\
        0 & \omega(g_1,g_1^{-1})^{-1}\omega(g,g_1)\bar{\omega}(g_1^{-1},gg_1)\bar{\rho}_0(g_1^{-1}gg_1) \\
    \end{bmatrix}, \textrm{ for } g\in G_0,
\end{equation}
\begin{equation}\label{eq:matrix of C[G omega s] otimes N grading 1}
    \rho(g)=
    \begin{bmatrix}
        0 & \omega(g,g_1)\rho_0(gg_1)K \\
        \omega(g_1,g_1^{-1})^{-1}\bar{\omega}(g_1^{-1},g)\bar{\rho}_0(g_1^{-1}g)K  & 0 \\
    \end{bmatrix}, \textrm{ for } g\in G_1.
\end{equation}  
\end{widetext}

Next, we determine the $\mathbb{Z}_2$-graded division algebra $D$. $D_0$ is always $\mathbb{C}$, and it acts on $V\oplus V$ by
\begin{equation}\label{eq:representation of C[G omega s] otimes N action of D0}
    \begin{bmatrix}
        \rho_0(z) & 0 \\
        0 & \rho_0(z) \\
    \end{bmatrix}, \textrm{ for } z\in \mathbb{C}.
\end{equation}
Notice that this map commutes with all $\rho(g)$; however, \cref{eq:representation of C[G omega s] otimes N complex number} does not have this property.
If $V\not\cong V^{g_1}$, then we have $D=\mathbb{C}$, where $V^{g_1}$ is a $\mathbb{C}[G_0,\omega]$ module with
\begin{equation}
    g\cdot v=\bar{\omega}(g_1,g_1^{-1})^{-1}\bar{\omega}(g,g_1)\omega(g_1^{-1},gg_1) g_1^{-1}gg_1v,
\end{equation}
and $z\cdot v=\rho_0(\bar{z})(v)$ for $z\in \mathbb{C}$ and $v\in V$.

Assume that we have an isomorphism $f:V\cong V^{g_1}$, and $f$ is constructed using the methods in Appendix \ref{subsubsection:constructing isomorphisms between two modules twisted}. A nonzero element of $D_1$ is enough to generate the whole $D$, and we have one such element,
\begin{equation}\label{eq:grading 1 element of D twisted group algebra}
    \theta=
    \begin{bmatrix}
        0 & \omega(g_1,g_1)\rho_0(g_1^2)f \\
        f & 0
    \end{bmatrix}.
\end{equation}
To check whether $\theta$ commutes with all $\rho(g)$, we need to use $\omega(g,g^{-1})={}^{g}\omega(g^{-1},g)$. $\theta$ is an antilinear element with respect to $D_0=\mathbb{C}$; hence, $D$ is determined by the sign of $\theta^2$. If $\theta^2$ is positive, then $D=Cl^{0,2}$. If $\theta^2$ is negative, then $D=Cl^{2,0}$.

\subsubsection{\texorpdfstring{The inverse of $\Phi$}{The inverse of Phi}}\label{subsubsection:inverse of Phi twisted group algebra with antilinear elements}

For $f\in M_{d_{i}'}(D_i)(\bar{c}_{i})$, we have
\begin{widetext}
\begin{equation}\label{eq:inverse of Phi twisted group algebra with antilinear elements}
    \Phi^{-1}(f) = \frac{d_i}{2|G|}\sum_{g\in G}\omega(g^{-1},g)^{-1}\left( \mathrm{Tr}\left(\mathrm{Re}\left(\rho_i(g^{-1})f\right)\right) - \mathrm{Tr}\left(\mathrm{Re}\left(\rho_i(ig^{-1})f\right) \right) i \right)\, g ,
\end{equation}
\end{widetext}
where $\mathrm{Re}$ is the projection from $D$ to the subspace $\mathbb{R}$ spanned by the unit. This formula actually works for any grading group $\Gr$.

To prove this is the inverse, we assume $\Phi^{-1}(f) = \sum_{g\in G} a_g\, g \in \mathbb{C}[G,\omega,s]$. Notice that we have
\begin{widetext}
\begin{equation}
    \begin{aligned}
        a_g & = \frac{1}{2|G|}\left(\mathrm{Tr}_{\mathbb{R}}\left(\Phi^{-1}(f)(g)^{-1}\cdot -\right) - \mathrm{Tr}_{\mathbb{R}}\left(i\Phi^{-1}(f)(g)^{-1}\cdot -\right)i \right)\\
        & = \frac{d_i}{2|G|}\left( \mathrm{Tr}\left(\mathrm{Re}\left(\rho_i(g)^{-1}f\right)\right) - \mathrm{Tr}\left(\mathrm{Re}\left(i\rho_i(g)^{-1}f\right)\right)i \right) \\
        & = \frac{d_i}{2|G|}\omega(g^{-1},g)^{-1}\left( \mathrm{Tr}\left(\mathrm{Re}\left(\rho_i(g^{-1})f\right)\right) - \mathrm{Tr}\left(\mathrm{Re}\left(i\rho_i(g^{-1})f\right)\right)i \right).
    \end{aligned}
\end{equation}
\end{widetext}
The trace in the first line is calculated as a linear map on $\mathbb{C}[G,\omega,s]$. The second line uses the method in Appendix \ref{subsubsection:Inverse of Phi graded}. The third line uses the fact that $\mathrm{Tr}(\mathrm{Re}\left(-\right) - \mathrm{Tr}\left(\mathrm{Re}\left(i\cdot -\right)i\right))$ is $\mathbb{C}$ linear.

\subsubsection{Character theory}\label{subsubsection:character theory of twisted group algebra with antilinear elements}

We show how character theory determines the $\mathbb{Z}_2$-graded division algebra $D$. Let $V$ be a $\mathbb{C}[G_0,\omega]$-module, and let $\chi$ be its character. Since there are only three possibilities, $D=\mathbb{C}$, $Cl^{0,2}$, or $Cl^{2,0}$, $D$ is determined by the type of $\mathbb{C}[G,\omega,s]\otimes_{\mathbb{C}[G_0,\omega]}V$. If the type is $\mathbb{R}$, we have $D=Cl^{0,2}$. If the type is $\mathbb{C}$, we have $D=\mathbb{C}$.  If the type is $\mathbb{H}$, we have $D=Cl^{2,0}$.

Notice that $\mathbb{C}[\tilde{G},1,\tilde{s}]$ is isomorphic to the real twisted group algebra $\mathbb{R}[G',\omega']$ without antilinear elements, where $G'=\mathbb{Z}_2\times \tilde{G}$. We denote the elements of $G'$ by $g$ and $ig$ for $g\in \tilde{G}$. We have $\omega'(ig,ig')=-1$ for $g\in \tilde{G}_0$, $\omega'(g,ig')=-1$ for $g\in \tilde{G}_1$, and 1 otherwise.

We lift $\mathbb{C}[G,\omega,s]\otimes_{\mathbb{C}[G_0,\omega]}V$ to an $\mathbb{R}[G',\omega']$ module. Now we can use the Frobenius--Schur indicator \cref{eq:Frobenius--Schur indicator twisted} of twisted group algebras,
\begin{equation}\label{eq:intermediate W twisted}
    \begin{aligned}
        & \frac{1}{|G'|}\left( \sum_{g \in \tilde{G}} \omega'(g,g)\chi'(g^2)+\sum_{g \in \tilde{G}} \omega'(ig,ig)\chi'\left((ig)^2\right) \right) \\
        = & \frac{1}{|G'|}\left( \sum_{g \in \tilde{G}} \chi'(g^2)-\sum_{g \in \tilde{G}_0} \chi'\left(g^2\right)+\sum_{g \in \tilde{G}_1} \chi'\left(g^2\right) \right) \\
        = & \frac{1}{|\tilde{G}|} \sum_{g \in \tilde{G}_1} \mathrm{Tr}_{\mathbb{R}}\left(\tilde{\rho}(g^2)\right) =\frac{1}{|\tilde{G}|} \sum_{g \in \tilde{G}_1} 2\mathrm{Tr}\circ \mathrm{Re}\left(\tilde{\rho}(g^2)\right) \\
        = & \frac{1}{|\tilde{G}|} \sum_{g \in \tilde{G}_1} \mathrm{Tr}\left(\tilde{\rho}(g^2)+\bar{\tilde{\rho}}(g^2)\right) = \frac{2}{|\tilde{G}|} \sum_{g \in \tilde{G}_1} \mathrm{Tr}\left(\tilde{\rho}(g^2)\right) \\
        = & \frac{2}{|\tilde{G}|} \sum_a \sum_{g\in G_1} |a|^2 \omega(g,g)\left(\chi(g^2)+\chi^{g_1}(g^2)\right) \\
        = & \frac{4}{|\tilde{G}|} \sum_a \sum_{g\in G_1} |a|^2 \omega(g,g)\chi(g^2)\\
        = & \frac{2}{|G_0|} \sum_{g\in G_1} \omega(g,g)\chi(g^2) ,
    \end{aligned}
\end{equation}
where $\chi'$ is the character of the $\mathbb{R}[G',\omega']$ module, and $\tilde{\rho}$ is the representation of the $\mathbb{C}[\tilde{G},1,\tilde{s}]$ module. $\mathrm{Tr}_{\mathbb{R}}$ is the trace of real linear maps; if the vector space has a complex structure, then we have $\mathrm{Tr}_{\mathbb{R}}(f)=2\mathrm{Tr}\circ \mathrm{Re}(f)$ for a complex linear map $f$. $\chi^{g_1}$ is the character of $V^{g_1}$. After the third equality, we use the complex structure \cref{eq:representation of C[G omega s] otimes N complex number}; therefore, $\tilde{\rho}$ is given by \cref{eq:matrix of C[G omega s] otimes N grading 0,eq:matrix of C[G omega s] otimes N grading 1}.  To get the fifth equality, we use the fact that the Frobenius--Schur indicator of a module is invariant under complex conjugation. In the sixth equality, we use the fact that $V$ and $V^{g_1}$ have the same type; therefore, the Frobenius--Schur indicators are the same.

When $D=\mathbb{C}$, \cref{eq:intermediate W twisted} is 0. When $D=Cl^{0,2}$, \cref{eq:intermediate W twisted} is 2 because a simple $\mathbb{Z}_2$-graded module with type $Cl^{0,2}$ decomposes into two simple ungraded modules with type $\mathbb{R}$. When $D=Cl^{2,0}$, \cref{eq:intermediate W twisted} is $-2$ because the complexification of a simple ungraded module over $\mathbb{R}$ with type $\mathbb{H}$ is a direct sum of two isomorphic modules over $\mathbb{C}$ with type $\mathbb{H}$.

In summary, to determine $D$, we need to calculate
\begin{equation}\label{eq:W twisted}
    w=\frac{1}{|G_0|}\sum_{g\in G_1} \omega(g,g)\chi(g^2),
\end{equation}
where $\chi$ is the character of $V$. If $w=1$, then $D=Cl^{0,2}$. If $w=0$, then $D=\mathbb{C}$. If $w=-1$, then $D=Cl^{2,0}$.

%% file: a5.K-theory.tex
\section{\texorpdfstring{Karoubi's $K$-Theory}{Karoubi's K-Theory}}\label{section:K-theory}

We collect some results from Karoubi's works \cite{karoubiKtheorieEquivariante1970,karoubiAlgebresClifford$K$theorie1968,karoubiEquivariantKtheoryReal2007,karoubiKtheoryIntroduction2009}. Notice that Karoubi only proved the results for a group $G$ that allows antilinear symmetries. There are no results for symmetries involving projective representations or anticommutation with the Hamiltonian. In what follows, we assume that the fundamental theorem for the general case is correct, and leave the proof to the experts. It should be noted that the result for the general case has been used implicitly in physics-related works, for example, Refs. \cite{cornfeldTenfoldTopologyCrystals2021,shiozakiClassificationSurfaceStates2022}.

\subsection{The fundamental theorem}

Let $V$ be a real $G$-vector bundle over a $G$-space $X$ equipped with a $G$-invariant inner product. We can construct a fiber bundle $Cl(V)$ over $X$ whose fiber over $x\in X$ is the Clifford algebra $Cl(V_x)$. We denote the $G$ action on $V$ by $\rho_V$; $\rho_V(g)$ maps an element of $V_x$ to an element of $V_{g\cdot x}$. $\rho_V(g)$ also induces a $G$ action on $Cl(V)$; we still denote it by $\rho_V(g)$.

An $\symalg^V$ module $E$ is a real vector bundle equipped with actions of $\symalg$ and $Cl(V)$ that are compatible with each other. In particular, for $x\in X$, $g\in \symalg$, $a_x\in Cl(V)_x$, and $e_x\in E_x$, we have
\begin{equation}
    g\cdot_{\symalg}(a_x\cdot_{Cl(V)}e_x) = (-1)^{h(g)}\rho_V(g)(a_x)\cdot_{Cl(V)}(g\cdot_{\symalg}e_x)
\end{equation}
Notice that $g\cdot_{\symalg}e_x$ is a vector in $E_{g\cdot x}$. The vector bundle $E$ is also a complex vector bundle because $\mathbb{C}$ is a subalgebra of $\symalg$.

\begin{eg}
    For $X$ a point, $V$ is just a real $G$ module, and an $\symalg^V$ module is the module over the algebra $\symalg^V$ as defined in \cref{subsubsection:classifying space of dirac hamiltonians}.
\end{eg}
\begin{eg}\label{eg:symalg^V-module with V trivial}
    For $V$ a trivial bundle, by abuse of notation, we take the trivial bundle to be $X\times V$ where $V$ is a real $G$ representation. $Cl(X\times V)$ is $X\times Cl(V)$; here, $Cl(V)$ is a Clifford algebra. The action of $Cl(X\times V)$ on a vector bundle $E$ reduces to the action of $Cl(V)$, since the action is given by a bundle map $Cl(X\times V)\times_{X}E\to E$, and the domain is just $Cl(V)\times E$.

    In summary, an $\symalg^V$ module is a vector bundle over $X$ with the action of $\symalg^V$.
\end{eg}

Karoubi defined a group $\ktheory_{\symalg}^V(X,Y)$ for a real $G$-vector bundle $V$ over a $G$-space $X$; $Y$ is a (nice) subspace of $X$. An element of $\ktheory_{\symalg}^V(X,Y)$ is represented by $(E,\varepsilon_1,\varepsilon_2)$, where $E$ is a $\symalg^V$ module, and $\varepsilon_1$, $\varepsilon_2$ are two gradations that are equal on $Y$. A gradation $\varepsilon$ of $E$ is a $\symalg^V$-module map such that $\varepsilon^2=\mathrm{id}$. Here, we define the $\symalg^V$ module maps to anticommute with the action of the grading-1 elements of $\symalg^V$, i.e., $g$ with $h(g)=1$ and $v\in V\subset Cl(V)$. $(E,\varepsilon_1,\varepsilon_2)=(E,\varepsilon'_1,\varepsilon'_2)$ if and only if there exist two trivial elements $(F,\eta_1,\eta_2)$ and $(F',\eta_1',\eta_2')$ with $\eta_1$ (or $\eta_1'$) homotopic to $\eta_2$ (or $\eta_2'$) such that
\begin{equation}
    (E\oplus F,\varepsilon_1\oplus \eta_1,\varepsilon_2\oplus \eta_2)\cong(E'\oplus F',\varepsilon'_1\oplus \eta'_1,\varepsilon'_2\oplus \eta'_2).
\end{equation}
We denote $\ktheory_{\symalg}^V(X)$ by $\ktheory_{\symalg}^V(X,\emptyset)$, i.e., $Y$ is the empty set.

We construct a map from $\ktheory_{\symalg}^V(X)$ to $\ktheory_{\symalg}^0(S_+(V\oplus 1),S(V))$, where $S_+(V\oplus 1)$ is the upper hemisphere of the sphere bundle of $V\oplus 1$, and $S(V)$ is the sphere bundle of $V$. The superscript $0$ of $\ktheory_{\symalg}^0$ denotes the zero-dimensional vector bundle over $S_+(V\oplus 1)$. The map is given by
\begin{equation}
    \begin{split}
        t:(E,\varepsilon_1,\varepsilon_2) \mapsto (\pi^*E,&\ \varepsilon_1'\mathrm{cos}(\theta)+v\mathrm{sin}(\theta),\\
        &\ \varepsilon_2'\mathrm{cos}(\theta)+v\mathrm{sin}(\theta)),
    \end{split}
\end{equation}
where $\pi^* E$ is the pullback bundle of $E$ through $\pi:S_+(V\oplus 1)\to X$, and it also induces the gradation $\varepsilon_1'$ and $\varepsilon_2'$ on $\pi^* E$. $(v,\theta)$ is the polar coordinate of $S_+(V\oplus 1)$, and $v\in V \subset Cl(V)$.
Then Karoubi proved the fundamental theorem:
\begin{thm}
    $t:\ktheory_{\symalg}^V(X) \to \ktheory_{\symalg}^0(S_+(V\oplus 1),S(V))$ is an isomorphism.
\end{thm}

\begin{eg}
    For $X$ a point and $(E,\varepsilon_1,\varepsilon_2)\in \ktheory_{\symalg}^V(X)$, $E$ is a $\symalg^V$ module, and $\varepsilon_i$ lifts $E$ to a $Cl^{0,1}\hat{\otimes}\symalg^V$ module. An element $(F,\eta_1,\eta_2)$ of the relative K-group $\ktheory_{\symalg}^0(S_+(V\oplus 1),S(V))$ is a real vector bundle $F$ on $S_+(V\oplus 1)$ equipped with the action of $\symalg$ and gradations $\eta_1$, $\eta_2$ that are equal on $S(V)$.

    The image under $t$ of a gradation is actually the upper-hemisphere part of the Dirac Hamiltonian \cref{eq:Dirac Hamiltonian}.
\end{eg}
\begin{eg}\label{eg:K-theory with V trivial}
    We take the same notation as \cref{eg:symalg^V-module with V trivial}. For a trivial bundle $X\times V$, by abuse of notation, we denote $\ktheory_{\symalg}^{X\times V}(X)$ by $\ktheory_{\symalg}^V(X)$. The codomain of $t$ becomes $\ktheory_{\symalg}^0(X\times S_+(V\oplus 1),X\times S(V))$.
\end{eg}

\subsection{\texorpdfstring{The classifying spaces for $K$-theory}{The classifying spaces for K-theory}}\label{subsection:Classifying Spaces of K-Theory}

In the following sections, we discuss the cases in \cref{eg:symalg^V-module with V trivial,eg:K-theory with V trivial}. We have the following theorem:
\begin{thm}\label{thm:classifying spaces of K-theory}
    The following diagram is correct:
    \begin{widetext}
    \begin{equation}
        \begin{tikzcd}
            \ktheory_{\symalg}^V(X)
            \arrow[r,"\cong"]
            \arrow[d,"\cong"{right},"t"{left}]
                & \pi_0\left(\mathrm{Top}\left(X,\clspace(\symalg^V)\right)^G\right)
                \arrow[d] \\
            \ktheory_{\symalg}^0(X\times S_+(V\oplus 1),X\times S(V))
            \arrow[r,"\cong"]
                & \pi_0\left(\mathrm{Top}\left(X,\Omega^V(\clspace(\symalg))\right)^G\right)
        \end{tikzcd}
    \end{equation}    
    \end{widetext}
    The right map is induced by the map of ``construction of Dirac Hamiltonians'' as defined in \cref{subsubsection:comparison of two classifying spaces}.
\end{thm}
\begin{proof}
    For the upper arrow, given $(E,\varepsilon_1,\varepsilon_2)$ on the left, we can take its direct sum with a trivial element $(F,\eta_1,\eta_2)$ such that $E\oplus F$ is a trivial vector bundle $X\times W$, and we embed the $Cl^{0,1}\hat{\otimes} \symalg^V$ module $W$ with $\varepsilon_1\oplus \eta_1$ into $\mathcal{U}(Cl^{0,1}\hat{\otimes} \symalg^V)$; then $\varepsilon_2\oplus \eta_2$ would be the image. Conversely, an element on the right is a grading on $X\times W$, where $W$ is a $Cl^{0,1}\hat{\otimes} \symalg^V$ module, and the image is $(X\times W,M_0,M)$.

    The lower arrow is similar, but with $X\times S(V)$ mapped to the basepoint. We have $X\times S_+(V\oplus 1)/X\times S(V)$ $G$-homotopic to $S^V\wedge X_+$, and we use the tensor-hom adjunction of $\mathrm{Top}_*$.

    The commutativity of the diagram is easy to check.
\end{proof}

\begin{rk}
    In general, we can prove
    \begin{equation}
        \ktheory_{\symalg}^V(X,Y)\cong \pi_0\left(\mathrm{Top}_*\left(X/Y,\clspace(\symalg^V)\right)^G\right).
    \end{equation}
    This result also shows
    \begin{equation}
        \ktheory_{\symalg}^V(X,Y)\cong \ktheory_{\symalg}^V(X/Y,\{x_0\}),
    \end{equation}
    where $x_0$ is the basepoint of $X/Y$, i.e., the point that represents the points in $Y$.
\end{rk}

\begin{cor}
    $\clspace(\symalg^V)\simeq \Omega^V (\clspace(\symalg))$ is a $G$-equivariant homotopy equivalence.
\end{cor}
Along with \cref{subsubsection:atomic insulator and Bott periodicity}, this shows $\ktheory_{\symalg}^{-V}$ is an $RO(G)$-graded cohomology theory, and $\clspace(\symalg^{-V})$ is the genuine $G$-$\Omega$-spectrum, or equivalently the $RO(G)$-graded spectrum; see \cite{mayEquivariantHomotopyCohomology1996}. There is also a helpful lecture note \cite{debrayAdebrayEquivariant_homotopy_theory2025} on this topic.

\subsection{\texorpdfstring{$\tilde{\ktheory}_{\symalg}(X)$}{KAsym(X)}}\label{subsection:reduced K-theory}

For a $G$-space $X$ with a basepoint $x_0$, we define $\tilde{\ktheory}_{\symalg}(X)$ to be $\ktheory_{\symalg}^0(X,x_0)$, which is also defined in \cref{subsubsection:bundle-theoretic classification} in another way. The two definitions are equivalent, since both are isomorphic to $\pi_0\left(\mathrm{Top}_*\left(X,\clspace(\symalg^V)\right)^G\right)$.

%% file: a6.gap_example.tex
\section{An Example Calculated Using the GAP Package}\label{section:gap example}

The Dirac Hamiltonian for the $D_{3d}$ example in \cref{subsection:use GAP to compute more complicated examples}, obtained via \texttt{Hamiltonians(A)[1]}, is as follows:
\begin{widetext}
\begin{equation}
    \tilde{\Gamma}_1 =
        \begin{bmatrix}
            0 & 1 & 0 & 0 & 0 & 0 & 0 & 0 \\
            1 & 0 & 0 & 0 & 0 & 0 & 0 & 0 \\
            0 & 0 & 0 & -1 & 0 & 0 & 0 & 0 \\
            0 & 0 & -1 & 0 & 0 & 0 & 0 & 0 \\
            0 & 0 & 0 & 0 & 0 & 1 & 0 & 0 \\
            0 & 0 & 0 & 0 & 1 & 0 & 0 & 0 \\
            0 & 0 & 0 & 0 & 0 & 0 & 0 & -1 \\
            0 & 0 & 0 & 0 & 0 & 0 & -1 & 0
        \end{bmatrix}
\end{equation}
\begin{equation}
    \tilde{\Gamma}_2 =
        \begin{bmatrix}
            0 & 0 & 0 & 1 & 0 & 0 & 0 & 0 \\
            0 & 0 & 1 & 0 & 0 & 0 & 0 & 0 \\
            0 & 1 & 0 & 0 & 0 & 0 & 0 & 0 \\
            1 & 0 & 0 & 0 & 0 & 0 & 0 & 0 \\
            0 & 0 & 0 & 0 & 0 & 0 & 0 & 1 \\
            0 & 0 & 0 & 0 & 0 & 0 & 1 & 0 \\
            0 & 0 & 0 & 0 & 0 & 1 & 0 & 0 \\
            0 & 0 & 0 & 0 & 1 & 0 & 0 & 0
        \end{bmatrix}
\end{equation}
\begin{equation}
    \tilde{\Gamma}_3 =
        \begin{bmatrix}
            0 & 0 & 0 & -E(4) & 0 & 0 & 0 & 0 \\
            0 & 0 & -E(4) & 0 & 0 & 0 & 0 & 0 \\
            0 & E(4) & 0 & 0 & 0 & 0 & 0 & 0 \\
            E(4) & 0 & 0 & 0 & 0 & 0 & 0 & 0 \\
            0 & 0 & 0 & 0 & 0 & 0 & 0 & E(4) \\
            0 & 0 & 0 & 0 & 0 & 0 & E(4) & 0 \\
            0 & 0 & 0 & 0 & 0 & -E(4) & 0 & 0 \\
            0 & 0 & 0 & 0 & -E(4) & 0 & 0 & 0
        \end{bmatrix}
\end{equation}
In what follows, $a = \frac{1}{2}E(3)-\frac{1}{2}E(3)^2$, $b = \frac{1}{2}E(4)$, and $K$ is the complex conjugation operator.
\begin{equation}
    \hat{s}_6 =
        \begin{bmatrix}
            0 & 0 & 0 & 0 & a & 0 & -b & 0 \\
            0 & 0 & 0 & 0 & 0 & a & 0 & -b \\
            0 & 0 & 0 & 0 & b & 0 & a & 0 \\
            0 & 0 & 0 & 0 & 0 & b & 0 & a \\
            a & 0 & -b & 0 & 0 & 0 & 0 & 0 \\
            0 & a & 0 & -b & 0 & 0 & 0 & 0 \\
            b & 0 & a & 0 & 0 & 0 & 0 & 0 \\
            0 & b & 0 & a & 0 & 0 & 0 & 0
        \end{bmatrix}
\end{equation}
\begin{equation}
    \hat{m}_y =
        \begin{bmatrix}
            0 & 0 & 0 & 0 & -1 & 0 & 0 & 0 \\
            0 & 0 & 0 & 0 & 0 & -1 & 0 & 0 \\
            0 & 0 & 0 & 0 & 0 & 0 & 1 & 0 \\
            0 & 0 & 0 & 0 & 0 & 0 & 0 & 1 \\
            1 & 0 & 0 & 0 & 0 & 0 & 0 & 0 \\
            0 & 1 & 0 & 0 & 0 & 0 & 0 & 0 \\
            0 & 0 & -1 & 0 & 0 & 0 & 0 & 0 \\
            0 & 0 & 0 & -1 & 0 & 0 & 0 & 0
        \end{bmatrix}
\end{equation}
\begin{equation}
    \hat{C} =
        \begin{bmatrix}
            0 & 0 & 0 & 0 & 0 & E(4) & 0 & 0 \\
            0 & 0 & 0 & 0 & E(4) & 0 & 0 & 0 \\
            0 & 0 & 0 & 0 & 0 & 0 & 0 & E(4) \\
            0 & 0 & 0 & 0 & 0 & 0 & E(4) & 0 \\
            0 & -E(4) & 0 & 0 & 0 & 0 & 0 & 0 \\
            -E(4) & 0 & 0 & 0 & 0 & 0 & 0 & 0 \\
            0 & 0 & 0 & -E(4) & 0 & 0 & 0 & 0 \\
            0 & 0 & -E(4) & 0 & 0 & 0 & 0 & 0
        \end{bmatrix}K
\end{equation}
\begin{equation}
    \hat{T} =
        \begin{bmatrix}
            0 & 0 & -1 & 0 & 0 & 0 & 0 & 0 \\
            0 & 0 & 0 & -1 & 0 & 0 & 0 & 0 \\
            1 & 0 & 0 & 0 & 0 & 0 & 0 & 0 \\
            0 & 1 & 0 & 0 & 0 & 0 & 0 & 0 \\
            0 & 0 & 0 & 0 & 0 & 0 & 1 & 0 \\
            0 & 0 & 0 & 0 & 0 & 0 & 0 & 1 \\
            0 & 0 & 0 & 0 & -1 & 0 & 0 & 0 \\
            0 & 0 & 0 & 0 & 0 & -1 & 0 & 0
        \end{bmatrix}K
\end{equation}
\begin{equation}
    M = -M_0 =
        \begin{bmatrix}
            1 & 0 & 0 & 0 & 0 & 0 & 0 & 0 \\
            0 & -1 & 0 & 0 & 0 & 0 & 0 & 0 \\
            0 & 0 & 1 & 0 & 0 & 0 & 0 & 0 \\
            0 & 0 & 0 & -1 & 0 & 0 & 0 & 0 \\
            0 & 0 & 0 & 0 & 1 & 0 & 0 & 0 \\
            0 & 0 & 0 & 0 & 0 & -1 & 0 & 0 \\
            0 & 0 & 0 & 0 & 0 & 0 & 1 & 0 \\
            0 & 0 & 0 & 0 & 0 & 0 & 0 & -1
        \end{bmatrix}
\end{equation}
\end{widetext}